\RequirePackage{silence}
\documentclass[11pt]{article}
\pdfoutput=1
\usepackage[ascii]{inputenc}
\usepackage[bookmarksnumbered,hypertexnames=false,colorlinks=true,linkcolor=blue,urlcolor=blue,citecolor=blue,anchorcolor=green,pdfusetitle]{hyperref}
\usepackage{bbm,braket,microtype,mathrsfs,amsmath,amssymb,amsthm,mathtools,fullpage,
graphicx,enumitem,ae,aecompl,float,mathabx,caption,bm,mathdots,thmtools,thm-restate,chngcntr}
\usepackage{abstract}
\usepackage{comment}
\usepackage[capitalize]{cleveref}
\usepackage[dvipsnames,table]{xcolor}
\usepackage{tikz-cd}
\usepackage{pagecolor}
\usepackage{tkz-graph}
\usepackage{tkz-berge}
\usepackage[numbers,sort&compress]{natbib}
\usepackage[T1]{fontenc}
\usepackage{newpxtext}
\usepackage{complexity}
\usepackage{titlesec}
\usepackage{longtable}
\titlespacing*{\paragraph}{0pt}{2ex plus 1ex minus .2ex}{1em}
\newtheorem{thm}{Theorem}[section]\crefname{thm}{Theorem}{Theorems}
\newtheorem{lem}[thm]{Lemma}\crefname{lem}{Lemma}{Lemmas}
\newtheorem{prb}[thm]{Problem}\crefname{prb}{Problem}{Problems}
\newtheorem{rem}[thm]{Remark}\crefname{rem}{Remark}{Remarks}
\newtheorem{cor}[thm]{Corollary}\crefname{cor}{Corollary}{Corollaries}
\newtheorem{dfn}[thm]{Definition}\crefname{dfn}{Definition}{Definitions}
\newtheorem{prp}[thm]{Proposition}\crefname{prp}{Proposition}{Propositions}
\newtheorem{exa}[thm]{Example}\crefname{exa}{Example}{Examples}
\floatstyle{boxed}\newfloat{Algorithm}{ht}{alg}\crefname{Algorithm}{Algorithm}{Algorithms}
\counterwithin{figure}{section}
\counterwithin{Algorithm}{section}
\numberwithin{equation}{section}
\allowdisplaybreaks[4]

\DeclareMathOperator*{\argmin}{arg\,min}
\DeclareMathOperator{\tr}{tr}

\DeclareMathOperator{\Mat}{Mat}
\DeclareMathOperator{\diag}{diag}
\DeclareMathOperator{\PD}{PD}
\renewcommand{\SL}{\operatorname{SL}}
\DeclareMathOperator{\SU}{SU}
\DeclareMathOperator{\GL}{GL}
\DeclareMathOperator{\End}{End}
\DeclareMathOperator{\T}{T}
\DeclareMathOperator{\GT}{GT}
\DeclareMathOperator{\ST}{ST}
\DeclareMathOperator{\Un}{U}

\DeclareMathOperator{\aff}{Aff}

\DeclareMathOperator{\Herm}{Herm}

\DeclareMathOperator{\spec}{spec}

\renewcommand{\poly}{\operatorname{poly}}

\DeclareMathOperator{\conv}{conv}
\DeclareMathOperator{\dist}{dist}

\DeclareMathOperator{\capa}{cap}

\DeclareMathOperator{\supp}{supp}
\DeclareMathOperator{\spn}{span}
\DeclareMathOperator{\Lie}{Lie}

\newcommand{\CC}{\mathbb C}
\newcommand{\RR}{\mathbb R}
\renewcommand{\PP}{\mathbb P}
\newcommand{\ZZ}{\mathbb Z}
\newcommand{\NN}{\mathbb N}

\renewcommand{\EE}{\mathbb E}

\newcommand{\ot}{\otimes}

\newcommand{\eps}{\varepsilon}
\newcommand{\id}{\mathbbm 1}

\newcommand{\norm}[1]{\lVert#1\rVert}

\title{Barriers for recent methods in geodesic optimization}

\author{Cole Franks\thanks{Department of Mathematics, Massachusetts Institute of Technology, \texttt{franks@mit.edu}} \hspace{2pt} and Philipp Reichenbach\thanks{Institut f\"ur Mathematik, Technische Universit\"at Berlin, \texttt{reichenbach@tu-berlin.de}. Supported by the European Research Council (ERC) under the European's Horizon 2020 research and innovation programme (grant agreement no. 787840).}}
\date{}

\begin{document}

\maketitle

\begin{abstract}
We study a class of optimization problems including matrix scaling, matrix balancing, multidimensional array scaling, operator scaling, and tensor scaling that arise frequently in theory and in practice. Some of these problems, such as matrix and array scaling, are convex in the Euclidean sense, but others such as operator scaling and tensor scaling are \emph{geodesically convex} on a different Riemannian manifold. Trust region methods, which include box-constrained Newton's method, are known to produce high precision solutions very quickly for matrix scaling and matrix balancing (Cohen et. al., FOCS 2017, Allen-Zhu et. al. FOCS 2017), and result in polynomial time algorithms for some geodesically convex problems like operator scaling (Garg et. al. STOC 2018, B\"urgisser et. al. FOCS 2019). One is led to ask whether these guarantees also hold for multidimensional array scaling and tensor scaling.

We show that this is not the case by exhibiting instances with exponential \emph{diameter bound}: we construct polynomial-size instances of 3-dimensional array scaling and 3-tensor scaling whose approximate solutions all have doubly exponential condition number. Moreover, we study convex-geometric notions of complexity known as margin and gap, which are used to bound the running times of all existing optimization algorithms for such problems. We show that margin and gap are exponentially small for several problems including array scaling, tensor scaling and polynomial scaling. Our results suggest that it is impossible to prove polynomial running time bounds for tensor scaling based on diameter bounds alone. Therefore, our work motivates the search for analogues of more sophisticated algorithms, such as interior point methods, for geodesically convex optimization that do not rely on polynomial diameter bounds.
\end{abstract}

\thispagestyle{empty}

\newpage

\thispagestyle{empty}

\setcounter{tocdepth}{3}

\tableofcontents


\newpage
\setcounter{page}{1}

\section{Introduction}
We study a class of optimization problems ubiquitous in theoretical computer science, machine learning,  quantum information theory and statistics. The programs we consider are continuous optimization problems over matrix groups. More precisely, they can be posed as Euclidean norm minimization over the closure of a group orbit. The programs span two historically distinct contexts: In one context, the optimization problems are convex, and in the other they are not convex but rather \emph{geodesically convex} on a suitable manifold.

The \emph{commutative setting}, in which the underlying group is Abelian, captures matrix scaling, matrix balancing and array scaling, which arise in scientific computing and optimal transport \cite{cuturi2013sinkhorn, parlett1971balancing}. Such problems fall into the framework of unconstrained geometric programming. Though these problems are convex, there are at least two reasons to study them further. Firstly, they are of such practical importance that speed matters. Na\"ively applying powerful algorithms like ellipsoid and interior point methods can be impractically slow. Hence, it is important to understand when faster methods can succeed. Matrix scaling and balancing, in particular, have enjoyed some success stories~- there are fast algorithms to obtain high precision solutions \cite{cohen2017matrix,allen2017much}, and there are more general upper bounds \cite{burgisser2020interior}. Secondly, the algorithms developed for the commutative setting are candidates for generalization to our second setting, which takes place in the less well-understood arena of geodesically convex optimization.

The second context, which we call the \emph{noncommutative setting}, arises when the underlying group is non-Abelian. The noncommutative setting captures problems like operator and tensor scaling \cite{garg2016deterministic,burgisser2017alternating},  the quantum marginal problem \cite{burgisser2018efficient} and statistical estimators such as Tyler's M estimator \cite{franks2020rigorous} and maximum likelihood estimates for matrix and tensor normal models \cite{amendola2020invariant}. Deciding whether the value of the optimization problem is zero or not is equivalent to deciding a central polynomial identity testing (P.I.T.) problem in invariant theory known as the \emph{null cone} problem. It is hoped that efficient optimization algorithms will result in efficient algorithms for the null-cone problem.
One approach to complexity lower bounds, geometric complexity theory, suggests that these P.I.T. problems should be in $\P$ \cite{mulmuley2017geometric, garg2019search}, and the optimization approach has resulted in polynomial time algorithms in some cases \cite{garg2016deterministic, allen2018operator}. The optimization problems that arise in the noncommutative setting are not convex in the Euclidean sense, but rather \emph{geodesically convex}, a notion of convexity on a Riemannian manifold. Currently, the only implementable algorithms for geodesically convex optimization are analogues of gradient descent and trust region methods \cite{absil2008optimization, zhang2016first,allen2018operator}. There are, as of yet, no efficiently implementable geodesically convex counterparts to the interior point or cutting plane methods.

In both the commutative and noncommutative settings, algorithms are typically analysed using two quantities. One is \emph{diameter}, or how far approximate minimizers can be from the origin. The other is a geometric measure of well-conditionedness known as \emph{margin} (or \emph{gap} in the noncommutative case), which has several variants in the literature and appears in two primary ways. Firstly, the smaller the margin, the higher the degree of precision required to decide if the value of the optimization problem is zero or not \cite{gradflow, gurvits2004combinatorial}. Secondly, the larger the margin, the smaller the diameter \cite{singh2014entropy,straszak2017computing,gradflow,burgisser2020interior}. In this paper we show the following:

\begin{enumerate} 
\item[i)] In the commutative setting, and in particular for array scaling, approximate minimizers for the functions we study can have doubly exponential condition number. That is, the problems have exponential diameter. As a consequence, popular classes of algorithms such as gradient descent and trust region methods cannot produce high-precision solutions in polynomial time in general. This result applies in the noncommutative setting as well, which provides evidence that even cutting plane methods are unlikely to produce high-precision solutions in polynomial time. This shows it is necessary to develop powerful methods like the interior point method in the geodesically convex setting. 
\item[ii)] In the commutative and noncommutative settings, we study the margin and gap, respectively, which appear in running time bounds for all existing algorithms. We prove that these measures can be exponentially small in the input size for several problems including array scaling and tensor scaling. In the commutative case, this gives evidence that existing algorithms for array scaling do not run in near-linear time. In the noncommutative case, our results show that margin-based analyses like \cite{gradflow} cannot prove polynomial time guarantees for deciding the null cone problem for tensor scaling using trust region methods. 
\end{enumerate}

We use the remainder of the introduction to describe both settings in more detail, state our main results precisely, and discuss previous work. For both the commutative and noncommutative settings, we proceed in the following order. We start with an introduction and motivation of the setting, continue with diameter bounds and afterwards treat bounds on the margin and gap, respectively. We end each setting with a short discussion of the main proof techniques.

\subsection{The commutative setting: matrix scaling and its relatives}

\paragraph{Matrix scaling and array scaling.}

Consider the matrix scaling problem: given a nonnegative matrix $A$, find nonnegative diagonal matrices $X, Y$ such that $XA Y$ is doubly stochastic (i.e. has row and column sums equal to one). The matrices, if they exist, can be found by the exceedingly simple and fast alternating minimization method known as \emph{Sinkhorn's algorithm}. It is frequently used in practice, e.g. for quickly approximating the solution to optimal transport problems \cite{cuturi2013sinkhorn}.

Like all other algorithms for matrix scaling,  Sinkhorn's algorithm is typically analyzed through optimization. One finds that $X$ and $Y$ are $e^{\diag(x)}, e^{\diag(y)}$, where $x,y \in \RR^n$ are solutions to the following optimization problem: 
\begin{align}\inf_{x, y \in \RR^n} \sum A_{ij}e^{x_i + y_j - \bar{x} - \bar{y}}\label{eq:matrix-capacity}\end{align}
for $\bar{z}:=\frac{1}{n} \sum z_i$ (c.f. \cite{kalantari1996complexity}). Moreover, the infimum is greater than zero if and only if $A$ is \emph{approximately scalable}, i.e. the row and column sums of $XAY$ can be made arbitrarily close to one for $X, Y$ nonnegative, diagonal.

More generally, given a finite set $\Omega \subseteq \RR^m$ and a nonnegative function $p:\Omega \to \RR_{\geq 0}$, define the \emph{capacity} \cite{gurvits2004combinatorial} as the value of the unconstrained geometric program
\begin{gather}
\capa(p):= \inf_{x \in \RR^m} f_p(x) := \inf_{x \in \RR^m} \sum_{\omega \in \Omega} p_\omega e^{\omega \cdot x}.\label{eq:abelian}
\end{gather}
The capacity is positive if and only if zero is in the \emph{Newton polytope} $\conv(\supp p)$. Matrix scaling arises when $m = 2n$ and $\Omega = \{(\eps_i, \eps_j): i,j \in [n]\}$ for $\eps_k:=e_k - \frac{1}{n} \id_n$, where $e_k \in \RR^n$ is the $k^{th}$ canonical unit vector and $\id_n \in \RR^n$ denotes the all-ones vector. In this case \cref{eq:abelian} reduces to precisely \cref{eq:matrix-capacity}, and $\|\nabla \log f_p (x)\|$ measures the deviation of $p$ from doubly stochastic.

Matrix balancing, in which we instead wish to find a scaling for which the $i^{th}$ row and column sum match, arises when $m = n$ and $\Omega = \{e_i - e_j: i\neq j \in [n]\}$. When $m = 3n$ and $\Omega = \{(\eps_i, \eps_j, \eps_k): i,j,k \in [n]\}$ we obtain the $3$-dimensional \emph{array scaling problem}. In analogy to matrix scaling, in array scaling one has an array $p$ of numbers in $(\RR^n_{\geq 0})^{\ot 3}$ and seeks positive vectors $X,Y,Z \in \RR^n_{\geq 0}$ so that the array $q$ with entries $q_{ijk} =p_{ijk}X_i Y_j Z_k$ is \emph{tristochastic}. That is, the sum over every \emph{slice} is equal to one, i.e.
    $\sum_{j,k} q_{i_0,j,k} = \sum_{i,k} q_{i,j_0,k} = \sum_{i,j} q_{i,j,k_0} = 1$
for all $i_0,j_0,k_0 \in [n]$. If it is possible to satisfy these equations to arbitrary precision we say $p$ is approximately scalable. As for matrix scaling, $p$ is approximately scalable if and only if $\capa(p) > 0$. In the same manner, we obtain $d$-dimensional array scaling for $m=dn$ and
	\begin{equation}\label{eq:defnOmega-n-d}
	\Omega = \Omega_{n,d} := \big\lbrace \eps_i \colon i \in [n] \big\rbrace^d \subseteq \big( \RR^n \big)^d.
	\end{equation}
We can think of subsets of $\Omega_{n,d}$ as $d$-uniform, $d$-partite hypergraphs. Up to an additive shift by $- \frac{1}{n} \id_{nd}$, the elements of $\Omega_{n,d}$ are indicator vectors of the edges in such hypergraphs. For $d = 2$, the matrix $p$ is scalable if and only if the bipartite graph corresponding to $\supp p$ contains a perfect matching, but this is not the case for $d \geq 3$ (indeed, $d$-partite hypergraph matching is $\NP$-hard).

\paragraph{Algorithms for array scaling.}

Array scaling serves the same role for speeding up multimarginal transport as matrix scaling for optimal transport, and yet again there is a simple and fast alternating minimization algorithm that produces $\eps$-tristochastic scalings in time $O(1/\eps^2)$ \cite{altschuler2020polynomial,lin2019complexity}.
Moreover, algorithms to approximate the capacity arise in varied settings including radial isotropic position \cite{hardt2013algorithms}, entropy maximization \cite{straszak2017computing}, and approximate counting \cite{anari2018log}.

It is natural to ask if there are \emph{high-precision} algorithms for array scaling with $\log(1/\eps)$ dependence on the error and linear or mild dependence on the number of nonzero entries. For matrix scaling and matrix balancing, several works have shown that trust regions and interior point methods \emph{can} obtain such guarantees \cite{cohen2017matrix,allen2017much}.
Our work is concerned with whether the performance of such algorithms carries over to array scaling and the computation of the capacity in general.

\subsubsection{Diameter lower bounds}

Guarantees for many iterative algorithms in convex optimization require \emph{diameter bounds}, or bounds on the distance $R$ from the starting point to an $\eps$-approximate solution. Trust region methods, also called \emph{box-constrained Newton's method}, are iterative algorithms that, at each step, move to the best solution within a typically small distance $D$ of the previous solution. By their nature, trust region methods take at least $R/D$ steps to produce an $\eps$-approximate solution. Gradient descent for Lipschitz functions also depends quadratically on a diameter bound, and cutting plane methods typically use diameter bounds to control the volume of a starting region.

\paragraph{Known diameter upper and lower bounds.}
For matrix scaling and matrix balancing, it has been shown in \cite{cohen2017matrix} that one may take $R = O(n \log(w_A/\eps))$, where $w_A$ is the ratio between the sum of the entries of the matrix and the least nonzero entry. For $3$-dimensional array scaling, the best upper bound of which we are aware is $R = O(n^{3/2} 2^{6n} \log (1/\eps)),$ which follows from the general upper bound of \cite{straszak2017computing} on diameter bounds for unconstrained geometric programming. There is also a diameter bound for array scaling in the multimarginal transport context that is polynomial in the input size assuming the tensor has no nonzero entries \cite{lin2019complexity}.

Regarding diameter \emph{lower} bounds, in the context of computing maximum entropy distributions it was shown that there is some bounded set $\Omega \subset \ZZ^m$ in a $\poly(m)$ size ball such that there are \emph{no} $\eps$-approximate minimizers of norm $\poly(m, \log 1/\eps)$ for $f_p$ as in \cref{eq:abelian}  \cite{straszak2017computing}. 

\paragraph{Main theorem.}
Where do the polynomial diameter bounds for matrix scaling (i.e. $2$-dimensional array scaling) transition to the superpolynomial diameter bounds for general $\Omega$? We show that this transition takes place in the next simplest problem, the $3$-dimensional array scaling problem.

\begin{thm}\label{thm:diameter}
There is an absolute constant $C > 0$ and an array $p_{ijk} \in (\RR_{\geq 0}^n)^{\ot 3}$ with $O(n)$ nonzero entries, each of bit-complexity $O(n)$, that satisfies the following property. For all $0 <\eps \leq  \exp(- C n^2 \log n)$ and $(x,y,z) \in \RR^{3n}$, if
	$$f_p (x,y,z) \leq \capa(p) + \eps $$
	then $\norm{(x,y,z)}_2 = \Omega\left(2^{n/3}\log(1/\eps)\right).$
\end{thm}

To emphasize that the difficulties do not lie in an additive vs multiplicative approximation, we remark that our array $p$ has unit sum and $\capa(p) = 1/2$. By a simple duplication trick, the same bound holds for $d$-dimensional array scaling with $d \geq 3$; see \cref{cor:diameter-d}.

\paragraph{Implications of \cref{thm:diameter} and relation to the literature.}

Theorem~\ref{thm:diameter} shows that trust region methods for array scaling with polynomial step size cannot provide high-precision solutions in $\poly(n,\log(1/\eps))$ time for
$d \geq 3$. Moreover, gradient descent on the Lipschitz convex function $\log f_p$ has a bounded step size, and so also cannot provide high precision solutions in polynomial time.

In \cite[Section 2.1]{straszak2017computing} the authors ask whether there is $\Omega$ whose elements are Boolean (up to an additive shift) with a superpolynomial diameter lower bound. As subsets of $\Omega_{n,d}$ are automatically of this form, we answer their open problem in the affirmative. Our lower bound on $\log R$ is tight up to constant factors by the diameter upper bound from \cite{straszak2017computing} mentioned above; moreover the logarithmic dependence on $\eps$ is best possible. Determining the correct constant in the exponent is an interesting open direction. We believe that that the requirement that $\eps$ is very small is an artifact of our specific construction and proof strategy, and thus can probably be relaxed significantly.

Lastly, we remark that \cite{burgisser2020interior} bounds the diameter for $f_p$ by a polynomial in the \emph{facet gap}, i.e. the minimum distance between an element of $\supp p$ and an affine hull of a facet of the Newton polytope. The construction in \cref{thm:diameter} has exponentially small facet gap; see \cref{cor:facet-fap}.

\subsubsection{Margins: the geometry of scaling problems}
Many computational aspects of the capacity rely on the convex geometry of the finite set $\Omega \subseteq \RR^m$. Consider the following quantity, which we call the \emph{margin} of $\Omega$. 
The margin is the minimum \emph{positive} distance from a convex hull of a subset of $\Omega$ to the origin. Formally,
\begin{dfn}[Margin]\label{dfn:margin} For a finite set $\Omega \subseteq \RR^m$, define the margin $\gamma(\Omega)$ by
	\[ \gamma(\Omega) := \min \left\lbrace \dist \big(0, \conv(S) \big) \mid S \subseteq \Omega, \; 0 \notin \conv(S) \right\rbrace.\]
\end{dfn}
We point out that for all considered capacity problems in this paper, the margin is actually the \emph{weight margin} (c.f. \cite{gradflow} and our \cref{dfn:WeightMarginGapConstant}) of a certain group representation. For example, the margin for array scaling is the weight margin for tensor scaling. 
We now discuss how the margin enters in decision problems and diameter bounds.

\paragraph{Margin as a precision parameter for the decision problem.}

To illustrate how the margin enters the decision problem of whether $\capa(p) > 0$, consider matrix scaling.
To certify that the capacity of a matrix is nonzero, we compute $\eps$-doubly stochastic scalings for some $\eps$ smaller than the distance to doubly stochastic attained by any matrix that is \emph{not} approximately scalable. This turns out to be precisely $\gamma(\Omega_{n,2})$. More generally, it is a classical fact that for $p$ with support contained in $\Omega$, the gradient $\nabla \log f_p(x)$ can take any value in the Newton polytope of $p$. Thus, $\capa(p) > 0$ if and only if there is some $x$ with $\|\nabla \log f_p(x)\| \leq \gamma(\Omega)$.

For matrix scaling and matrix balancing, it is known that $\gamma(\Omega)$ is on the order of $n^{-3/2}$, despite the exponential number of subsets $S \subseteq \Omega$! This luck can be attributed to the extraordinary geometry of $\Omega$ in these cases, whose elements form the rows of a totally unimodular matrix (up to a shift).
On the other hand, for $d$-dimensional array scaling for $n = 2$, the margin $\gamma(\Omega_{2,d})$ is on the order of the margin of the $d$-dimensional hypercube $\{\pm 1\}^d$, which satisfies $\gamma \big( \{\pm 1\}^d \big)=d^{-\frac{d}{2}(1 + o(1))}$ by \cite{alon1997anti}.
However, between the extreme cases $\Omega_{n,2}$ (matrix scaling) and $\Omega_{2,d}$ (the hypercube), very little is known.

\paragraph{Margin and related quantities for diameter bounds.}
In addition to their role in the decision problem, margins and related quantities can be used to prove diameter bounds for \cref{eq:abelian}. The work \cite{gradflow} proves the diameter bound $\poly(\gamma(\Omega)^{-1}, \log(1/\eps))$. In \cite{straszak2017computing} it is shown that the diameter is polynomial in the logarithm of the minimum nonzero $p_\omega$ and a quantity called the \emph{unary facet complexity}. The latter is defined as the maximal length of an integer normal vector of a face of the Newton polytope $\conv(\supp p)$. In the case of $d$-dimensional arrays, one can use Cramer's rule to crudely bound the unary facet complexity by $(d + 1)^{dn}$. In the case when $0$ is in the relative interior of the Newton polytope, \cite{singh2014entropy} has shown that there is a minimizer with Euclidean norm $O(\log |\supp p|/\eta)$, where $\eta$ is the distance from $0$ to the boundary of the Newton polytope. The diameter bounds in \cite{singh2014entropy,straszak2017computing} were used to design ellipsoid methods that are tractable even for $|\supp p|$ very large, and in \cite{burgisser2020interior} they were used to bound the running time of interior point methods.

\paragraph{Main theorem.}
One is led to ask if the margin remains large for array scaling when $d \geq 3$. We show that this is not the case. In fact, the margin becomes exponentially small in $nd$ for $d \geq 3$. What follows is stated in more detail later in \cref{thm:MarginTensor}.

\begin{thm}\label{thm:tensor-margin} Let $d \geq 3$ and $n \geq 2$. Let $\Omega_{n,d} = \{\eps_{i}: i \in [n]\}^d \subseteq (\RR^n)^d$, where $\eps_j:=e_j - \frac{1}{n} \id_n$. There exists a constant $C>0$, independent of $n$ and $d$, such that
$\gamma(\Omega_{n,d}) \leq 2^{-Cnd}.$ 
\end{thm}

That is, there are $d$-dimensional arrays $p \in (\RR^n_{\geq 0})^{\ot d}$ such that the $d$-tuple of marginals of $p$ is at distance at most $2^{ - Cnd}$ from $\frac{1}{n}(\id_n, \dots, \id_n)$, yet the support of $p$ does not admit an array with uniform marginals, i.e. $\capa(p) = 0$. We note that the support of the array $p$ we construct has $O(nd)$ elements.

\paragraph{Implications of \cref{thm:tensor-margin} and relation to the literature.}
We remark that the construction yields a tensor whose Newton polytope has a facet exponentially close to the origin. Therefore, the bound proved in \cite{burgisser2020interior} on the number of iterations for interior point methods on $3$-tensors is $\Omega(k^{3/2} + k^{1/2} \log(1/\eps))$ for tensors with $O(k)$ nonzero entries.

\cref{thm:tensor-margin} aligns with existing results showing that the $d>2$ array case is more complex than the matrix case. Indeed, it is known that the polytope of arrays with uniform marginals, known as the $d$-\emph{index axial assignment polytope}, has many more vertices when $d \geq 3$ and that the vertices can have exponential entries \cite{linial2014vertices}. In contrast, for $d = 2$ this polytope (known as the Birkhoff-von Neumann polytope) has integral vertices by the Birkhoff-von Neumann theorem.

The exponential rate of decay in \cref{thm:tensor-margin} is tight up to log factors: \cite[Theorem~6.10 Item~3]{gradflow} shows that the margin for $d$-dimensional array scaling is at least $(n \sqrt{d})^{ - dn - 1}$. It is interesting to ask whether the true bound is $2^{-\Theta( n d)}$ as in our upper bound or $2^{- \Theta( nd (\log n + \log d))}$ as in the lower bound. \cite{alon1997anti} shows that the latter is correct in the case $n =2$.

\subsubsection{Proof techniques for the commutative setting}

We first discuss the techniques for proving our margin bounds. \cref{thm:tensor-margin} is proven by explicit construction of witness sets $\Gamma_{n,d} \subseteq \Omega_{n,d} := \{ \eps_i \colon i \in [n] \}^d$,  i.e. $0 \notin \conv(\Gamma_{n,d})$ but zero is exponentially close to $\conv(\Gamma_{n,d})$.  This is done by using that $\sum_i n^{-1} \eps_i$ is the unique way to express zero as a convex combination of the $\eps_i$, compare Lemma~\ref{lem:convCombEps-i}, and by heavily exploiting the combinatorics of $\Omega_{n,d}$. For example, in the case $d=3$ and $n \geq 3$ the key combinatorial idea builds on a construction by Kravtsov in \cite{krav}. Kravtsov's motivation is to characterize the non-integer vertices of the $3$-index axial assignment polytope. He explicitly constructs a certain non-integer vertex with maximal support \cite[Theorem~1 with $k=0$]{krav} which has an exponentially small entry. 

By definition of the 3-index axial assignment polytope, the support of this vertex corresponds to a subset $S\subseteq \Omega_{n,3}$ with $0 \in \conv(S)$. Removing the element of $S$ corresponding to the small entry in Kravtsov's vertex yields our witness set $\Gamma_{n,3}$ with a convex hull very close to zero. In fact, the whole idea generalizes (in a technical way) whenever $d=6r-3$, $r \geq 1$ and $n \geq 3$, see section~\ref{subsec:dTensors}. For $n = 2$ and $d \geq 3$, the bound follows from the existing work \cite{alon1997anti}, as mentioned before. While the construction in that work via $\{-1,1\}$ matrices yields a stronger bound, we provide a different construction of $\{-1,1\}$ matrices\footnote{The $(-1,1)$ matrices from our construction are obtained by replacing all two's in the entries of $A_{2r}$ \eqref{eq:defA2r} with $-1$.}, which has the additional property of \emph{freeness}. The latter will prove useful when we adapt \cref{thm:tensor-margin} to the noncommutative case.

We now discuss the proof of the diameter lower bound, \cref{thm:diameter}. The high level idea is as follows. We first construct a subset $\Omega_0 \subseteq \Omega_{n,3}$ with $0 \in \conv(\Omega_0)$ such that there is another element $\omega \in \Omega_{n,3}$ exponentially close to $\conv(\Omega_0)$, much like our construction of the witness set for small margin discussed above.  We then choose an appropriate array $p$ supported on $\Omega_0 \cup \omega$. This suggests that the only approximate minimizers of $f_p$ have a very large component in the direction $x$ from $\omega$ to $\conv(\Omega_0)$, because as $y \in \RR^m$ tends to a minimizer of $f_p$ the term $e^{y \cdot \omega}$ should vanish compared to the others. This reasoning requires that $y$ is approximately a multiple of $x$; to enforce this we also ensure that zero is far into the relative interior of $\conv(\Omega_0)$.

The structure of this argument bears some similarity to that in \cite{straszak2017computing}, which uses the construction of \cite{alon1997anti}. The main difference is that the set $\Omega_{n,3}$ in the 3-dimensional array scaling problem consists of vectors of very specific structure: up to an additive shift of $-\frac{1}{n} \id_{3n}$, they are Boolean vectors in $\RR^{3n}$ with exactly one nonzero entry among indices in the intervals $[1,n], [n+1,2n], [2n + 1, 3n]$.
Thus, our construction of $\Omega_0$ must consist of vectors of this special form and not simply bounded integral vectors as in \cite{straszak2017computing}. This is the main additional technical contribution of our construction.

\subsection{The noncommutative setting}\label{subsec:noncomm-intro}

In the noncommutative setting, we consider a group $G$ acting on $\CC^m$.\footnote{Technically we require that $G$ is a reductive group over $\CC$ which acts rationally on $\CC^m$. All the group actions in this paper satisfy this assumption.} The optimization problem we investigate is given by the \emph{capacity} of a vector $v \in \CC^m$ (c.f. \cite{gradflow}):
\begin{gather} \capa(v):=\inf_{g \in G} \; f_v(g) := \inf_{g \in G} \; \|g \cdot v\|^2\label{eq:capacity}.
\end{gather}
For the majority of this paper we work with the \emph{tensor scaling action}, in which $G = \SL(n, \CC)^d$, the group of $d$-tuples of complex matrices with determinant one, acts on $v \in (\CC^n)^{\ot d}$ by $(g_1, \dots, g_d) \cdot v = (g_1 \otimes \dots \otimes g_d) v$. The corresponding representation is always denoted by $\pi_{n,d}$.
Sometimes we also consider the \emph{operator scaling action}, in which $\SL(n)^2$ acts on $v \in (\CC^n)^{\ot 2} \otimes \CC^k$ by $(g_1, g_2) \cdot v = (g_1 \otimes g_2 \otimes I_k) v$. 

Though \cref{eq:capacity} looks quite different from \cref{eq:abelian}, one can show that restricting \cref{eq:capacity} to a certain Abelian subgroup of $G$ (a torus) and making a change of variables yields an instance of \cref{eq:abelian} (c.f. \cite{gradflow}). For example, restricting the tensor scaling action to the diagonal matrices in $G$ amounts precisely to the array scaling problem from the previous subsection. Likewise, restricting to diagonal matrices in the operator scaling action yields an instance of matrix scaling.

\paragraph{Relation to null cone problem and Geometric Complexity Theory.}
We study \cref{eq:capacity} because it is deeply connected to invariant theory through a well-known connection between group orbits and invariant polynomials: zero is in the closure of an orbit of a vector $v$ if and only if every non-constant homogeneous $G$-invariant polynomial vanishes on $v$, i.e. if $v$ is in the null-cone. Null-cone membership is a well-studied polynomial identity testing (P.I.T.) problem. One approach to complexity lower bounds, geometric complexity theory, suggests that null-cone membership should be in $\P$ \cite{mulmuley2017geometric, garg2019search}.

Solving \cref{eq:capacity} directly allows one to study the null-cone problem through optimization: one notes that $\capa(v) = 0$ if and only if $v$ is in the null cone. In fact, \cref{eq:capacity} is a geodesically convex optimization problem over a certain Riemannian manifold. 
Algebraic and optimization-based algorithms have, independently and nearly concurrently, resulted in polynomial time algorithms for nearly the same set of P.I.T. problems arising in invariant theory \cite{forbes2013explicit,mulmuley2017geometric, garg2016deterministic, ivanyos2017constructive, derksen2018algorithms,allen2018operator}, including the null-cone problem for the operator scaling and simultaneous conjugation action. However, neither approach has succeeded in solving the null-cone problem for the $3$-tensor action. Recent degree lower bounds for invariant polynomials for the $3$-tensor action pose significant challenges for the algebraic approach \cite{derksen2020exponential}. It is natural to ask whether the optimization approach can overcome these challenges.

\paragraph{Algorithms for computing the capacity.} 
A nonzero tensor $w = g \cdot v$ attains the capacity when $w$ has all \emph{quantum marginals} equal to $I_n/n$. The quantum marginals of a tensor $w$, analogous to the sums along slices of an array, are the three $n\times n$ matrices $M_1 M_1^\dagger, M_2^\dagger, M_3 M_3^\dagger$ for the $n \times n^2$ matrices $M_1, M_2, M_3$ known as \emph{flattenings} of $w/\|w\|$. For operator scaling, the capacity is attained when the first two quantum marginals are $I_n/n.$ To compute the capacity, existing algorithms attempt to find $g$ such that the quantum marginals of $g \cdot v$ are all close to $I_n/n$. There are alternating minimization algorithms that can attain distance $\eps$ in time $\poly(n, 1/\eps)$ \cite{garg2016deterministic, burgisser2017alternating}, and for the operator scaling this is possible in $\poly(n, \log(1/\eps))$ time \cite{allen2018operator}. However, for 3-tensor scaling, running time $\poly(1/\eps)$ is not sufficient to efficiently decide null-cone membership, and the only algorithms with $\log(1/\eps)$ dependence on $\eps$ have an exponential dependence on $n$ \cite{gradflow}.

To explain the increased complexity, we discuss a noncommutative analogue of the Newton polytope known as the \emph{moment polytope}, denoted $\Delta_G(v)$. In particular, $0 \notin \Delta_G(v)$ if and only if $v$ is in the null-cone (i.e. $\capa(v) = 0$).\footnote{Moment polytope membership is an interesting problem in and of itself; for $d = 3$, for generic $v \in (\CC^n)^{\ot 3}$, $\Delta_G(v)$ is the \emph{Kronecker polytope} arising in representation theory and quantum information theory. Deciding membership in this polytope is known to be in $\NP \cap \coNP$ but not known to be in $\P$ \cite{burgisser2017membership}. }
For tensor scaling, the moment polytope is the set of tuples of spectra of the quantum marginals as $w$ ranges over $\overline{G \cdot v}$, shifted by $- \frac{1}{n}(\id_n, \id_n,\id_n)$.
The \emph{gap} of the action of $G$, i.e. the minimum positive distance from $0$ to a moment polytope $\Delta_G(v)$, is a noncommutative generalization of the margin. Whereas the operator scaling and simultaneous conjugation actions have polynomially large gaps, we show that the gap for the tensor scaling action is exponentially small. Scaling algorithms amount to outer $\eps$-approximation algorithms for $\Delta_G(v),$ which is why $\poly(1/\eps)$-time algorithms do not suffice to decide null-cone membership. Like for the margin, the smaller gap corresponds to a larger diameter, which is why so far no algorithm has had running time $\poly(n, \log(1/\eps))$.

\subsubsection{Diameter lower bound for noncommutative scaling}

Here we describe how diameter bounds cause the state-of-the-art algorithms to be slow for the tensor scaling action. We begin by discussing geodesically convex optimization. In general \cref{eq:capacity} is not convex, but rather \emph{geodesically convex}. That is, $G$ can be viewed as a manifold in such a way that the function $g \mapsto \|g \cdot v\|^2$ is convex along ``geodesics'' of the form $\gamma(t) = e^{t H} g$ for $H$ Hermitian. The manifold we consider is not exactly $G$ but rather a quotient $P$ of it; we will make this more precise later in \cref{subsec:free-diameter}. For $G = \SL(n)^d$, the manifold $P$ is the set of tuples of positive-definite matrices with determinant one. $P$ is equipped with the geometry on positive-definite matrices known in statistics as the Fisher-Rao metric, and studied in depth in e.g. \cite{bhatia2009positive}. Though we do not need many details of this geometry here, one can think of the distance between $g, h \in G$ as a bound on the logarithms of the singular values of $g^{-1}h$. In particular, the geodesic ``ball'' of radius $R$ about the identity in $G$ is the intersection of $G$ with the set $\{U \exp(A): A\text{ Hermitian}, \|A\|_F\leq R, U \text{ unitary}\}$. Note that the ball of radius $\sqrt{n}R$ includes all elements of $G$ whose singular values are in $[e^{-R}, e^{R}]$. \footnote{We define exponentials, Hermitian-ness, and Frobenius norm on tuples by treating them as block diagonal matrices.}

The existing algorithms to compute \cref{eq:capacity} adapt simple first order methods, such as gradient descent, and second order methods, such as trust regions, to the geodesically convex setting \cite{zhang2016first,allen2018operator,gradflow}. As in the commutative case, to run in polynomial time such algorithms require that an $\eps$-approximate solution is contained in a geodesic ball of radius $\poly(n^d, \log(1/\eps))$. However, for $3$-tensors we have the following diameter lower bound.

\begin{thm}[Noncommutative diameter lower bound]\label{thm:nc-diameter} There is a constant $C > 0$ such that the following holds. For all $\eps \leq  \exp(- C n^2 \log n)$, there is a tensor $v = v(\eps) \in (\CC^n)^{\ot 3}$ with $O(n)$ nonzero entries of bit complexity $O(\log n + \log(1/\eps))$, and a geodesic ball $B=B(\eps)$ of radius $\Omega\left(2^{n/3}\log(1/\eps)\right)$ about the identity in $\SL(n)^3$, such that 
$$ \inf_{g \in B} \; \|g \cdot v\|^2 \geq \capa(v) + \eps.$$
\end{thm} 

To emphasize that the difficulties are not caused by requiring additive approximation, we remark that the vector $v$ satisfies $1/4 \leq \capa(v) \leq 1$ and $1/2 \leq \|v\| \leq 1$. A duplication trick analogous to \cref{cor:diameter-d} yields the same diameter bound for $d \geq 3$, but for the action of $G$ simultaneously on a tuple of tensors rather than on a single one. See \cref{cor:diameter-d-noncomm}.

\paragraph{Implications of \cref{thm:nc-diameter} and relation to the literature.} \cref{thm:nc-diameter} shows that trust region methods with constant step size cannot $\eps$-approximate the capacity in $\poly(n, 1/\eps)$ time for $3$-tensors. It also shows that cutting plane methods are unlikely to do so.
Cutting plane methods, such as ellipsoid, require an exponential bound on the volume of a known region containing an approximate optimizer. This is the case for Rusciano's non-constructive query upper bound for cutting plane methods on manifolds of non-positive curvature \cite{rusciano2018riemannian}, which is essentially tight \cite{hamilton2021no}\footnote{\cite{hamilton2021no} applies to the hyperbolic plane, which is a totally geodesic submanifold of the manifold $P$ we consider}. The volume of a ball in the manifold we consider grows exponentially in the radius (see \cref{subsec:free-diameter}), so this query bound will be exponential.
Regarding tightness, the best upper bound known to the authors for the diameter bound in the noncommutative case is $O( n (\sqrt{3} n)^{1 + 3n} \log (1/\eps))$, which can be deduced from the diameter and margin bounds \cite[Proposition 5.6, Theorem 6.10]{gradflow}. This matches our lower bound up to logarithmic factors in the exponent. As with \cref{thm:diameter}, \cref{thm:nc-diameter} holds only values of $\eps$ that are very small (though still of polynomial bit-complexity). It would be very interesting to prove a version of \cref{thm:diameter} for $\eps$ larger than the gap, which is $\exp(-O(n))$. This would imply that trust region methods cannot solve the null-cone problem for the $3$-tensor action in polynomial time.

\subsubsection{Gaps: the geometry of noncommutative scaling problems}

In analogy to the commutative case, one typically attempts to certify $\capa(v) > 0$, i.e. $0 \in \Delta_G(v)$, by finding a tensor $g\cdot v$ such that all the quantum marginals are close to $\frac{1}{n} I_n$. In order to certify $\capa(v) > 0$ their distance to $\frac{1}{n}(I_n, I_n, \dots, I_n)$ must be at most a certain quantity, which we call the \emph{gap}.

\begin{dfn}[Gap] The \emph{gap}\footnote{This notion can be defined similarly for any rational representation $\pi$ of a reductive group $G$, see \cref{dfn:WeightMarginGapConstant}. This definition of the gap is already described in \cite{gradflow}.} for the $d$-tensor scaling problem is
	\[ \gamma_G(\pi_{n,d}) := \min \left\lbrace \dist \big( 0, \Delta_G(v) \big) \mid v \in (\CC^n)^{\ot d}, \; v \neq 0,\; 0 \notin \Delta_G(v) \right\rbrace. \]
\end{dfn}

If the gap is exponentially small, high-precision algorithms will be necessary to decide if $\capa(v) > 0$. In operator scaling, the gap is known to be $\Omega(n^{-3/2})$ \cite{gurvits2004classical}, which explains why we do not need high-precision algorithms for the decision problem in that case. In addition to its role in the decision problem, the inverse of the gap\footnote{actually, a smaller quantity known as \emph{weight margin}} is used to control the diameter bound \cite{gradflow}! In that sense, the presence of a small gap can explain both the need for high precision algorithms and the slowness of existing high-precision algorithms. We show that, indeed, the tensor scaling action has an exponentially small gap for $d \geq 3$. 

\begin{thm}\label{thm:tensor-gap} There is a constant $C > 0$ such that for all $d \geq 3$ and $n \geq 2$, there are non-zero tensors $v \in (\CC^n)^{\ot d}$ such that $0 \not\in \Delta_G(v)$ but $ \dist ( 0, \Delta_G(v) ) \leq 2^{- C dn}$.
That is, the gap for $d$-tensor scaling satisfies $$\gamma_G(\pi_{n,d}) \leq 2^{-C dn}.$$
\end{thm}

A detailed statement on bounds for the gap can be found in \cref{thm:GapConstantTensor}, and we show in \cref{subsec:extension} how to fill in the missing values of $n,d$ to obtain \cref{thm:tensor-gap}. Since the gap is larger than the margin (c.f.  \cref{prp:GapConstantWeightMargin}), \cref{thm:tensor-gap} is at least as tight as \cref{thm:tensor-margin}, i.e. the exponent $Cnd$ is tight up to an $O(\log n + \log d)$ factor. 

Interestingly, for local dimension $n=2$ \cite[Main result]{MaciazekSawicki2018} shows that $\dist(0,\Delta_{G}(v))^2$ for some moment polytope $\Delta_G(v) \notni 0$ tends for $d \to \infty$ to the Gamma distribution $\Gamma(1/2, 2d)$, where $2d$ is the rate parameter. Therefore, the witnesses of the exponential behaviour in \cref{thm:GapConstantTensor}(a) are quite rare. Moreover, the authors numerically found several tensors of format $(\CC^{2})^{\otimes d}$ with $\dist(0, \Delta_G(v))$ at most $\exp(-d);$ \cref{thm:tensor-gap} confirms that this exponential behavior is the case for all $n$ and $d$.

\subsubsection*{Margin and gap results for other group actions}

In addition to the tensor scaling action, we also consider some other actions of groups $G$ of interest in computational invariant theory. The first is the action of the special linear group on the space of homogeneous $d$-forms $\CC[x_1, \dots, x_n]_d$, in which $G = \SL(n)$ acts by $g \cdot p (x) = p (g^{-1} x)$ for $p \in \CC[x_1, \dots, x_n]_d$. Homogeneous $d$-forms were among the objects studied earliest in computational invariant theory, and much of the theory was developed to catalogue invariants of the $\SL(n)$ action on forms \cite{weyl1946classical}. Still, deciding null-cone membership for $d = 3$ seems challenging. After extending the definition of the gap to other group actions in \cref{sec:noncommutative}, we explain the difficulty by showing that the gap for this action is also inverse exponential in $n$ as soon as $d \geq 3$, see \cref{thm:dFormsGap}. This shows that the diameter bound in \cite{gradflow} becomes exponentially large in $n$.

The other group action we consider is the action of $\SL(n)^d$ on \emph{quivers} with $d$ vertices. A quiver is a directed multigraph, and a quiver representation is a labelling of the vertex set $Q_0$ of the quiver with finite-dimensional vector spaces and the edge set $Q_1$ with a linear map from the vector space at the tail of the edge to the vector space at the head of the edge. Given a quiver representation $A$ with vertices labeled by $\CC^{n_x}$ for $x \in Q_0$ and edges $e:x\to y$ labeled with matrices $A_e$, the group $G = \prod_{x \in Q_0} \SL(n_x)$ acts on $A$ by $(g \cdot A)_e = g_y A g_x^{-1}$. Quiver representations include the operator scaling action, and an action used to bound the Brascamp-Lieb constant in analysis. In \cref{subsec:quivers} we show that the \emph{(weight) margin} can become exponentially small as the number of vertices grows. For this, we exhibit a quiver with $d-1$ arrows, $d$ vertices of dimension $n$ and weight margin $O(n^{-d})$, see \cref{thm:UpperBoundQuiver}. This bound shows that the diameter bound computed in \cite{gradflow} can become exponentially large in $d$. Furthermore, when allowing $n$ copies of each arrow in the constructed quiver, i.e. $n(d-1)$ arrows in total, we can ensure the same bound for the gap, \cref{thm:UpperBoundQuiver}.

\subsubsection{Proof technique in the noncommutative case: Freeness}

Regarding the idea of the proof, we may transfer both the diameter lower bound and the gap upper bound to the commutative case by virtue of the tensors we construct having \emph{free support}. 

A tensor has free support if any two distinct $(d-1)$-dimensional slices of the tensor have disjoint support. This condition ensures that, even after being acted on by any diagonal group elements, the tensor's quantum marginals are all diagonal. This allows us to restrict to the action of the diagonal matrices and thereby reduce to the commutative (array scaling) case. Thus, we may obtain the same bounds on the tensor gap as for the array margin. However, this requires additional care to ensure freeness of our constructions. This is why we cannot na\"ively use the construction of \cite{alon1997anti} for $d$-tensors with $n = 2$. Regarding the noncommutative diameter bound, we show that for tensors with free support the diameter bound matches that of the commutative problem obtained by restricting to the diagonal. To do this, we project the group elements to the set of diagonal elements, and use the properties of spaces of non-positive curvature to show that this projection moves the point nearer to the origin and decreases the function value. 

The idea and the concept of freeness generalize to rational representations of reductive groups \cite{franz}.\footnote{This concept is also implicitly contained in \cite[Lemma~7.1]{Sjamaar} and can at least be traced back to \cite{dadok1985polar} as \emph{strong orthogonality}.}  The key statement is given in full generality in Proposition~\ref{prp:FreeForGapConstant}. This proposition is needed to prove bounds on the gap for the action on homogeneous polynomials and for the action on quivers.  Interestingly, in \cite{derksen2020exponential} the concept of freeness is used in a similar way\footnote{Indeed, \cite[Theorem~6.5]{derksen2020exponential} is used to show the vanishing of the moment map at a vector. First, freeness is used as in \cref{prp:FreeForGapConstant} to ensure that one can restrict to the moment map for the maximal torus. Second,  condition~(2) of \cite[Theorem~6.5]{derksen2020exponential} just states that the moment map for the torus action vanishes at the vector.}
to prove exponential lower bounds on the degree of invariants for actions on cubic forms and $3$-tensors. There, \emph{free} is called \emph{uncramped} and it is used crucially to prove closedness of certain orbits.

Freeness also played a role in the numerical results by Sawicki and Maci{\k{a}}{\.{z}}ek, which were obtained by applying the algorithm of \cite{MaciazekSawicki2015} to several free tensors of local dimension two.

\subsection{Organization of the paper}

We begin with the commutative case, which is split into the study of the margin in \cref{sec:MarginComm} and diameter bounds in \cref{sec:diameterComm}. Then we move to the noncommutative case in \cref{sec:noncommutative}. The appendix contains some representation-theoretic background and proofs of technical lemmas, as well as a glossary of notation.


\section{The geometry of commutative scaling problems}\label{sec:MarginComm}

The purpose of this section is to show the following theorem on the margin of $d$-dimensional array scaling. Recall that the latter arises for $\Omega_{n,d} := \{\eps_i: i \in [n]\}^d \subseteq (\RR^n)^d$.

\begin{thm}[Margin for array scaling]\label{thm:MarginTensor}
The margin of $\Omega_{n,d} \subseteq (\RR^n)^d$ is bounded as follows.
	\begin{itemize}
	\item[(a)] If $n=2$ and $d \geq 3$, then
	$
	\gamma \left(\Omega_{2,d} \right)  \leq  2^{-\frac{d}{2} + 1}.
	$
	\item[(b)] If $n \geq 3$ and $d = 3$, then $\gamma(\Omega_{n,3}) \leq 2^{-n+1}$.
	\item[(c)] If $n \geq 3$ and $d = 6 r - 3$ for some integer $r \geq 2$, then
		\begin{align*}
		\gamma(\Omega_{n,d})  \leq \frac{\sqrt{6}}{(n-1)\sqrt{r}} \; 2^{-r(n-1) + 1} 
		\leq 2^{- r (n-1) + 1} = 2^{- \frac{(d+3)(n-1)}{6} + 1}.
		\end{align*}
		\end{itemize}
\end{thm}
By ``padding'' the tensors appropriately, one sees that a bound for $\gamma(\Omega_{n,d})$ also applies to $\gamma(\Omega_{n,d+1})$ (see  \cref{prp:dTensorsPadding}). Combining this result with \cref{thm:MarginTensor} above implies \cref{thm:tensor-margin} from the introduction. The next three subsections each prove one of the parts of \cref{thm:MarginTensor}; the construction for part~(a) with $n = 2$ is slightly different and the construction for part (c), $d > 3$ builds on the one for part~(b), $d = 3$.

To prove the results, we will frequently use the following simple lemma. Recall that an \emph{affine linear combination} of $v_1, \dots, v_k \in \RR^m$ is $\lambda_1 v_1 + \dots + \lambda_k v_k$ for $\lambda_i \geq 0, \sum_{i = 1}^k \lambda_i = 1$. The affine hull $\aff(S)$ of a set $S\subset \RR^m$ is the set of all affine linear combinations of finite subsets of $S$, or equivalently the affine space (i.e. translate of a subspace) of lowest dimension containing $S$.

\begin{lem}\label{lem:convCombEps-i}
In $\RR^n$ we have
\begin{equation}\label{eq:SumEps-i}
\sum_{i=1}^n \frac{1}{n} \: \varepsilon_i = 0_n 
\end{equation}
and this is the only affine linear combination of $\eps_1,\ldots,\eps_n$ giving zero.
\end{lem}

\begin{proof}
One calculates directly that $\sum_i \frac{1}{n} \, \eps_i = 0_n$. To show uniqueness of this affine combination, we note that the vectors $e_2, \ldots, e_n, \id_n$ are linearly independent. Thus, $\eps_2, \ldots, \eps_n$ are linearly independent. On the other hand, $\eps_1,\ldots, \eps_n$ are linearly dependent. Therefore
	$
	\left\lbrace (\lambda_1,\ldots,\lambda_n) \in \RR^n \mid \sum_i \lambda_i \: \eps_i = 0_n \right\rbrace
	$
is a one-dimensional subspace of $\RR^n$, which yields the uniqueness of the affine linear combination.
\end{proof}

\subsection{Local dimension two: the hypercube}\label{sec:qubits}

In this subsection we prove part (a) of \cref{thm:MarginTensor} by showing that the margin of $\Omega_{2,d}$ is exponentially small in $d$. This follows from \cite{alon1997anti}, but we present a new construction which has the additional property of \emph{freeness}, which we discuss later in \cref{sec:noncommutative}. Recall that 
	\begin{align*}
	\Omega_{2,d} =  \big\lbrace (\eps_{i_1}, \ldots, \eps_{i_d}) \mid i_1, \ldots, i_d \in [2] \big\rbrace \subseteq \big( \RR^2 \big)^d.
	\end{align*}
In the following we construct a subset of $\Omega_{2,d}$, which witnesses the exponentially small margin. For this, we construct a matrix with entries in $[2]$, and each row of the matrix will correspond to an element of $\Omega_{2,d}$. For example, the row $(1,2,2)$ would correspond to $(\eps_1, \eps_2, \eps_2) \in \Omega_{2,3}$. To do so, we begin with the matrices 
	\begin{align*}
	A_2 : = \begin{pmatrix} 1 & 1 \\ 2 & 1 \end{pmatrix} , \;
	B_1 := \begin{pmatrix} 1 & 1 \\ 2 & 2 \end{pmatrix} , \;
	B_2 := \begin{pmatrix} 1 & 2 \\ 2 & 2 \end{pmatrix} , \;
	B_3 := \begin{pmatrix} 2 & 1 \\ 1 & 1 \end{pmatrix} ,
	\end{align*}
and define recursively
	\begin{equation}\label{eq:defA2r}
	A_{2r+2} := \begin{pmatrix}
	  &  &  & B_1 \\ 
	  &  A_{2r} &  & \vdots \\ 
	  &  &  & B_1 \\ 
	B_2 & \cdots & B_2 & B_3
	\end{pmatrix} = 
	\begin{pmatrix}
	 A_2 & B_1 & \cdots & B_1 \\ 
	 B_2  & B_3 & \ddots & \vdots \\ 
	 \vdots & \ddots & \ddots& B_1 \\ 
	 B_2 & \cdots & B_2 & B_3
	\end{pmatrix}
	\end{equation}
for $r \geq 1$. \cref{fig:qubit-matrices} is supplied as a visualization aid.

\begin{figure}
$$ A_4 = \left(\begin{array}{c|c}
\cellcolor{black!30}\begin{array}{cc} * & * \\  & * \end{array}&\cellcolor{green!10} \begin{array}{cc} * & * \\  &  \end{array}\\
\hline
\cellcolor{red!10}\begin{array}{cc} * & \;\text{ }  \\  &  \end{array}& \cellcolor{blue!10}\begin{array}{cc}  & * \\ * & * \end{array}
 \end{array}\right),
 \quad 
 A_6 = \left(\begin{array}{c|c|c} 
 \cellcolor{black!30} \begin{array}{cc} * & * \\ & * \end{array}& \cellcolor{green!10}\begin{array}{cc} * & * \\  &  \end{array} & \cellcolor{green!10} \begin{array}{cc} * & * \\  &  \end{array} \\
 \hline
\cellcolor{red!10}\begin{array}{cc} * & \;\text{ }  \\  &  \end{array} & \cellcolor{blue!10}  \begin{array}{cc}  & * \\ * & * \end{array} &  \cellcolor{green!10}\begin{array}{cc} * & * \\  &  \end{array}  \\
\hline
\cellcolor{red!10}\begin{array}{cc} * & \;\text{ } \\  &  \end{array}& \cellcolor{red!10}\begin{array}{cc} * & \;\text{ } \\  &  \end{array} &  \cellcolor{blue!10}\begin{array}{cc}  & * \\ * & * \end{array} \\
\end{array}\right)
  $$
\caption{The positions of the ones in $A_4$ and $A_6$ are marked by $*$ in the following figure and the cells are colored according to whether they belong to $A_2, B_1, B_2$ or $B_3$.}\label{fig:qubit-matrices}
\end{figure}

We remark that the entry of $A_{2r}$ at position $(i,j)$ is independent of $r$ and denote it by $a(i,j)$. We set for $r \geq 1$
	\begin{align*}
	\Gamma_{2,2r} &:= \left\lbrace \left( \eps_{a(i,1)}, \eps_{a(i,2)}, \ldots, \eps_{a(i,2r)} \right) \mid i \in [2r] \right\rbrace \subseteq \Omega(\pi_{2,2r}) \subseteq \big( \RR^2 \big)^{2r},\\
	\Gamma_{2,2r+1} &:= \left\lbrace \left( \eps_{a(i,1)}, \eps_{a(i,2)}, \ldots, \eps_{a(i,2r)}, \eps_{\chi(i)} \right) \mid i \in [2r] \right\rbrace \subseteq \Omega(\pi_{2,2r+1}) \subseteq \big( \RR^2 \big)^{2r+1},
	\end{align*}
where $\chi \colon \NN \to \{1,2\}, \; i \mapsto i \mod 2$. That is, $\Gamma_{2,2r}$ is the subset of $\Omega_{2,2r}$ induced by the rows of $A_{2r}$ and $\Gamma_{2,2r+1}$ is obtained by alternatingly appending $\eps_1$ or $\eps_2$ to the $2r$-many elements of $\Gamma_{2,2r}$.
	
\begin{lem}\label{lem:affHullQubits}
For $r \geq 1$ it holds that $0 \notin \aff(\Gamma_{2,2r})$ and $0 \notin \aff(\Gamma_{2,2r+1})$.
\end{lem}

\begin{proof}
By construction, $0 \in \aff(\Gamma_{2,2r+1})$ implies $0 \in \aff(\Gamma_{2,2r})$, so it suffices to prove $0 \notin \aff(\Gamma_{2,2r})$. We proceed by induction on $r \geq 1$. For $r=1$, it is clear that $0 \notin \aff(\Gamma_{2,2}) \subseteq \RR^2 \times \lbrace \eps_1 \rbrace$.
Now assume that $0 \notin \aff(\Gamma_{2,2r})$. For the sake of contradiction, let
	\begin{equation}\label{eq:affCombGamma}
	\sum_{i=1}^{2r+2} \lambda_i \left( \eps_{a(i,1)}, \eps_{a(i,2)}, \ldots, \eps_{a(i,2r+2)} \right) = 0 \in \left( \RR^2 \right)^{2r+2}
	\end{equation}
be an affine linear combination of $\Gamma_{2,2r+2}$. Then equation~\eqref{eq:affCombGamma} gives in each of the $(2r+2)$-many $\RR^2$-components the affine linear combination $2^{-1} (\eps_1 + \eps_2) = 0$, by Lemma~\ref{lem:convCombEps-i}. Considering the scalar factor of $\eps_1$ in the first, the penultimate and the last $\RR^{2}$-component respectively, we conclude
	\begin{align*}
	\underbrace{ \sum_{j=1}^{r+1} \lambda_{2j-1} }_{\text{first}} = \frac{1}{2}
	= \underbrace{ \lambda_{2r+2} + \sum_{j=1}^{r} \lambda_{2j-1} }_{\text{penultimate}} 
	= \frac{1}{2} = \underbrace { \lambda_{2r+2} + \sum_{j=1}^{r+1} \lambda_{2j-1} }_{\text{last}}
	\end{align*}
by construction of $A_{2r+2}$. Hence, $\lambda_{2r+2} = 0$ using the first and last component. Furthermore, the first and penultimate column give $\lambda_{2r+1} = \lambda_{2r+2} = 0$. Therefore, the first $2r$-many components in \cref{eq:affCombGamma} show $0 \in \aff(\Gamma_{2,2r})$, which contradicts our induction hypothesis.
\end{proof}

\begin{lem}\label{lem:convCombQubits}
For $r \geq 1$ it holds that $\dist(0, \conv(\Gamma_{2,2r}) ) \leq 2^{-r+ \frac{1}{2}}$ and $\dist(0, \conv(\Gamma_{2,2r+1}) ) \leq 2^{-r+ \frac{1}{2}}$.
\end{lem}

\begin{proof}
We first prove the inequality for $\conv(\Gamma_{2,2r})$. For $i \in [2r]$ let $\omega_i := \big( \eps_{a(i,1)}, \ldots, \eps_{a(i,2r)} \big) \in \left( \RR^2 \right)^{2r}$ be the weight in $\Gamma_{2,2r}$ that corresponds to the $i^{th}$ row of $A_{2r}$. Consider the convex combination
	\begin{equation}\label{eq:convCombGamma}
	(x_1, \ldots, x_{2r}) := 2^{-r} ( \omega_{2r-1} + \omega_{2r} ) + \sum_{l=1}^{r-1} 2^{-l-1} ( \omega_{2l-1} + \omega_{2l}) \in \left( \RR^2 \right)^{2r}.
	\end{equation}
Note that $x_i \in \RR^2$. We will argue that $(x_1, \dots, x_{2r}) = 2^{-r+1} (0_2,\ldots,0_2, \eps_{1})$.
Since $x$ is a convex combination of the elements in $\Gamma_{2,2r}$, the statement then follows from $\| \eps_1 \| = \sqrt{2}^{-1}$.

We consider $A_{2r}$ like in its construction~\eqref{eq:defA2r} as a $r \times r$ block matrix with block entries being $2 \times 2$ matrices. For $m \in [r]$ the two weights $\omega_{2m-1}$ and $\omega_{2m}$ correspond to the $m^{th}$ block row of $A_{2r}$ and have the same scalar factor in~\eqref{eq:convCombGamma}. Hence, whenever for $i \in [2r]$ the $i^{th}$ column of the $m^{th}$ block row of $A_{2k}$ contains exactly one entry equal to one (and so the other entry equals two), then the contribution of $\omega_{2m-1}$ and $\omega_{2m}$ to $x_i$ cancels due to $\eps_1 + \eps_2 = 0_2$.
In particular, in \eqref{eq:convCombGamma} all contributions of block entries equal to $B_1$ cancel. Therefore the last column of $A_{2r}$ gives
	\begin{align*}
	x_{2r} = 2^{-r} (\eps_1 + \eps_1) = 2^{-r+1} \eps_1.
	\end{align*}
Furthermore, $x_1 = x_{3} = \ldots = x_{2r-1} = 0_2$ using that also the first columns of $A_2$, of $B_2$ and of $B_3$ contain exactly one entry equal to one. For $r=1$ we are done.
If $r \geq 2$, then reading off the second column of $A_{2r}$, we find 
	\begin{align*}
	x_2 = \underbrace{2^{-2} (\eps_1 + \eps_1)}_{\text{first block row}} + \underbrace{2^{-r} (\eps_2 + \eps_2)}_{\text{last block row}} + \sum_{l=2}^{r-1} \underbrace{ 2^{-l-1} (\eps_2 + \eps_2)}_{\text{middle rows}} = 2^{-1} (\eps_1 + \eps_2) = 0_2.
	\end{align*}
Analogously, as $B_1$ does not contribute we compute for $j = 2,3,\ldots,r-1$ that
	\begin{align*}
	x_{2j} = \underbrace{2^{-j-1} (\eps_1 + \eps_1)}_{j^{th} \text{ block row}} + \underbrace{2^{-r} (\eps_2 + \eps_2)}_{\text{last block row}} + \sum_{l=j+1}^{r-1} \underbrace{ 2^{-l-1} (\eps_2 + \eps_2)}_{\text{in between rows}} = 2^{-j} (\eps_1 + \eps_2) = 0_2,
	\end{align*}
because the second columns of $B_2$ and $B_3$ are, respectively, $(2,2)^T$ and $(1,1)^T$. This proves the inequality in the case $\Gamma_{2,2r}$.

By construction, for $\Gamma_{2,2r+1}$ the same convex combination works, because the last $\RR^2$-component does not contribute as the entries of the weights alternate between $\eps_1$ and $\eps_2$.
\end{proof}

Finally, Lemma~\ref{lem:affHullQubits} and Lemma~\ref{lem:convCombQubits} together yield Theorem~\ref{thm:MarginTensor}(a), noting that for odd $d=2r+1$ one has $\; -r + 1/2 = -(d/2) + 1$.

\subsection{3-tensors}\label{sec:3tensors}

The main goal of this section is to show that the margin of $\Omega_{n,3}$ is exponentially small in $n$, i.e. to show \cref{thm:MarginTensor}(b). To do so, we set
\begin{equation}
\mathfrak{W}_n := \bigcup_{s=2}^n \lbrace (s,1,s), (s,s,1), (s-1,s,s) \rbrace \subseteq [n] \times [n] \times [n]
\end{equation}
and consider the corresponding subset
\begin{equation}
\Gamma_{n,3} := \big\lbrace (\eps_i,\eps_j,\eps_k) \mid (i,j,k) \in \mathfrak{W}_n \big \rbrace \subseteq \Omega_{n,3}.
\end{equation}

The key combinatorial idea, which is presented in the following lemma, is due to \cite[Theorem~1 with $k=0$]{krav}.\footnote{In \cite{krav} Kravtsov extensively studies so-called complete $r$-noninteger vertices ($r$-CNVs) of the three-index axial assignment polytope. For $k \in \{0,1,\ldots,n-2\}$, \cite[Theorem~1]{krav} states explicitly a $(3n-2-k)$-CNV, among these we use the $(3n-2)$-CNV (i.e. $k=0$). Moreover, \cite[Theorem~2]{krav} states that such $r$-CNVs of the three-index axial assignment polytope actually only occur for $r \in \{2n,2n+1, \ldots, 3n-2\}$, and the later theorems in \cite{krav} fully characterize the $r$-CNVs and study their combinatorial properties.}  According to \cite{krav} the special case $k=0$ is already contained in \cite[Theorem~9]{luk}.

\begin{lem}\label{lem:Kravtsov}
Let $n \geq 3$. For $(i,j,k) \in [n]^3 \setminus \big(\mathfrak{W}_n \cup \lbrace (1,1,1) \rbrace \big)$ set $\lambda_{i,j,k} := 0$. Moreover, define
\begin{align*}
\lambda_{1,1,1} := 2^{-n+1}, \quad \lambda_{1,2,2} := 1- 2^{-n+1}, \quad \lambda_{n,1,n} = \lambda_{n,n,1} := 2^{-1}
\end{align*}
and for $s=2,3,\ldots,n-1$
\begin{align*}
\lambda_{s,1,s} = \lambda_{s,s,1} := 2^{-n+s-1} ,\quad \lambda_{s,s+1,s+1} := 1-2^{-n+s} \, .
\end{align*}
Then the following equations hold:
\begin{align}
\left( \forall i \in [n]: \; \sum_{j,k = 1}^n \lambda_{i,j,k} = 1 \right), \;
\left( \forall j \in [n]: \; \sum_{i,k = 1}^n \lambda_{i,j,k} = 1 \right), \;
\left( \forall k \in [n]: \; \sum_{i,j = 1}^n \lambda_{i,j,k} = 1 \right) \, .\label{eq:marginals-krav}
\end{align}
In particular, $\,\sum_{i,j,k} \lambda_{i,j,k} = n$.
\end{lem}

\begin{proof}
This is \cite[Theorem 1 with $k=0$]{krav}. Alternatively, the statement can be checked by straightforward computation.
\end{proof}

\begin{exa}
To visualize \cref{lem:Kravtsov} it is helpful to consider the slices $\Lambda_i$ given by $(\Lambda_i)_{j,k} = \lambda_{i,j,k}$.
For $n=4$ one has
    \begin{align*}
        \Lambda_1 = {\color{blue}\frac{1}{8}} \begin{pmatrix}
            {\color{blue}1} & 0 & 0 & 0 \\ 0 & {\color{blue}7} & 0 & 0 \\ 0 & 0 & 0 & 0 \\ 0 & 0 & 0 & 0
            \end{pmatrix},\quad
        \Lambda_2 = {\color{blue}\frac{1}{8}} \begin{pmatrix}
            0 & {\color{blue}1} & 0 & 0 \\ {\color{blue}1} & 0 & 0 & 0 \\ 0 & 0 & {\color{blue}6} & 0 \\ 0 & 0 & 0 & 0
            \end{pmatrix},\quad
        \Lambda_3 = {\color{blue}\frac{1}{8}} \begin{pmatrix}
            0 & 0 & {\color{blue}2} & 0 \\ 0 & 0 & 0 & 0 \\ {\color{blue}2} & 0 & 0 & 0 \\ 0 & 0 & 0 & {\color{blue}4}
            \end{pmatrix},\quad
        \Lambda_4 = {\color{blue}\frac{1}{8}} \begin{pmatrix}
            0 & 0 & 0 & {\color{blue}4} \\ 0 & 0 & 0 & 0 \\ 0 & 0 & 0 & 0 \\ {\color{blue}4} & 0 & 0 & 0
            \end{pmatrix}.
    \end{align*}
For $n = 5$ one has
    \begin{align*}
        \Lambda_1 = {\color{blue}\frac{1}{16}} \begin{pmatrix}
            {\color{blue}1} & 0 & 0 & 0 & 0 \\ 0 & {\color{blue}15} & 0 & 0 & 0 \\ 0 & 0 & 0 & 0 & 0 \\ 0 & 0 & 0 & 0 & 0 \\ 0 & 0 & 0 & 0 & 0
            \end{pmatrix},\quad
        \Lambda_2 = {\color{blue}\frac{1}{16}} \begin{pmatrix}
            0 & {\color{blue}1} & 0 & 0 & 0 \\ {\color{blue}1} & 0 & 0 & 0 & 0 \\ 0 & 0 & {\color{blue}14} & 0 & 0 \\ 0 & 0 & 0 & 0 & 0 \\ 0 & 0 & 0 & 0 & 0
            \end{pmatrix},\quad
        \Lambda_3 = {\color{blue}\frac{1}{16}} \begin{pmatrix}
            0 & 0 & {\color{blue}2} & 0 & 0 \\ 0 & 0 & 0 & 0 & 0 \\ {\color{blue}2} & 0 & 0 & 0 & 0 \\ 0 & 0 & 0 & {\color{blue}12} & 0 \\ 0 & 0 & 0 & 0 & 0
            \end{pmatrix}\\
        \Lambda_4 = {\color{blue}\frac{1}{16}} \begin{pmatrix}
            0 & 0 & 0 & {\color{blue}4} & 0 \\ 0 & 0 & 0 & 0 & 0 \\ 0 & 0 & 0 & 0 & 0 \\ {\color{blue}4} & 0 & 0 & 0 & 0 \\ 0 & 0 & 0 & 0 & {\color{blue}8}
            \end{pmatrix},\quad
        \Lambda_5 = {\color{blue}\frac{1}{16}} \begin{pmatrix}
            0 & 0 & 0 & 0 & {\color{blue}8} \\ 0 & 0 & 0 & 0 & 0 \\ 0 & 0 & 0 & 0 & 0 \\ 0 & 0 & 0 & 0 & 0 \\ {\color{blue}8} & 0 & 0 & 0 & 0
            \end{pmatrix}.
    \end{align*}
\end{exa}

\begin{lem}\label{lem:distKravtsov}
For $n \geq 3$, it holds that $\dist\big( 0, \conv(\Gamma_{n,3}) \big) \leq 2^{-n+1}$.
\end{lem}

\begin{proof}
Define $\lambda_{i,j,k} \geq 0$ for all $i,j,k \in [n]$ as in Lemma~\ref{lem:Kravtsov}. Note that $\sum_{i = 1}^n \eps_i = 0$; thus Lemma~\ref{lem:Kravtsov} implies
	\begin{align*}
	\sum_{i,j,k} \lambda_{i,j,k} (\eps_i,\eps_j,\eps_k) = 0_{3n} \, , \quad \text{equivalently} \quad
	-2^{-n+1}(\eps_1,\eps_1,\eps_1) = \sum_{(i,j,k) \in \mathfrak{W}_n} \lambda_{i,j,k}(\eps_i,\eps_j,\eps_k).
	\end{align*}
Normalizing the latter equation we obtain
	\begin{align*}
	x:= -\frac{1}{c \, 2^{n-1}} (\eps_1,\eps_1,\eps_1) \in \conv(\Gamma_{n,3}), \quad \text{where }\;
	c := \sum_{(i,j,k)\in \mathfrak{W}_n} \lambda_{i,j,k} = n-2^{-n+1} \geq \sqrt{3}.
	\end{align*}
Finally, $\norm{\eps_1}^2 \leq 1$ implies $\norm{x} \leq c^{-1} 2^{-n+1} \sqrt{3} \leq 2^{-n+1}$.
\end{proof}

To finish the proof of Theorem~\ref{thm:MarginTensor}(b) we are left to show $0 \notin \conv(\Gamma_{n,3})$. We actually prove the stronger statement $0 \notin \aff(\Gamma_{n,3})$.

\begin{lem}\label{lem:affineHullKravtsov}
The zero vector is not contained  in the affine hull of $\Gamma_{n,3}$.
\end{lem}

\begin{proof}
For a proof by contradiction we assume $0 \in \aff(\Gamma_{n,3})$.  Then there exist $a_s, b_s, c_s \in \RR$ for $s=2,3,\ldots,n$ such that $\sum_s a_s + b_s + c_s = 1$ and
	\begin{align*}
	\sum_{s=2}^n \big( \, a_s (\eps_s, \eps_1, \eps_s) + b_s (\eps_s, \eps_s, \eps_1) + c_s (\eps_{s-1}, \eps_s, \eps_s) \, \big) = (0_n, 0_n, 0_n) \in (\RR^n)^3.
	\end{align*}
In each of the three $\RR^n$-components we obtain $0_n$ as an affine linear combination of $\eps_1, \ldots,\eps_n$.
Applying Lemma~\ref{lem:convCombEps-i} to the coefficient of $\eps_{s-1}$ in the first component,  respectively to the coefficient of $\eps_s$ in the second and third component yields
	\begin{align}
	a_{s-1} + b_{s-1} + c_{s} &= \frac{1}{n} \quad \text{ for } s=2,3,\ldots,n \label{eq:prpAffHull1} \\
	\text{respectively } \qquad
	b_s+ c_s = a_s + c_s &= \frac{1}{n} \quad \text{ for } s=2,3,\ldots,n  \qquad\qquad \label{eq:prpAffHull2}
 	\end{align}
where we necessarily set $a_1 = b_1 := 0$. Equation~\eqref{eq:prpAffHull1} for $s=2$ is $c_2 = n^{-1}$ and hence $a_2=b_2=0$ by \eqref{eq:prpAffHull2} for $s=2$. But now \eqref{eq:prpAffHull1} for $s=3$ gives $c_3 = n^{-1}$ and we can proceed inductively to conclude $c_s = n^{-1}$ and $a_s = b_s = 0$ for all $s=2,3,\ldots,n$.  This gives the contradiction $1 = \sum_{s=2}^n (a_s + b_s + c_s) = \frac{n-1}{n}$, so we must have $0 \notin \aff(\Gamma_{n,3})$. Another contradiction arises when one applies \cref{lem:convCombEps-i} to the coefficient $\eps_n$ in the first component, which yields $a_n + b_n = n^{-1}$.
\end{proof}

\subsection{$d$-tensors}\label{subsec:dTensors}
In this subsection we show that the margin of $\Omega_{n,d}$ is inverse exponential in $nd$ for $n, d \geq 3$, proving part $(c)$ of \cref{thm:MarginTensor}.

Let us give some intuition for our construction. The main idea is to recycle the construction from the previous subsection for some multiple of $n$, i.e. considering $\mathfrak{W}_{rn}$ for $r \geq 2$. Thereby, the main challenge is to ensure that the constructed subset of $\Omega_{n,d}$ does not contain zero in its convex hull. We can try to extend the elements of $\Omega_{n,3}$ to elements of $\Omega_{n,d}$. One natural idea is duplicate each component $d/3$ times, i.e. when $d=6$ the vector $(\eps_i, \eps_j, \eps_k) \in \Omega_{n,3}$ becomes $(\eps_i, \eps_i, \eps_j, \eps_j, \eps_k, \eps_k) \in \Omega_{n,6}$. However, we need a subset of $\Omega_{n,d}$ with $rn$ many elements to imitate the construction from the previous subsection. We still extend the elements of $\Omega_{n,3}$ in this way, but will additionally ``shift'' and ``twist'' by some functions $\sigma_1, \dots, \sigma_{2r-1} \colon [rn] \to [n]$, so that the elements of our set will look like 
$$
\left( \eps_{\sigma_1(i)}, \ldots, \eps_{\sigma_{d/3}(i)},
	\eps_{\sigma_1(j)}, \ldots, \eps_{\sigma_{d/3}(j)}, \eps_{\sigma_1(k)}, \ldots, \eps_{\sigma_{d/3}(k)} \right)
  $$
for $d/3 = 2r-1$ and $(i,j,k)$ in $\mathfrak{W}_{rn}$. We now set about choosing the functions $\sigma_k$. For this, let $n \geq 3$ and fix a natural number $r \geq 2$. It is convenient to use an \emph{adjusted} modulo $n$ function $\mathrm{mod}' \;\; n$ that takes values in $[n]$, i.e. instead of zero it outputs $n$.
For $i \in [r]$ we consider
	\begin{align*}
	&\sigma_{i} \colon [r n] \to [n], \quad j \mapsto \left\lceil \frac{j + (i-1)}{r} \right\rceil \quad \mathrm{mod}' \;\; n \\
	&\sigma_{r + i} := \sigma_1 \circ (r - i + 1 \quad r + 1) \colon [rn] \to [n]
	\end{align*}
where $(r - i + 1 \quad r + 1)$ denotes the corresponding transposition in the symmetric group of $[rn]$.\footnote{We stress that we always take $\sigma_1$ (and \emph{not} $\sigma_i$) to define $\sigma_{r+i}$.} We only need the first $2 r - 1$ of these functions and combine them to obtain 
	\begin{align*}
	\sigma \colon [r n] \to [n]^{2 r - 1}, \quad j \mapsto \big( \sigma_1(j), \sigma_2(j),\ldots, \sigma_{2 r - 1}(j) \big).
	\end{align*}

\begin{exa}\label{exa:sigmaCase3}
For $r = 3$ the functions $\sigma_1, \sigma_2, \ldots, \sigma_6$ are sketched by the following table.
	\begin{center}
	\begin{tabular}{|c||c|c|c|c|c|c|c|c|c|c|c|c|c|}
	\hline 
	$j$ & $1$ & $2$ & $3$ & $4$ & $5$ & $6$ & $\cdots$ & $3n -5$ & $3n - 4$ & $3n-3$ & $3n -2$ & $3n - 1$ & $3n$ \\ 
	\hline \hline
	$\sigma_1$ & \cellcolor{cyan}$1$ & \cellcolor{cyan}$1$ & \cellcolor{cyan}$1$ & \cellcolor{Yellow}$2$ & \cellcolor{Yellow}$2$ & \cellcolor{Yellow}$2$ & $\cdots$ & \cellcolor{YellowGreen}$n-1$ & \cellcolor{YellowGreen}$n-1$ & \cellcolor{YellowGreen}$n-1$ & \cellcolor{YellowOrange}$n$ & \cellcolor{YellowOrange}$n$ & \cellcolor{YellowOrange}$n$ \\
	\hline 
	$\sigma_2$ & \cellcolor{cyan}$1$ & \cellcolor{cyan}$1$ & \cellcolor{Yellow}$2$ & \cellcolor{Yellow}$2$ & \cellcolor{Yellow}$2$ & \cellcolor{Tan!70}$3$ & $\cdots$ & \cellcolor{YellowGreen}$n-1$ & \cellcolor{YellowGreen}$n-1$ & \cellcolor{YellowOrange}$n$ & \cellcolor{YellowOrange}$n$ & \cellcolor{YellowOrange}$n$ & \cellcolor{cyan}$1$ \\ 
	\hline 
	$\sigma_3$ & \cellcolor{cyan}$1$ & \cellcolor{Yellow}$2$ & \cellcolor{Yellow}$2$ & \cellcolor{Yellow}$2$ & \cellcolor{Tan!70}$3$ & \cellcolor{Tan!70}$3$ & $\cdots$ & \cellcolor{YellowGreen}$n-1$ & \cellcolor{YellowOrange}$n$ & \cellcolor{YellowOrange}$n$ & \cellcolor{YellowOrange}$n$ & \cellcolor{cyan}$1$ & \cellcolor{cyan}$1$ \\ 
	\hline \hline
	$\sigma_4$ & \cellcolor{cyan}$1$ & \cellcolor{cyan}$1$ & \cellcolor{Yellow}$2$ & \cellcolor{cyan}$1$ & \cellcolor{Yellow}$2$ & \cellcolor{Yellow}$2$ & $\cdots$ & \cellcolor{YellowGreen}$n-1$ & \cellcolor{YellowGreen}$n-1$ & \cellcolor{YellowGreen}$n-1$ & \cellcolor{YellowOrange}$n$ & \cellcolor{YellowOrange}$n$ & \cellcolor{YellowOrange}$n$ \\  
	\hline 
	$\sigma_5$ & \cellcolor{cyan}$1$ & \cellcolor{Yellow}$2$ & \cellcolor{cyan}$1$ & \cellcolor{cyan}$1$ & \cellcolor{Yellow}$2$ & \cellcolor{Yellow}$2$ & $\cdots$ & \cellcolor{YellowGreen}$n-1$ & \cellcolor{YellowGreen}$n-1$ & \cellcolor{YellowGreen}$n-1$ & \cellcolor{YellowOrange}$n$ & \cellcolor{YellowOrange}$n$ & \cellcolor{YellowOrange}$n$ \\ 
	\hline
	$\sigma_6$ & \cellcolor{Yellow}$2$ & \cellcolor{cyan}$1$ & \cellcolor{cyan}$1$ & \cellcolor{cyan}$1$ & \cellcolor{Yellow}$2$ & \cellcolor{Yellow}$2$ & $\cdots$ & \cellcolor{YellowGreen}$n-1$ & \cellcolor{YellowGreen}$n-1$ & \cellcolor{YellowGreen}$n-1$ & \cellcolor{YellowOrange}$n$ & \cellcolor{YellowOrange}$n$ & \cellcolor{YellowOrange}$n$ \\ 
	\hline 
	\end{tabular} 
	\end{center}
For $r = 3$ and $n = 5$ the functions $\sigma_1, \sigma_2, \ldots, \sigma_6$ are given by the following table.    
    \begin{center}
	\begin{tabular}{|*{16}{>{\centering\arraybackslash}p{15pt}|}}
	\hline 
	$j$ & $1$ & $2$ & $3$ & $4$ & $5$ & $6$ & $7$ & $8$ & $9$ & $10$ & $11$ & $12$ & $13$ & $14$ & $15$ \\ 
	\hline \hline
	$\sigma_1$ & \cellcolor{cyan}$1$ & \cellcolor{cyan}$1$ & \cellcolor{cyan}$1$ & \cellcolor{Yellow}$2$ & \cellcolor{Yellow}$2$ & \cellcolor{Yellow}$2$ & \cellcolor{Tan!70}$3$ & \cellcolor{Tan!70}$3$ & \cellcolor{Tan!70}$3$ & \cellcolor{YellowGreen}$4$ & \cellcolor{YellowGreen}$4$ & \cellcolor{YellowGreen}$4$ & \cellcolor{YellowOrange}$5$ & \cellcolor{YellowOrange}$5$ & \cellcolor{YellowOrange}$5$ \\
	\hline 
	$\sigma_2$ & \cellcolor{cyan}$1$ & \cellcolor{cyan}$1$ & \cellcolor{Yellow}$2$ & \cellcolor{Yellow}$2$ & \cellcolor{Yellow}$2$ & \cellcolor{Tan!70}$3$ & \cellcolor{Tan!70}$3$ & \cellcolor{Tan!70}$3$ & \cellcolor{YellowGreen}$4$ & \cellcolor{YellowGreen}$4$ & \cellcolor{YellowGreen}$4$ & \cellcolor{YellowOrange}$5$ & \cellcolor{YellowOrange}$5$ & \cellcolor{YellowOrange}$5$ & \cellcolor{cyan}$1$ \\ 
	\hline 
	$\sigma_3$ & \cellcolor{cyan}$1$ & \cellcolor{Yellow}$2$ & \cellcolor{Yellow}$2$ & \cellcolor{Yellow}$2$ & \cellcolor{Tan!70}$3$ & \cellcolor{Tan!70}$3$ & \cellcolor{Tan!70}$3$ & \cellcolor{YellowGreen}$4$ & \cellcolor{YellowGreen}$4$ & \cellcolor{YellowGreen}$4$ & \cellcolor{YellowOrange}$5$ & \cellcolor{YellowOrange}$5$ & \cellcolor{YellowOrange}$5$ & \cellcolor{cyan}$1$ & \cellcolor{cyan}$1$ \\ 
	\hline \hline
	$\sigma_4$ & \cellcolor{cyan}$1$ & \cellcolor{cyan}$1$ & \cellcolor{Yellow}$2$ & \cellcolor{cyan}$1$ & \cellcolor{Yellow}$2$ & \cellcolor{Yellow}$2$ & \cellcolor{Tan!70}$3$ & \cellcolor{Tan!70}$3$ & \cellcolor{Tan!70}$3$ & \cellcolor{YellowGreen}$4$ & \cellcolor{YellowGreen}$4$ & \cellcolor{YellowGreen}$4$ & \cellcolor{YellowOrange}$5$ & \cellcolor{YellowOrange}$5$ & \cellcolor{YellowOrange}$5$ \\
	\hline 
	$\sigma_5$ & \cellcolor{cyan}$1$ & \cellcolor{Yellow}$2$ & \cellcolor{cyan}$1$ & \cellcolor{cyan}$1$ & \cellcolor{Yellow}$2$ & \cellcolor{Yellow}$2$ & \cellcolor{Tan!70}$3$ & \cellcolor{Tan!70}$3$ & \cellcolor{Tan!70}$3$ & \cellcolor{YellowGreen}$4$ & \cellcolor{YellowGreen}$4$ & \cellcolor{YellowGreen}$4$ & \cellcolor{YellowOrange}$5$ & \cellcolor{YellowOrange}$5$ & \cellcolor{YellowOrange}$5$ \\
	\hline
	$\sigma_6$ & \cellcolor{Yellow}$2$ & \cellcolor{cyan}$1$ & \cellcolor{cyan}$1$ & \cellcolor{cyan}$1$ & \cellcolor{Yellow}$2$ & \cellcolor{Yellow}$2$ & \cellcolor{Tan!70}$3$ & \cellcolor{Tan!70}$3$ & \cellcolor{Tan!70}$3$ & \cellcolor{YellowGreen}$4$ & \cellcolor{YellowGreen}$4$ & \cellcolor{YellowGreen}$4$ & \cellcolor{YellowOrange}$5$ & \cellcolor{YellowOrange}$5$ & \cellcolor{YellowOrange}$5$ \\
	\hline 
	\end{tabular} 
	\end{center}
\end{exa}

\begin{rem}\label{rem:Sigma}
By construction, each element of $[n]$ is attained exactly $r$-times by $\sigma_k$, $k \in [2 r- 1]$. Moreover, the definition of $\sigma_1, \ldots, \sigma_r$ yields that $\sigma$ is injective. \end{rem}

For $i,j,k \in [r n]$ we introduce the short-hand
	\begin{align*}
	\eps_{\sigma(i)} &:= \left( \eps_{\sigma_1(i)}, \eps_{\sigma_{2}(i)}, \ldots, \eps_{\sigma_{2 r - 1}(i)} \right) \in \left( \RR^n \right)^{2 r - 1} \\
	\eps_{\sigma(i), \sigma(j), \sigma(k)} &:= \left( \eps_{\sigma_1(i)}, \ldots, \eps_{\sigma_{2 r - 1}(i)},
	\eps_{\sigma_1(j)}, \ldots, \eps_{\sigma_{2 r - 1}(j)}, \eps_{\sigma_1(k)}, \ldots, \eps_{\sigma_{2 r - 1}(k)} \right)
	\in \left( \RR^n \right)^{6r - 3}
	\end{align*}
and we set\footnote{One could suggest to consider the set $\lbrace \eps_{\sigma(i), \sigma(j), \sigma(k)} \mid (i,j,k) \in \mathfrak{W}_{r n} \rbrace$, but this still won't ensure that zero is not in the convex hull. The intuition behind is, that $\Gamma_{n,3}$ from the last section is ``nearly at the limit'', i.e. $0 \notin \conv(\Gamma_{n,3})$ but $0 \in \conv(\Gamma_{n,3} \cup \{ (\eps_1,\eps_1,\eps_1) \})$. Now the function $\sigma$ ``introduces $2r-2$ additional linear relations'' as $\eps_{\sigma(i)} \in (\id_n^\perp)^{2r-1}$, since the orthogonal complement $\id_n^\perp \subseteq \RR^n$ has codimension one while $(\id_n^\perp)^{2r-1} \subseteq (\RR^n)^{2r-1}$ has codimension $2r-1$. Thus, it is reasonable to remove $2r-2$ many elements from $\mathfrak{W}_{rn}$.}
	\begin{align*}
	\mathfrak{J}_r := \big\lbrace (s,1,s), (s,s,1) \mid s = 2,3,\ldots, r \big\rbrace \subseteq \ZZ^3.
	\end{align*}
In the following we show that the convex hull of the set
	\begin{align*}
	\Gamma_{n, 6 r - 3} = \big\lbrace \eps_{\sigma(i), \sigma(j), \sigma(k)} \mid (i,j,k) \in \mathfrak{W}_{r n} \setminus \mathfrak{J}_{r} \big\rbrace \subseteq \Omega_{n, 6 r - 3} \subseteq \Big( \big( \RR^n \big)^{2 r - 1} \Big)^3
	\end{align*}
does not contain the zero vector, but is very close to it.

\begin{lem}\label{lem:convStackingKravtsov}
For $n \geq 3$ and $r \geq 2$ it holds that $0 \notin \aff \left( \Gamma_{n, 6 r - 3} \right)$.
\end{lem}

Below we give the proof in the special case $r=3$, in which all main ideas of the general proof become apparent and visible. The proof for the general statement is given in \cref{sec:StackingKravtsovGeneral} and certainly looks technical at a first encounter. Therefore, we strongly suggest that the reader first reads the proof for $r=3$ below.

\begin{proof}[Proof of \cref{lem:convStackingKravtsov} for $r=3$]
For the sake of contradiction assume that $0 \in \aff(\Gamma_{n, 15})$. Then there are coefficients $a_s, b_s, c_s \in \RR$, where $2 \leq s \leq 3 n$, such that $a_2 = a_3 = b_2 = b_3 = 0$, $\sum_s (a_s + b_s + c_s) = 1$ and
	\begin{align}
	\sum_{s= 2}^{3 n} \left( a_s \, \eps_{\sigma(s),\sigma(1),\sigma(s)}
	+ b_s \, \eps_{\sigma(s),\sigma(s),\sigma(1)} + c_s \, \eps_{\sigma(s-1),\sigma(s),\sigma(s)}  \right) = 0 \in (\RR^n)^{15}.\label{eq:main-comboCase3}
	\end{align}
	The bulk of our work will consist of proving the equations
	\begin{align}\label{eq:StackBandCCase3}
	b_2 + c_2 &= b_3 + c_3 = \ldots = b_{3 n} + c_{3 n}\\
	\label{eq:StackAandCCase3} a_2 + c_2 &= a_3 + c_3 = \ldots = a_{3 n} + c_{3 n}.
	\end{align}
	From here we will derive a contradiction. We now set about proving \cref{eq:StackAandCCase3,eq:StackBandCCase3}. Rewrite the left-hand-side of \cref{eq:main-comboCase3} as the collection for $k \in [5]$ of the following affine linear combinations of $\eps_1,\ldots,\eps_n$ in $\RR^n$:
	\begin{align}
	\sum_{s= 2}^{3n} \left( a_s \, \eps_{\sigma_k(s)} 
	+ b_s \, \eps_{\sigma_k(s)} + c_s \, \eps_{\sigma_k(s-1)}  \right) &= 0 \label{eq:StackComp1Case3}\\
	\sum_{s= 2}^{3n} \left( a_s \, \eps_{\sigma_k(1)}
	+ b_s \, \eps_{\sigma_k(s)} + c_s \, \eps_{\sigma_k(s)}  \right) &= 0 \label{eq:StackComp2Case3} \\
	\sum_{s= 2}^{3n} \left( a_s \, \eps_{\sigma_k(s)}
	+ b_s \, \eps_{\sigma_k(1)} + c_s \, \eps_{\sigma_k(s)}  \right) &= 0. \label{eq:StackComp3Case3}
	\end{align}
If we expand each expression as an affine linear combination of the $\eps_l$, then by Lemma~\ref{lem:convCombEps-i} the coefficient of $\eps_l$ must be $n^{-1}$ for all $l \in [n]$. Translating this for equation \eqref{eq:StackComp1Case3} with $k = 2$, $l=2,\ldots,n$ and using \cref{exa:sigmaCase3} we obtain

 \begin{align}
	(a_{m-3} + a_{m-2} + a_{m-1}) + (b_{m-3} + b_{m-2} + b_{m-1})  + (c_{m-2} + c_{m-1} + c_{m}) &= \frac{1}{n} \label{eq:First2}
		\end{align}
for $m = 6,9,12,\dots, 3n$. A similar calculation for $k=1,3$ and $l=2,\ldots,n$ shows \cref{eq:First2} holds for all $5 \leq m \leq 3n +1$, where we set  $c_{3 n + 1} := 0$.

Similarly for \cref{eq:StackComp2Case3} with $l=2,\ldots,n$ and $k=1,2,3$ we obtain for $4 \leq m \leq 3n$ that
	\begin{align}
	(b_{m-2} + c_{m-2})  + (b_{m-1} + c_{m-1}) + (b_{m} + c_m)  &= \frac{1}{n} \label{eq:Second1} 
	\end{align}
and the same equations with ``$b$'' replaced by ``$a$'' when considering \cref{eq:StackComp3Case3}.

In the following we prove \cref{eq:StackBandCCase3}. Subtracting \eqref{eq:Second1} from \eqref{eq:Second1} with values of $m$ differing by one, we deduce that 
\begin{align*}
	b_{2} + c_2 &= b_5 + c_5 = \ldots = b_{3n-1} + c_{3n-1}\\
	b_3 + c_3 &= b_6 + c_6 = \ldots = b_{3n} + c_{3n}, \\
	\qquad \text{and} \qquad b_4 + c_4 &= b_7 + c_7 = \ldots = b_{3n-2} + c_{3n-2}.
	\end{align*}
Next we deduce \cref{eq:StackBandCCase3} by showing $b_2 + c_2 = b_3 + c_3 = b_4 + c_4$. 

To do so, we apply  \cref{lem:convCombEps-i} to \eqref{eq:StackComp2Case3} for the coefficient of $\eps_2$ using \cref{exa:sigmaCase3}, which yields for $k=4,5$ the equations
	\begin{align}
	(b_3 + c_3) + (b_5 + c_5) + (b_6 + c_6)  &= \frac{1}{n}  \label{eq:finishBandC1}\\
	(b_2 + c_2) + (b_5 + c_5) + (b_6  + c_6)  &= \frac{1}{n} \label{eq:finishBandC2}
	\end{align}
respectively. Subtracting the two shows $b_2 + c_2  = b_3 + c_3$, and we have $b_3 + c_3  = b_4 + c_4$ via subtracting \eqref{eq:finishBandC1} from \eqref{eq:Second1} for $m=6$. This completes the proof of \cref{eq:StackBandCCase3}; using \cref{eq:StackComp3Case3} we similarly deduce \cref{eq:StackAandCCase3}.

To get a contradiction we show that $a_s = b_s = c_s = 0$ for all $s = 2,3,\ldots, 3 n$. For this, we set $a := \sum_{s} a_s$ and $b := \sum_s b_s$, and recall that we have defined $a_2 = a_3 = b_2 = b_3 = 0$. This time we use \cref{lem:convCombEps-i} applied to the coefficient of $\eps_1$ in \eqref{eq:StackComp1Case3}, in \eqref{eq:StackComp2Case3} and in \eqref{eq:StackComp3Case3} respectively for $k=1$ to get
	\begin{equation}\label{eq:Case3contradiction}
	c_2 + c_3 + c_4 = \frac{1}{n}, \qquad
	a + c_2 + c_3 = \frac{1}{n}  \qquad \text{ and } \qquad
	b + c_2 + c_3 = \frac{1}{n}
	\end{equation}
respectively. We deduce from these three equations that $a=b=c_4$. Furthermore, $b_2 = b_3 = 0$ shows that \eqref{eq:Second1} for $m=4$ is $b_4 + (c_2 + c_3 + c_4) = n^{-1}$. Subtracting from the latter the left-hand equation in \eqref{eq:Case3contradiction} yields $b_4 = 0$. Similarly, $a_4 = 0$ follows from $a_2 = a_3 = 0$ and the analogous equation of \eqref{eq:Second1} with $a$'s replaced by $b$'s.

Now, \eqref{eq:First2} for $m=5$ simplifies to $c_3 + c_4 + c_5 = n^{-1}$. Thus, $c_2 = c_5$ with \eqref{eq:Case3contradiction} and therefore $a_5 = b_5 = 0$ by \eqref{eq:StackBandCCase3}, \eqref{eq:StackAandCCase3} and $a_2=b_2=0$. This simplifies \eqref{eq:First2} for $m=6$ to $c_4 + c_5 + c_6 = n^{-1}$. Hence, $c_3 = c_6$ as we also have $c_3 + c_4 + c_5 = n^{-1}$ and we get via \eqref{eq:StackBandCCase3} and \eqref{eq:StackAandCCase3} that $a_6 = b_6 = 0$. The latter in turn shows that \eqref{eq:First2} for $m=7$ becomes $c_5 + c_6 + c_7 = n^{-1}$, so $c_4 = c_7$ and $a_7 = b_7 = 0$ by, again, \eqref{eq:StackBandCCase3} and \eqref{eq:StackAandCCase3}.

It should have become apparent that we can proceed inductively in the same manner with \eqref{eq:First2} for $m=5,\ldots,3n+1$; thereby using \eqref{eq:StackBandCCase3} and \eqref{eq:StackAandCCase3} to deduce $a_s = b_s = 0$ for all $s=2,3,\ldots,3n$. In particular, $a = b = c_4 = 0$. Finally, \cref{eq:StackBandCCase3} implies $c_4 = c_s$ for all $s=2,3,\ldots,3n$, which gives the desired contradiction.
\end{proof}

We finish the proof of part $(c)$ of \cref{thm:MarginTensor} by showing the following Lemma.

\begin{lem}\label{lem:distStackingKravtsov}
Let $n \geq 3$ and $r \geq 2$. Then 
	\begin{align*}
	\dist \big(0, \conv(\Gamma_{n,6 r - 3}) \big) \leq \frac{\sqrt{6}}{(n-1)\sqrt{r}} \; 2^{-r(n-1) + 1} \leq 2^{- r (n-1) + 1} .
	\end{align*}
\end{lem}

\begin{proof}
We set $N := r n$ and for $i,j,k \in [N]$ we set $\lambda_{i,j,k}$ as in \cref{lem:Kravtsov} applied for the dimension $N$. Then \cref{eq:marginals-krav} of \cref{lem:Kravtsov} yields
	\begin{align*}
	&\sum_{i,j,k=1}^N \lambda_{i,j,k} \, \left( \eps_{\sigma(i)}, \eps_{\sigma(j)}, \eps_{\sigma(k)} \right)\\
	= &\sum_{i,j,k=1}^N \lambda_{i,j,k} \, \left( \eps_{\sigma(i)}, 0, 0 \right)+ \sum_{i,j,k=1}^N \lambda_{i,j,k} \, \left( 0, \eps_{ \sigma(j)}, 0 \right) + \sum_{i,j,k=1}^N \lambda_{i,j,k} \, \left( 0, 0, \eps_{\sigma(k)} \right) \\
	= &\sum_{i=1}^N \left( \eps_{\sigma(i)}, 0, 0 \right)+ \sum_{j=1}^N  \left( 0, \eps_{ \sigma(j)}, 0 \right) + \sum_{k=1}^N \left( 0, 0, \eps_{\sigma(k)} \right) 
	= \sum_{i=1}^N \eps_{\sigma(i),\sigma(i),\sigma(i)} = 0 \in \left(\RR^n \right)^{6 r - 3},
	\end{align*}
where we used in the last step equation~\eqref{eq:SumEps-i} and \cref{rem:Sigma}, i.e. that each element of $[n]$ is attained exactly $r$-many times by all $\sigma_k \colon [r n] \to [n]$, $k \in [2 r - 1]$. Because $\mathfrak{W}_{N}$ contains the support of $\lambda$ apart from the element $(1,1,1)$, we have
	\begin{align}
	\sum_{(i,j,k) \in \mathfrak{W}_{N} \setminus \mathfrak{J}_{r}} \lambda_{i,j,k} \, \eps_{\sigma(i), \sigma(j), \sigma(k)} =  - \lambda_{1,1,1} \, \eps_{\sigma(1), \sigma(1), \sigma(1)} -  \sum_{(i,j,k) \in \mathfrak{J}_{r}} \lambda_{i,j,k} \, \eps_{\sigma(i), \sigma(j), \sigma(k)} =: x \in \left( \RR^n \right)^{6 r - 3},\label{eq:x-stacked}
	\end{align}
which is an element in the positive cone of $\Gamma_{n,6 r - 3} = \lbrace \eps_{\sigma(i), \sigma(j), \sigma(k)} \mid (i,j,k) \in \mathfrak{W}_N \setminus \mathfrak{J}_{r} \rbrace$.
Normalizing the latter equation with
	\begin{align*}
	c := \sum_{(i,j,k) \in \mathfrak{W}_{N} \setminus \mathfrak{J}_{r}} \lambda_{i,j,k} 
	= \sum_{i,j,k=1}^N \lambda_{i,j,k} - \left( \lambda_{1,1,1} + \sum_{(i,j,k)\in \mathfrak{J}_{r}} \lambda_{i,j,k} \right) \geq N-1
	\end{align*}
shows $c^{-1} x \in \conv(\Gamma_{n,6 r - 3})$. To bound the norm of $c^{-1} x$ we compute
	\begin{align*}
	\lambda_{1,1,1} + \sum_{(i,j,k) \in \mathfrak{J}_{r}} \lambda_{i,j,k} &= 2^{-N+1} + \sum_{s=2}^r \left( \lambda_{s,1,s} + \lambda_{s,s,1} \right) \\
	&= 2^{-N+1} + \sum_{s=2}^r \left( 2^{-N+s-1} + 2^{-N+s-1} \right)
	= \sum_{s=1}^r 2^{-N+s} < 2^{-N + r + 1}.
	\end{align*}
Finally, using $\| \eps_{i_1, i_2, \ldots, i_{6 r  -3}} \| \leq \sqrt{6 r - 3}$ for any $i_1, i_2, \ldots,i_{6r - 3} \in [n]$ together with the triangle inequality on \cref{eq:x-stacked} implies
	\begin{align*}
	\| c^{-1} x\| \leq \frac{\sqrt{6 r - 3}}{N - 1} \; 2^{-N + r + 1}
	\leq \frac{\sqrt{6}}{(n-1)\sqrt{r}} \; 2^{-N + r + 1} \leq 2^{-N + r + 1} = 2^{- r (n-1) + 1},
	\end{align*}
where we used $n \geq 3$ and $r \geq 2$ for $\sqrt{6} \leq (n-1) \sqrt{r}$.
\end{proof}

\subsection{Polynomial scaling}
A simple example of \cref{eq:abelian} is the minimization of an $n$-variate homogeneous polynomial of degree $d$ with nonnegative coefficients over the set $x_1, \dots, x_n > 0$, $\prod x_i = 1$, as studied in \cite{gurvits2004combinatorial}. In this case the sets $\conv(S)$ for $S \subseteq \Omega$ are Newton polytopes of homogeneous polynomials, and the minimum of a polynomial is bounded below if and only if the Newton polytope contains $\frac{d}{n} \id_n$. If the polynomials are hyperbolic of degree $n$, as in \cite{gurvits2004combinatorial}, their Newton polytope either contains $ \id_n$ or is at least $1/\sqrt{n}$ away from it. However, we show that for general homogeneous polynomials the margin can get exponentially small in $n$ even for $d = 3$. 

Minimizing a degree $d$ homogeneous polynomial $\sum_{\alpha \in \ZZ_{\geq 0}^n}p_\alpha x^{\alpha}$ with nonnegative coefficients over the set $x_1, \dots, x_n > 0$, $\prod x_i = 1$ is the same as computing \cref{eq:abelian} for 
	\begin{align}
	\Omega' := \left\lbrace - \alpha + \frac{d}{n} \id_n \; \bigg\vert \; \alpha \in (\ZZ_{\geq 0})^n \text{ with } \vert \alpha \vert = d \right\rbrace.\label{eq:poly-scaling}
	\end{align}
If $n = dm$ for some integer $m \geq 1$, then we have $- \Omega_{m,d} \subseteq \Omega'$. Therefore, \cref{thm:MarginTensor}(b) and (c) and the padding from \cref{subsec:extension} directly yield the following.

\begin{cor}[Margin for Polynomial scaling]\label{cor:dFormsWeightMargin}
Fix some $d \geq 3$ and assume $n = dm$ for some $m \geq 3$. Let $\Omega'$ be as in \cref{eq:poly-scaling}. Then
	\begin{align*}
	\gamma(\Omega') \leq \gamma(\Omega_{m,d}) \leq 2^{-m + 1} = 2^{-\frac{n}{d} + 1}.
	\end{align*}
and for $d \geq 9$ we even have
	\begin{align*}
	\gamma(\Omega') \leq \gamma(\Omega_{m,d}) \leq 2^{- \left\lfloor \frac{(m-1)(d+3)}{6} \right\rfloor + 1} \approx 2^{-\frac{n}{6}}.
	\end{align*}
\end{cor}
Thus, for fixed $d \geq 3$ and $n \to \infty$ the margin of $\Omega'$ can be exponentially small in $n$. In terms of polynomials, this states that the Newton polytope of a degree $d \geq 3$ homogeneous polynomial can be exponentially close to the origin without containing it.


\section{Diameter bounds in the commutative case}\label{sec:diameterComm}

In this section we describe an array such that all approximate scalings are very ill conditioned, proving \cref{thm:diameter}. Let us define the diameter bound.
\begin{dfn}\label{dfn:diameterBound} Let $\eps \to 0$ and $f:\RR^m \to \RR$. The \emph{diameter bound} $D_f(\eps)$ is defined as the infimum over $R > 0$ such that 
$$ \inf_{\|x\| \leq R} f(x) \leq \eps + \inf_{x \in \RR^m}f(x).$$
\end{dfn}
Thus, \cref{thm:diameter} is equivalent to the statement that $D_f(\eps) = \Omega(2^{n/3} \log (1/\eps)$ for $\eps \leq e^{- C n^2 \log n}$. We now give a proof outline for \cref{thm:diameter}. 

\subsection{Proof outline}\label{subsec:outline}

The high-level intuition applies not only to array scaling but to the capacity in general. Recall that the array scaling capacity is 
$$ \inf_{x \in \RR^{3n}}  \sum_{\omega \in \Omega} p_{\omega} e^{\omega \cdot x} $$ for $\Omega =  \Omega_{n,3} = \{e_i - \frac1n \id_n: i \in [n]\} \subseteq \RR^{3n}$. We build both the support $\Omega'\subseteq \Omega_{n,3}$ and the entries $p$ in the following way. We construct a set $\Omega_0 \subseteq \Omega_{n,3}$, another element $\omega \in \Omega_{n,3}$, and an array $q$ with the following properties.
\begin{enumerate}
\item The set $\Omega_0 \subseteq \Omega_{n,3}$ should be the support of a tristochastic array $q$. 
\item The affine hull of $\Omega_0$, should have codimension one\footnote{This will not quite apply in our setting, because $\aff(\Omega_{n,3})$ is not full-dimensional. Instead, $\aff (\Omega_0)$ will be codimension one in $\aff(\Omega_{n,3})$.} in $\RR^{3n}.$ 
\item The origin is in the relative interior of $\conv(\Omega_0)$. Note that the origin is already in $\conv(\Omega_0)$ by the tristochasticity of $q$.
\item The vector $\omega \in \Omega_{n,3}$ should be at a very small, but positive, distance $\eta$ from $\aff (\Omega_0)$. Note that this already implies that the facet gap of $\Omega_0 \cup \omega$ is small.
\end{enumerate}
Finally, we define the entries of $p$ by $p|_{\Omega_0} = \frac{1}{2} q$, $p_\omega = \frac{1}{2}$, and $p_\omega = 0$ elsewhere. Assuming we have found $p$ according to this process, we now give intuition for the diameter bound.

Let $v$ be the projection of $\omega$ to the orthogonal complement of $\aff(\Omega_0)$. Intuitively, the capacity is only approximately attained by vectors very far in the $-v$ direction. Indeed, first note that $\capa(p) = 1/2$, because $\capa(q) = 1$ by tristochasticity, $\capa(p) \geq \frac{1}{2} \capa(q) = \frac{1}{2}$, and $f_p(- tv/\|v\|) = \frac{1}{2} + e^{ - \eta t}$ so $f_p(- tv/\|v\|)$ tends to $\frac 12$. However, $f_p( - t v/\|v\|)$ tends to $\frac 12$ slowly if $\eta$ is small. Indeed, $f_p( - t v/\|v\|) \leq \frac{1}{2}(1 + \eps)$ only if $t \geq \frac{1}{\eta} \log (1/\eps)$. 

To conclude rigorously that the capacity is only approached by vectors very far in the $-v$ direction, we must rule out directions with nonzero components in $\aff(\Omega_0)$. For this, we must use the assumption that $0$ is rather deep in the relative interior of $\conv(\Omega_0)$. If this is the case, then any $\eps$-approximate minimizer must have a bounded component in $\aff(\Omega_0)$, for otherwise the contribution to $f_p$ from the elements of $\Omega_0$ alone will be larger than $\frac{1}{2} + \eps$. 

The remainder of the section will be concerned with the construction of a subset $\Omega_0$, an array $q$, and  an element $\omega$ with these properties.

\subsection{The construction}\label{subsec:diameter-constr}

We construct the subset $\Omega_0$ from a directed graph $D$ on $[n]$, which we will determine later. If $i,j$ is an edge in $D$, then $\Omega_0$ includes the elements $(\eps_i, \eps_i, \eps_j)$ as well as the three cyclic permutations of it. That is, 
$$ \Omega_0 = \{(\eps_j, \eps_i, \eps_i), (\eps_i, \eps_j, \eps_i), (\eps_i, \eps_i, \eps_j): ij \in E(D)\}.$$
We now describe the graph, as seen in \cref{fig:graph}.

\begin{figure}
\centering
\begin{tikzpicture}
\GraphInit[vstyle=Classic]
\useasboundingbox (0,0) rectangle (15.0cm,15.0cm);
\Vertex[L=\hbox{$\overline{w}_{l-1}$},x=7.468cm,y=1.6336cm]{v0}
\Vertex[L=\hbox{$\overline{w}_{l}$},x=8.1205cm,y=0.0cm]{v1}
\Vertex[L=\hbox{$r$},x=7.3452cm,y=11.4526cm]{v2}
\Vertex[L=\hbox{$u_1$},x=9.6222cm,y=12.531cm]{v3}
\Vertex[L=\hbox{$u_{l-1}$},x=13.5892cm,y=14.2394cm]{v4}
\Vertex[L=\hbox{$u_{l-2}$},x=11.7656cm,y=13.4716cm]{v5}
\Vertex[L=\hbox{$u_{l}$},x=15.0cm,y=14.8352cm]{v6}
\Vertex[L=\hbox{$v_1$},x=5.2927cm,y=12.7084cm]{v7}
\Vertex[L=\hbox{$v_{l-1}$},x=1.4984cm,y=14.4533cm]{v8}
\Vertex[L=\hbox{$v_{l-2}$},x=3.3291cm,y=13.6941cm]{v9}
\Vertex[L=\hbox{$v_{l}$},x=0.0cm,y=15.0cm]{v10}
\Vertex[L=\hbox{$w_1$},x=7.1865cm,y=8.8119cm]{v11}
\Vertex[L=\hbox{$w_{l-1}$},x=5.1617cm,y=2.089cm]{v12}
\Vertex[L=\hbox{$w_{l-2}$},x=6.5507cm,y=3.5411cm]{v13}
\Vertex[L=\hbox{$w_{l-3}$},x=6.9099cm,y=6.2032cm]{v14}
\Vertex[L=\hbox{$w_{l}$},x=4.1288cm,y=0.8072cm]{v15}
\Edge[label=\hbox{$2 + 1$}, labelstyle={sloped,below,fill=\thepagecolornone}](v0)(v1)
\Edge[label=\hbox{$2 - \frac{1}{2}$}, labelstyle={sloped,below,fill=\thepagecolornone}](v0)(v13)
\Edge[label=\hbox{$2 + (-\frac{1}{2})^{l-1}$}, labelstyle={scale=0.8,sloped,above,fill=\thepagecolornone}](v2)(v3)
\Edge[label=\hbox{$2 + (-\frac{1}{2})^{l-1}$}, labelstyle={scale=0.8,sloped,below,fill=\thepagecolornone}](v2)(v7)
\Edge[label=\hbox{$2 + (-\frac{1}{2})^{l-2}$}, labelstyle={scale=0.8,sloped,above,fill=\thepagecolornone}](v2)(v11)
\Edge[style=dotted](v3)(v5)
\Edge[label=\hbox{$2 - \frac{1}{2}$}, labelstyle={sloped,above,fill=\thepagecolornone}](v4)(v5)
\Edge[label=\hbox{$2 + 1$}, labelstyle={sloped,above,fill=\thepagecolornone}](v4)(v6)
\Edge[style=dotted](v7)(v9)
\Edge[label=\hbox{$2 - \frac{1}{2}$}, labelstyle={sloped,below,fill=\thepagecolornone}](v8)(v9)
\Edge[label=\hbox{$2 + 1$}, labelstyle={sloped,below,fill=\thepagecolornone}](v8)(v10)
\Edge[style=dotted](v11)(v14)
\Edge[label=\hbox{$2 - \frac{1}{2}$}, labelstyle={sloped,above,fill=\thepagecolornone}](v12)(v13)
\Edge[label=\hbox{$2 + 1$}, labelstyle={sloped,above,fill=\thepagecolornone}](v12)(v15)
\Edge[label=\hbox{$2 -\frac{1}{2}$}, labelstyle={sloped,above,fill=\thepagecolornone}](v13)(v14)
\end{tikzpicture}

\caption{The graph $D_l$ from \cref{dfn:graph} with the edge labels proportional to the edge labeling $q$ in \cref{it:stoch} of \cref{lem:diameter} (the constant factor $1/6n$ is omitted for readability). We have also omitted the directions, which are all towards the root $r$.
}\label{fig:graph}
\end{figure}

\begin{figure}
\centering
$\left[\begin{array}{c|c|c|c|c}
\scalebox{1.5}0 &
\begin{array}{ccccc}
\cellcolor{blue!30}\mathrlap{A}{\phantom{\ddots}} & \cellcolor{green!30} I & \hdots & 0 \\
\mathrlap{0}{\phantom{\ddots}} & \cellcolor{blue!30}\ddots & \cellcolor{green!30}\ddots & \smash{\vdots} \\
\mathrlap{\vdots}{\phantom{\ddots}} & \ddots & \cellcolor{blue!30}\ddots & \cellcolor{green!30}I \\
\mathrlap{0}{\phantom{\ddots}} & \smash{\hdots} & \mathrlap{0}{\phantom{\ddots}} & \cellcolor{blue!30}\mathrlap{A}{\phantom{\ddots}} 
\end{array}
& \scalebox{1.5}0 & \scalebox{1.5}0 & 
\begin{array}{c}
\mathrlap{0}{\phantom{\ddots}} \\
\mathrlap{\vdots}{\phantom{\ddots}}  \\
\mathrlap{\vdots}{\phantom{\ddots}}  \\
\cellcolor{green!30}\mathrlap{I}{\phantom{\ddots}} 
\end{array}
\\
\hline
\scalebox{1.5}0 & \scalebox{1.5}0 & 
\begin{array}{ccccc}
\cellcolor{blue!30}\mathrlap{A}{\phantom{\ddots}} & \cellcolor{green!30} I & \hdots & 0 \\
\mathrlap{0}{\phantom{\ddots}} & \cellcolor{blue!30}\ddots & \cellcolor{green!30}\ddots & \smash{\vdots} \\
\mathrlap{\vdots}{\phantom{\ddots}} & \ddots & \cellcolor{blue!30}\ddots & \cellcolor{green!30}I \\
\mathrlap{0}{\phantom{\ddots}} & \smash{\hdots} & \mathrlap{0}{\phantom{\ddots}} & \cellcolor{blue!30}\mathrlap{A}{\phantom{\ddots}} 
\end{array}
& \scalebox{1.5}0 & 
\begin{array}{c}
\mathrlap{0}{\phantom{\ddots}} \\
\mathrlap{\vdots}{\phantom{\ddots}}  \\
\mathrlap{\vdots}{\phantom{\ddots}}  \\
\cellcolor{green!30}\mathrlap{I}{\phantom{\ddots}} 
\end{array}\\
\hline
\scalebox{1.5}0 & \scalebox{1.5}0 & \scalebox{1.5}0 & 
\begin{array}{ccccc}
\cellcolor{blue!30}\mathrlap{A}{\phantom{\ddots}} & \cellcolor{green!30} I & \hdots & 0 \\
\mathrlap{0}{\phantom{\ddots}} & \cellcolor{blue!30}\ddots & \cellcolor{green!30}\ddots & \smash{\vdots} \\
\mathrlap{\vdots}{\phantom{\ddots}} & \ddots & \cellcolor{blue!30}\ddots & \cellcolor{green!30}I \\
\mathrlap{0}{\phantom{\ddots}} & \smash{\hdots} & \mathrlap{0}{\phantom{\ddots}} & \cellcolor{blue!30}\mathrlap{A}{\phantom{\ddots}} 
\end{array} & 
\begin{array}{c}
\mathrlap{0}{\phantom{\ddots}} \\
\mathrlap{\vdots}{\phantom{\ddots}}  \\
\mathrlap{\vdots}{\phantom{\ddots}}  \\
\cellcolor{green!30}\mathrlap{I}{\phantom{\ddots}} 
\end{array}\\
\hline
\begin{array}{cc} \cellcolor{blue!30}\mathrlap{A}{\phantom{\ddots}} & \cellcolor{green!30}\mathrlap{I}{\phantom{\ddots}}  \\
0 & \cellcolor{blue!30}\mathrlap{A}{\phantom{\ddots}} 
\end{array} & \scalebox{1.5}0 & 
\scalebox{1.5}0
& \begin{array}{cccc} \mathrlap{0}{\phantom{\ddots}} & \mathrlap{0}{\phantom{\ddots}} & \mathrlap{0}{\phantom{\ddots}} & \hdots \\
\mathrlap{0}{\phantom{\ddots}} & \mathrlap{0}{\phantom{\ddots}} &  \cellcolor{green!30}\mathrlap{I}{\phantom{\ddots}} & \hdots
\end{array} 
& \begin{array}{c}\mathrlap{0}{\phantom{\ddots}}\\ \mathrlap{0}{\phantom{\ddots}}\end{array}
\end{array}\right]
$
\caption{The matrix $M$ written in the reordered basis described before \cref{lem:diameter}. From the left, the five groups of columns correspond to the $\overline{w}'s$, the $u's$, the $v's$, the $w's$, and $r$ among the vertices of $D_l$. As such the dimensions of the five column groups, from left, are $3 \cdot 2, 3  (l-1), 3  (l-1), 3  (l-1), 3$, and the dimensions of the four groups of rows from top are $3  (l-1), 3  (l-1), 3  (l-1), 3\cdot 2$. $A$ is as in \cref{eq:a-matrix} and $I$ is the $3\times 3$ identity matrix.}\label{fig:matrix}
\end{figure}

\begin{figure}
\centering
\begin{tikzpicture}[
            > = stealth, 
            shorten > = 1pt, 
            auto,
            node distance = 3cm, 
            semithick 
        ]
\GraphInit[vstyle=Classic]
\useasboundingbox (0,-1.5) rectangle (5cm,1.5cm);
\Vertex[L=\hbox{$\;$},x=2.5cm,y=0cm]{v0}
\node at (2.5cm,.5cm) [draw,draw=none] {$v$};
\Vertex[L=\hbox{$\;$},x=5.0cm,y=0cm]{v1}
\Vertex[L=\hbox{$\;$},x=0.0cm,y=0.0cm]{v2}
\path[->] (v2) edge node {$q_1$} (v0);
\path[->] (v0) edge node {$q_2$} (v1);

\end{tikzpicture}

\caption{If $v$ is a vertex of $D_l$ with edges weighted $q_1$ and $q_2$ incident to it, then the column $v,i$ of $M$ for $i \in [3]$ sums to $q_1 + 2q_2$. That is, the incoming edge contributes its weight and the outgoing edge contributes twice its weight.}
\label{fig:column-sum}
\end{figure}

\begin{dfn}\label{dfn:graph}
The graph $D_l = (W,E)$ is a directed tree with $l + 1$ levels, where the root is on the $0^{th}$ level and the leaves are on the $l^{th}$ level. The tree is constructed as follows.
\begin{itemize}
\item All the edges are directed towards the root and are between adjacent levels. 
\item The root has three children, and on the $l-1$ levels below the root every node has one child. 
\item Additionally, one of the vertices on level $l-2$ has an additional child which has its own child. 
\end{itemize}
Explicitly, the vertices $W$ and edges $E$ are given by
\begin{align*}W &= \{u_i, v_i, w_i: i \in [l]\} \cup \{w_0:=u_0:=v_0:=r, \bar{w}_{l-1}, \bar{w}_l\}.\\
 E &= \{u_{i}u_{i-1}, v_{i}v_{i-1}, w_{i}w_{i-1}: i \in [l] \} \cup \{\bar{w}_{l-1} w_{l-2}, \bar{w}_l \bar{w}_{l-1}\}.\end{align*}
\end{dfn}
Note that $D_l$ has $3(l + 1)$ vertices so we set $n = 3(l + 1)$. Thus $D_l$ has $3 l + 2$ edges and so $|\Omega_0| = 3(3l + 2) = 3n - 3.$ It is helpful to construct the matrix $M$ whose set of rows is $\Omega_0$. To make the matrix sparser, first replace $\eps_i$ by $e_i$ by restricting the minimization to the subspace $\sum x_i = \sum y_i = \sum z_i = 0$, which is without loss of generality. We define $\Omega_0' \subseteq \RR^{3n}$ to be $\Omega_0$ but with each $(\eps_i, \eps_j, \eps_k)$ replaced by $(e_i, e_j, e_k)$; define $\Omega'_{n,3}$ similarly and define $p_{(e_i, e_j, e_k)}:=p_{(\eps_i, \eps_j, \eps_k)}$. Then
$$ \inf_{x \in \RR^{3n}} \sum_{\omega \in \Omega_{n,3}} p_{\omega} e^{(\eps_i, \eps_j, \eps_k) \cdot x} =  \inf_{\underset{\sum x_i = \sum y_i = \sum z_i = 0}{x,y,z \in \RR^{n}}} \;\sum_{\omega \in \Omega'_{n,3}} p_{\omega} e^{(e_i, e_j, e_k) \cdot (x,y,z)}.$$
 Moreover, when we write the matrix $M$, it is easier to write the vector $(x,y,z)$ in the order $(x_1,y_1,z_1, x_2,y_2,z_2, \dots )$ instead of the order $(x_1, \dots, x_n, y_1, \dots, y_n, z_1, \dots, z_n)$. With this ordering, the matrix $M$ with rows in $\Omega_0'$ is a block matrix $M$ with blocks of size $3$, with $n-1$ block rows, and with $n$ block columns. Each block row corresponds to an edge in the directed graph $D_l = (W, E)$ on $n = 3(l+1)$ vertices. If $e \in E$ is an edge from $i \to j$, then the $e^{th}$ row of $M$ has the matrix 
\begin{gather}A = \begin{bmatrix} 0 & 1 & 1 \\ 1 & 0 & 1 \\ 1 & 1 & 0 \end{bmatrix} \label{eq:a-matrix}\end{gather}
 in the $i^{th}$ block entry and 
 $$ I = \begin{bmatrix} 1 & 0 & 0 \\ 0 & 1 & 0 \\ 0 & 0 & 1 \end{bmatrix} $$ in the $j^{th}$ block entry and zeroes elsewhere. See \cref{fig:matrix} for a portrayal of the whole matrix $M$.

The first three properties for $\Omega_0$ in the proof plan translate to the following three claims about $M$. The first relates to the tristochasticity of $q$, the second to the codimension of $\aff(\Omega_0')$ in the subspace $(\id_n^\perp)^3$, and the third to the depth of the point $\frac{1}{n}\id_{3n}$ in $\conv(\Omega_0')$.

 \begin{lem}\label{lem:diameter}Let $n = 3(l + 1)$. 
\begin{enumerate}
\item\label{it:stoch}The probability distribution $q$ on $E\times [3]$ defined $($for $i\in [3])$ by 
\begin{align*}\text{ for }j \in [l], \quad q_{u_j u_{j-1}, i} = q_{v_j v_{j -1}, i} &= \frac{1}{6n} \left(2 + (-2)^{- (l - j)}\right) \\
 \text{ for } j \in [l-2], \quad q_{w_j w_{j -1}, i}  &= \frac{1}{6n} \left(2 + (-2)^{- (l - j - 1)}\right) \\
q_{w_{l-1} w_{l -2}, i} = q_{\bar{w}_{l-1} \bar{w}_{l -2}, i} = \frac{1}{2} q_{w_{l} w_{l -1}, i} = \frac{1}{2} q_{\bar{w}_{l} \bar{w}_{l -1}, i} &= \frac{1}{6n} \left(\frac{3}{2}\right)
\end{align*}
on the rows of $M$ has expectation $\frac{1}{n}\id_{3n}$. That is, if the rows of $M$ are scaled by the values of $q$, each column sums to $1/n$. Note that the entries of $q$ are $\Theta(\frac{1}{n})$. Ignoring the index $i$ in $q_{uv,i}$ allows us to view $q$ as a labeling of the edges of the graph $D_l$; see \cref{fig:column-sum,fig:graph}. 
\item\label{it:kernel} $\ker M = \spn (\Omega_0')^\perp$ is spanned by the 2 dimensional space $S \subseteq \RR^{W \times [3]}$ given by
$$ S = \{s: s(v, 1) = \alpha, s(v, 2) = \beta, s(v, 3) = \gamma \text{ for all } v \in W,\; \alpha + \beta + \gamma = 0 \}$$
 and the function $f \in \RR^{W \times [3]}$ which for all $i \in [3]$ assigns
 \begin{align}f (u_j, i) = f(v_j, i) = f(w_j, i) &= (-2)^{-j} \text{ for } j \in [l] \cup \{0\}\nonumber \\
 \textrm{ and } f(\bar{w}_{l-k}, i)  &=  f(w_{l-k}, i) \text{ for } k \in \{0,1\}.\label{eq:kernel} 
 \end{align}
 Note that $f \in (\id_n^\perp)^3 \subseteq S^\perp$. Thus we have the orthogonal decomposition $\spn(\Omega_0')^\perp = S \oplus \spn{f}.$
\item\label{it:singular} Apart from the three zero singular values, all singular values of $M$ are $\Omega(1/n)$.

\end{enumerate}

\end{lem}

Given the lemma, let us prove that the diameter bound holds according to the proof outline at the beginning of the section.

\begin{proof}[Proof of \cref{thm:diameter}] We first show the claim for $n$ of the form $n = 3(l + 1)$; the bound follows for $3(l+1) < n < 3(l + 2)$ by applying \cref{prp:pad-diameter} with $t = 3(l +1)$, using that the array we construct has capacity $1/2$ and $t/n \geq 2/3$. 

We now show the diameter lower bound for $n = 3(l+1)$. It is enough to exhibit a constant $C>0$, and a probability distribution $p$ on $\Omega_{n,3}' = \{e_i: i \in [n]\}^3$ such that for for all $N \geq C n^2 \log n$ and all $x, y, z \in \id_n^\perp$, 
 $$\sum_{\omega \in \Omega_{n,d}'} p_\omega e^{\omega \cdot (x,y,z)} \leq  e^{-N} + \inf_{x', y', z' \in \id_n^\perp} \sum_{\omega \in \Omega_{n,d}'}\; p_\omega e^{\omega \cdot (x,y,z)}$$
 only if $\|(x,y,z)\|_2 = \Omega( 2^{n/3} N)$. Note that the space $(\id_n^\perp)^3$ over which we are infimizing is a subspace of $S^\perp$ where $S$ is as in \cref{lem:diameter}, and that $\Omega_{n,d}' \subseteq S^\perp$. The proof will follow the outline in \cref{subsec:outline}; namely, we will consider a subset $\Omega_0' \subseteq \Omega_{n,d}'$ and an element $\omega' \in \Omega_{n,d}'$ very close to, but outside of, $\aff(\Omega_0')$.

Consider the set $\Omega_0' \subseteq \Omega_{n,d}'$ of rows of $M$ in \cref{lem:diameter} and the probability distribution $q$ on $\Omega_0'$ from \cref{lem:diameter}. Let $\omega' = (e_{u_l}, e_{v_l}, e_{w_l})$ for the vertices $u_l, v_l, w_l \in D_l$.
Let $\Omega = \Omega_0' \cup \{\omega'\}$, and define the probability distribution $p$ on $\Omega$ by $p_{\omega'} = \frac{1}{2}$ and $p_{\omega} = \frac{1}{2} q_{\omega}$ for $\omega \in \Omega_0'$. Recall from \cref{lem:diameter} the orthogonal decomposition $\spn(\Omega_0')^\perp = S \oplus \spn{f}$. As $\RR^{3n} = \spn(\Omega_0') \oplus \spn(\Omega_0')^\perp$, we have the orthogonal decomposition $S^\perp = \spn(\Omega_0') \oplus \spn f$. Observe that $\omega' \not\in \spn(\Omega_0')$, because by \cref{lem:diameter} we have $\spn (\Omega_0')^\perp = \ker M =  S + \spn f$ and clearly $f \cdot \omega' \neq 0$.

By \cref{it:stoch} of \cref{lem:diameter} we have $\sum_{\omega \in \Omega_0'} q_\omega \omega = \frac{1}{n}(\id_n, \id_n, \id_n)$ and thus $\capa(q) = 1$. Therefore, $\omega' \not\in \spn(\Omega_0')$ implies that the infimum is $1/2$ for this choice of $\Omega$ and $p$. We claim that the infimum can only be approximately attained by $h \in (\id_n^\perp)^3$ with a very large component in the one-dimensional space $\spn f = \spn(\Omega_0')^\perp \cap (\id_n^\perp)^3$. As in the proof outline, we must bound the components in $\spn(\Omega_0')$ of the approximate minimizer $h$. For $h \in (\id_n^\perp)^3$ write $h = h_0 + a f$ and $\omega' = \omega_0 + bf$ where $h_0, \omega_0 \in \spn \Omega_0'$. Note that $|b| = \frac{|f \cdot \omega'|}{\|f\|^2} = O(2^{-l}) =  O(2^{-n/3})$ and that $h_0 \in (\id_n^\perp)^3$, because $h$ and $f$ are.
Suppose 
$$ \sum_{\omega \in \Omega} p_{\omega} e^{\omega \cdot h} \leq \frac{1}{2}e^{-N} + \frac{1}{2}.$$
Equivalently, 
\begin{align}
 \sum_{\omega \in \Omega_0'} q_{\omega} e^{\omega \cdot h_0} +  e^{h_0\cdot \omega_0 + a b \norm{f}^2} \leq e^{-N} + 1.\label{eq:almost}
\end{align}
Suppose $\|h_0\|$ is bounded by $L$. If $e^{h_0\cdot \omega_0 + ab \norm{f}^2 } \leq e^{-N}$, then $|a b| =  \Omega(N - L)$. In particular, $\|h\| \geq \|af\| =\vert a b \vert \|f\| / \vert b \vert= \Omega ((N - L)2^{n/3})$ because of the previous bounds on $|ab|, |b|,$ and the fact that $\|f\| = \Theta(1)$. It remains to prove a bound $L$ for $\|h_0\|$. We will do this by showing that if $\|h_0\|$ were too large, then the first term of the left-hand side of \cref{eq:almost} would be too large. This amounts to $\frac{1}{n}\id_{3n}$ being in the relative interior of $\conv(\Omega_0')$, but will be proved using lower bounds on the singular values of $M$.

Let $\alpha$ denote the least nonzero singular value of $M$; by \cref{it:singular} of \cref{lem:diameter} $\alpha = \Omega( 1/n)$. As $h_0 \in \spn(\Omega_0') =\operatorname{rowspan}(M)$, we have $\|M h_0\| \geq \alpha \|h_0\|$ by the singular value bound. We claim that there is some $\omega \in \Omega_0'$ satisfying $\omega \cdot h_0 = \Omega(\alpha \|h_0\|/n)$. To prove this, first note that the $\sum_{\omega \in \Omega_0'} q_{\omega} \omega \cdot h_0 = \frac{1}{n}(\id_n, \id_n, \id_n) \cdot h_0 = 0$ because $h_0 \in  (\id_n^\perp)^3$. Moreover, by \cref{lem:diameter} we have $q_{\omega} = \Theta(1/n)$. The claim follows from \cref{lem:simple} below applied to the sequence $(\omega \cdot h_0: \omega \in \Omega_0').$

Because $q_\omega=\Theta(\frac{1}{n})$, we must have that $\omega \cdot h_0 = O(\log n)$ for all $w \in \Omega_0'$. Else, the contribution from the term $q_\omega e^{\omega \cdot h_0}$ alone is larger than $1$, in which case $x$ cannot be an $e^{- N}$-approximate minimizer. Finally, $\|h_0\| = O(n(\log n)/\alpha) = O(n^{2}\log n),$ and so we may take $L = O(n^{2}\log n)$ and $N \geq 2L$. \end{proof}

In the above proof, we used the following simple lemma.
\begin{lem}\label{lem:simple} Let $0 < \beta < \gamma$. Suppose $z \in \RR^m$ is such that $\sum_{i = 1}^m q_i z_i = 0$ for $q_i \in (\beta/m, \gamma/m)$. Then there exists $i \in [m]$ such that $z_i \geq \frac{\beta }{2\gamma m} \|z\|_2$.
\end{lem}
\begin{proof}
Because $\sum q_i z_i = 0$, 
$$ \sum_{i: z_i < 0} q_i |z_i| = \sum_{i:z_i \geq 0} q_i z_i,$$
and 
$$ \sum_{i: z_i < 0} q_i |z_i| +  \sum_{i:z_i \geq 0} q_i z_i \geq (\beta/m) \|z\|_1 \geq (\beta/m)\|z\|_2.$$
Thus $\sum_{i: z_i \geq 0} q_i |z_i| \geq \frac{\beta}{2m}\|z\|_2,$ so there is some $i$ such that $q_i z_i > \frac{1}{m} \frac{\beta}{2m}\|z\|_2$. Thus $z_i > \frac{\beta}{2\gamma m}\|z\|_2$.
\end{proof}
To show that our diameter lower bound holds for all values of $n$, we need the following proposition, which is proved in \cref{sec:rounding}. The idea is to prove diameter bounds for larger arrays from diameter bounds for smaller ones by embedding the smaller array in a ``corner'' of the larger array.

\begin{prp} \label{prp:pad-diameter}
Suppose $1 \leq t \leq n$. Let $p$ be a $d$-dimensional array in $(\RR_{\geq 0}^t)^{\ot d}$ with unit sum; in particular $\capa(p) \leq 1$. Let $q$ be the $d$-dimensional array in $(\RR_{\geq 0}^n)^{\ot d}$ array such that $q_{i_1, \dots, i_d} = \frac{t}{n} p_{i_1, \dots, i_d}$ for $i_1, \dots, i_d \in [t]$, $q_{iii} = 1/n$ for $t+1 \leq i \leq n$, and $q_{i_1, \dots, i_d} = 0$ otherwise. For $\eps \leq 1 - \capa(p)$, 
$$D_{f_q}(\eps) \geq D_{f_p}\left(\frac{(1 - \capa(p))\eps}{1 - \capa(p)^{t/n}}\right).$$ 
In particular, the norm of any $\eps$-approximate minimizer of $f_q$ is at least the norm of some $\big( \frac{1 - \capa(p)}{1 - \capa(p)^{t/n}} \big) \eps$-approximate minimizer of $f_p$. 
\end{prp}

As a corollary of the proof of \cref{thm:diameter}, we have a bound on the \emph{facet gap} of \cite{burgisser2020interior}. The facet gap of a finite set $\Omega$ is defined to be the least distance of an element of $\Omega$ to the affine hull of a facet of $\conv(\Omega)$. We have shown that the distance between $\aff(\Omega_0')$ and $\omega'$ is $O(2^{-l})$, or $O(2^{- n/3})$. 
\begin{cor}[Facet gap of array scaling] \label{cor:facet-fap}
There is a subset $\Omega_1 \subseteq \Omega_{n,3}$ with facet gap $O(2^{- n/3})$.
\end{cor}

Analogously to what is done for the margin in \cref{prp:dTensorsPadding}, we may also embed this array inside a larger array to obtain a diameter bound for $d \geq 3$. For $d \geq 3$, take $q(i,j,k,l,l,\dots, l) = \frac{1}{n} p_{ijk}$ for all $i,j,k,l \in [n]$. Then for $(x_1, \dots, x_d) \in (\id_n^\perp)^d$ we have 
$$f_q(x_1, \dots, x_d) = \frac{1}{n} f_p(x_1,x_2,x_3)  \sum_{l = 1}^n e^{\sum_{j = 4}^d (x_j)_l}.$$
For fixed $x_1, x_2, x_3$, by Jensen's inequality $f_q$ is minimized when $x_j = 0_n$ for $j \geq 4$ and takes value $f_p(x_1, x_2, x_3)$, and thus $f_q$ has the same diameter bound as $f_p$.

\begin{cor}[Diameter bound for $d \geq 3$]\label{cor:diameter-d}
There is an absolute constant $C > 0$ such that the following holds. For all $d \geq 3$, there is a family of arrays $q \in (\RR_{\geq 0}^n)^{\ot d}$ with $O(n^2)$ nonzero entries, each of bit-complexity $O(n)$, that satisfies the following property. For all $0 <\eps \leq  \exp(- C n^2 \log n)$ and $x \in \RR^{dn}$, if
	$$f_q (x) \leq \capa(p) + \eps$$
	then $\norm{x}_2 = \Omega\left(2^{n/3}\log(1/\eps)\right).$

\end{cor}

\subsection{Proof of the properties of the construction}

 We now prove \cref{lem:diameter}.
 \begin{proof}[Proof of \cref{lem:diameter}] It is first helpful to change basis on each copy of $\RR^3$ so that the $A$ blocks are diagonalized. Let $U \in \Mat(3)$ be an orthogonal matrix such that 
$$U^\dagger A U =   \begin{bmatrix} 2 & 0 & 0 \\ 0 & -1 & 0 \\ 0 & 0 & -1 \end{bmatrix}.$$
This is possible because $2, -1, -1$ are the eigenvalues of the symmetric matrix $A$. In particular, the first column of $U$ is $(1,1,1)/\sqrt{3}$, and the second two columns span the space of vectors with sum zero. Then $M' = (U^{\oplus n})^\dagger M U^{\oplus n}$ is of the form $P \oplus L \oplus L$ where $P_{e,v} = M'_{(e, 1), (v, 1)}$ for $e, v \in E \times V$ and $L_{e, v} = M'_{(e, 2), (v, 2)}$.  Note that $L$ is the edge-vertex incidence matrix of the directed graph $D_l$, the row corresponding to the edge $(u,v)$ of $D_l$ has a $-1$ in the column indexed by the vertex $u$ and a $+1$ in the column indexed by $v$. Moreover, $P$ is the matrix obtained from $L$ by replacing every $-1$ entry by a $2$.

To prove \cref{it:kernel}, observe that $\ker M$ is $(U^{\oplus n}) \ker M' = (U^{\oplus n}) \ker P \oplus \ker L \oplus \ker L$. Because $D_l$ is connected, $\ker L = \spn \id_n$. As the second two columns of $U$ span the subspace of $\RR^3$ of vectors with sum $0$, the two-dimensional space $S$ is given by $(U^{\oplus n}) 0 \oplus \spn \id_n \oplus \spn \id_n = (U^{\oplus n}) 0 \oplus \ker L \oplus \ker L$. We next reason for $\ker P$, the other summand of the orthogonal decomposition of $\ker M'$. The graph $D_l$ is a connected tree, so $\ker P$ is one dimensional. This is because every choice of $g(w_0) \in \RR$ determines a unique function $g:V \to \RR$ in $\ker P$. We claim that the function $g(v) = f(v,1)$ for $f$ as in \cref{eq:kernel} is in $\ker P$, and hence spans it. To check this, one must check that for every edge $(v,w) \in E$ we have $2g(v) + g(w) = 0$. It is instructive to look at \cref{fig:graph}. Observe that this property holds for the edges $u_{k, k-1}$ if the sequence $g(u_k)$ obeys the recurrence relation $g(u_{k-1}) = -2g(u_k)$ for $k \in [l]$, which is indeed true by the definition of $f$. Checking the condition for $v$ and $w$ is similar. As the first column of $U$ is proportional to $\id_3$, $(U^{\oplus n}) \ker P \oplus 0 \oplus 0$ is spanned by the function $f$. This proves \cref{it:kernel}.

To show \cref{it:singular}, it is enough to argue that the singular values of $P, L, L$ obey the desired bound. For $L$ this follows straightforwardly from the fact that $L$ is an incidence matrix of a connected, directed tree and so is totally unimodular with linearly independent rows. The singular value bound follows by \cref{lem:unimod_sing}. Rather than arguing spectrally for $P$, we make an ad-hoc argument using the structure of $D_l$.
We first show that $\|x^t P\|_\infty = \Omega( \|x\|_\infty)$ for all $x \in \RR^{n-1}$, which suffices because $\|x^t P\|_2 \geq  \|x^tP\|_\infty$ and $\|x\|_\infty \geq \frac{1}{\sqrt{n}}\|x\|_2$.

Let $x \in \RR^{n-1} \setminus \lbrace 0 \rbrace$ and $e$ be an edge in $D_l$ such that $\vert x(e) \vert  = \| x \|_\infty$. If $e = u_iu_{i-1}$ for $i \in [l]$, then $|x^t P (u_{i})| \geq \| x \|_\infty$ because either $i < l$ in which case
	\begin{align*}
	\vert x^t P (u_{i }) \vert = \vert 2 x(u_{i}u_{i - 1}) + x(u_{i+1}u_{i}) \vert \geq 2\|x\|_\infty - \vert x(u_{i+1}u_{i}) \vert \geq \|x\|_\infty
	\end{align*}
or $i = l$ and so $\vert x^t P (u_{i}) \vert = \vert 2 x(e) \vert = 2 \|x\|_\infty$. The same argument applies to all other edges except $e = w_{l-2} w_{l-3}$. In the latter case we are done if $x^t P (w_{l-2}) \geq 1/3 \| x \|_\infty$. Otherwise we necessarily have $\vert x(w_{l-1} w_{l-2}) \vert + \vert x(\bar{w}_{l-1} w_{l-2}) \vert \geq 5/3 \|x\|_\infty$, since $x^t P(w_{l-2}) = 2x(e) + x(w_{l-1} w_{l-2}) + x(\bar{w}_{l-1} w_{l-2})$. It follows that $\vert x(w_{l-1} w_{l-2}) \vert \geq 5/3 \|x\|_\infty - \vert x(\bar{w}_{l-1} w_{l-2}) \vert \geq 5/3 \|x\|_\infty - \|x\|_\infty \geq 2/3 \|x\|_\infty$. As $\vert x^t P(w_{l-1}) \vert = 2  x(w_{l-1} w_{l-2})+ x(w_l w_{l-1}) $, we have 
	\begin{align*}
	\vert x^t P(w_{l-1}) \vert \geq 2 \vert x(w_{l-1} w_{l-2}) \vert - \vert x(w_l w_{l-1}) \vert \geq \frac{4}{3} \|x\|_\infty - \|x\|_\infty \geq \frac{1}{3} \|x\|_\infty.
	\end{align*}
In any case, there is some value of $x^t P$ with absolute value greater or equal $1/3 \|x\|_\infty$.

Finally, for \cref{it:stoch} we note that the probability distribution $q$ on the rows of $M$ has expectation equal to the all $1/n$ function if and only if the probability distribution $q'$ defined by $q'_e = 3q_{e, 1}$ on the rows of $P$ has expectation equal to the all $3/n$ function on the vertices of $D_l$. Recall that $P$ is obtained from the edge-vertex incidence matrix of $D_l$ by replacing every $-1$ with a $2$. Thus the expectation of the rows under $q'$ at a vertex $v$ is $\sum_{w: (w,v) \in D_l} q'_{(w,v)} + \sum_{w: (v,w) \in D_l} 2q'_{(v,w)}$; see \cref{fig:column-sum}. We now check that this is equal to $3/n$ for each vertex of $D_l$; it is helpful to look at \cref{fig:graph}. The leaves $u_l, v_l, w_l,$ and $\overline{w}_l$ all have outdegree one and indegree zero, and $q'$ takes the value $3 \cdot 3/6n = 3/2 n$ on the outgoing edges. The expectation under $q'$ thus takes value $3/n$ on these vertices. On vertices of indegree one and outdegree one, $q'$ takes the value $\frac{1}{2n} \left(2 + (-2)^{- k}\right)$ on the incoming edge and $\frac{1}{2n} \left(2 + (-2)^{- (k+1)}\right)$ on the outgoing edge. Thus the expectation takes the value $\frac{1}{2n} \left(2 + (-2)^{- k}\right) + \frac{1}{2n} \left(4 - (-2)^{- k}\right) = 3/n.$ The remaining vertices to check, those of total degree three, are $r$ and $w_{l-2}$. For $r$, which has only incoming edges, the expectation under $q'$ is $2\cdot \frac{1}{2n} \left(2 + (-2)^{- (l-1)}\right) + \frac{1}{2n} \left(2 + (-2)^{- (l-2)}\right),$ which is again $3/n$. For $w$ the expectation is $2\cdot \frac{1}{2n} \left(2 -\frac{1}{2}\right) + 2\cdot \frac{1}{2n} \left(2 -\frac{1}{2}\right) = 3/n.$ This completes the proof. \end{proof}

\begin{lem}\label{lem:unimod_sing}
If $A$ is an $n \times k$ totally unimodular matrix with linearly independent columns, then the eigenvalues of $A^T A$ are all at least $1/n^2$. 

\end{lem}

\begin{proof} First note that $k \leq n$ by the linear independence of the columns of $A$. The least eigenvalue of $A^T A$ is $\min_{x \in \RR^k\setminus \{0\}} (x^T A^T A x)/\|x\|^2 = \min_{x \in \RR^k\setminus \{0\}} \|Ax\|^2/\|x\|^2$, so it suffices to show that for all $x \in \RR^k$, $A x$ has norm at least $\|x\|/n$. Indeed, if $Ax = y$, then there is some invertible $k\times k$ submatrix $A'$ of $A$ and $k \times 1$ submatrix $y'$ of $y$ such that $A'x = y'$. By Cramer's rule and unimodularity of $A'$ we have that, for $i \in [k]$,
	$$ x_i = \dfrac{\det(B_i)}{\det(A')} = \pm \det(B_i)$$
	where $B_i$ is simply the matrix that one obtains by replacing the $i^{th}$ column of $A'$ with the
	vector $y'$. By performing the Laplace expansion with respect to the $i^{th}$ column, and by unimodularity of the
	minors, we have that $x_i \leq \|y\|_1$, and so $\|x\|_2 \leq \sqrt{k} \|y\|_1 \leq n \|Ax\|_2$ (using $k \leq n$).
\end{proof}


\section{The noncommutative case}\label{sec:noncommutative}

In this section we extend the results from the commutative to the noncommutative case. For this, we recall in the first subsection necessary concepts such as moment maps and moment polytopes, and we define the weight margin and the gap of a representation. The second subsection introduces the key concept of a \emph{free} subset of weights, see \cite{franz}. This concept dates at least back to \cite[Proposition~1.2]{dadok1985polar}, where it is called \emph{strong orthogonality}. Freeness will be used to transfer results from the commutative to the noncommutative case.\footnote{Actually all presented concepts in the first two subsections work in the very general setting of reductive groups and their rational representations. For the sake of clarity and concreteness we stick to the special case needed in this paper, i.e. the reductive group $\SL(n)^d := \SL(n) \times \cdots \times \SL(n)$ with $d \geq 1$ many copies of $\SL(n)$.}
The latter is done in the following three subsections, where we prove bounds on the tensor gap, on the gap for homogeneous polynomials and on the diameter for the natural $\SL(n)^3$ action on $3$-tensors. Finally, we show a bound for the weight margin of certain quiver representations. This provides an example, where the constructed set of weights is \emph{not} free, compare~\cref{rem:QuiverNotFree}. Still, after adding enough arrows to the considered quiver, we are able to ensure the same bound for the gap.

\subsection{Moment maps and moment polytopes}

In the following we introduce the null-cone problem and its dual characterization via moment maps and moment polytopes. This allows us to rigorously introduce the weight margin and the gap of a rational representation. Thereby we establish precise meaning and interpretation of our results regarding these two notions (in view of the null-cone problem). We stick to the notation of \cite{gradflow}, where the gap (implicitly) and the weight margin have been introduced. A reader unfamiliar with 
representation theory is referred to \cref{sec:RepTheoryBackground}.

Let $G = \SL(n)^d$, $K = \SU(n)^d$, $\T = \ST(n)^d$ and $\T_K = K \cap \T$ be matrix Lie subgroups of $\GL(dn)$ via block-diagonal embedding. Then we can think of their Lie algebras $\Lie(G)$ etc. as being block diagonally embedded into $\CC^{dn \times dn}$. For a rational representation $\pi \colon G \to \GL(V)$ we write $g \cdot v := \pi(g)v$ for the induced action, where $g \in G$ and $v \in V$. Moreover, we denote the set of weights of $\pi$ by $\Omega(\pi) \subseteq i \Lie(\T_K)$ and the induced representation on Lie algebras by $\Pi \colon \Lie(G) \to \End(V)$. We remark that we usually identify $i \Lie(\T_K) \cong (\id_n^\perp)^d \subseteq (\RR^n)^d$, where $\id_n^\perp$ denotes the orthogonal complement of the all-ones vector $\id_n$ in $\RR^n$.

The \emph{orbit} of $v \in V$ is $G \cdot v := \{ g \cdot v \mid g \in G \}$ and we denote its closure\footnote{The Euclidean- and the Zariski-closure of $G \cdot v$ coincide.} by $\overline{G \cdot v}$. A vector $v$ is called $G$-\emph{unstable}, if $0 \in \overline{G \cdot v}$, and otherwise $v$ is $G$-semistable. Equivalently, a vector $v \in V$ is $G$-unstable if and only if its \emph{capacity}
	\[ \capa_G(v) := \inf_{g \in G} \; \| g \cdot v \|^2 \]
equals zero. The $G$-unstable vectors form an affine subvariety of $V$ - the \emph{null-cone} (with respect to $G$). Orbit, stability, and capacity can also be defined for $\T$ by replacing $G$ by $\T$ in the definitions. 

As discussed in \cref{subsec:noncomm-intro}, the null-cone problem has many applications in different fields of computer science, mathematics and physics.

Next, we introduce the moment map. Given a rational representation $\pi \colon G \to \GL(V)$ there exists an Hermitian inner product $\langle \cdot , \cdot \rangle$ on $V$, by convention linear in the second argument, such that $\langle k \cdot v, k \cdot w \rangle = \langle v,w \rangle$ holds for all $k \in K$ and all $v,w \in V$.\footnote{In our concrete representations later on this will be the standard inner product.}

\begin{dfn}\label{dfn:momentMap}
For $v \in V \setminus \{ 0 \}$ we define $\mu_G(v) \in i \Lie(K)$ as the unique element of the real vector space $i \Lie(K)$, which satisfies for all $A \in i\Lie(K)$
	\begin{align*}
	\tr \big( \mu_G(v)A \big) = \frac{\langle v, \Pi(A)v \rangle}{\langle v,v \rangle}.
	\end{align*}
This defines the \emph{moment map} $\mu_G \colon V \setminus \{0\} \to i\Lie(\T)$ of $G$. Replacing $G$ by $\T$ and $K$ by $\T_K$ we derive the moment map $\mu_{\T} \colon V \setminus \{0\} \to i\Lie(\T_K)$ of $\T$.
\end{dfn}

The maps $\mu_G$ and $\mu_{\T}$ are indeed moment maps in the sense of symplectic geometry; namely for the induced action of $K$ and, respectively, $\T_K$ on the projective space $\PP(V)$. Recall $i\Lie(K) \subseteq \CC^{dn \times dn}$ so we can consider $\| \mu_G(v) \|_F$ and $\| \mu_{\T}(v) \|_F$.

An important application of these moment maps is due to the Kempf-Ness theorem \cite{KempfNess}, which provides a duality for the null-cone membership problem:
	\begin{equation}\label{eq:KempfNessDuality}
	\capa_G(v) = 0 \qquad \Leftrightarrow \qquad 
	0 < \inf_{g \in G} \; \| \mu_G(g \cdot v) \|_F = \min_{0 \neq w \in \overline{G \cdot v}} \; \| \mu_G(w) \|_F
	\end{equation}
and similarly for $\T$, replacing $G$ by $\T$ in the above equation. The two moment maps are related as follows. 

\begin{prp}\label{prp:MomentMaps}
Let $p \colon i \Lie(K) \to i \Lie(\T_K)$ be the orthogonal projection. Then $\mu_{\T} = p \circ \mu_G$ and $\norm{\mu_{\T}(v)}_F \leq \norm{\mu_{G}(v)}_F$ for all $v \in V \setminus \lbrace 0 \rbrace$.
\end{prp}

\begin{proof}
Since $i\Lie(\T_K) \subseteq i\Lie(K)$ the definition of the moment maps gives $\tr[\mu_{\T}(v) H] = \tr[\mu_{G}(v) H]$ for all $H \in i \Lie(\T_K)$. But $\mu_{\T}(v) \in i \Lie(\T_K)$ is the unique element with this property, hence $p(\mu_G(v)) = \mu_{\T}(v)$.  The inequality $\norm{\mu_{\T}(v)}_F \leq \norm{\mu_{G}(v)}_F$ follows directly from the first part.
\end{proof}

Now, we explain how the moment maps induce certain polytopes, which can also be used to express the duality in \eqref{eq:KempfNessDuality}. Moreover, the combinatorics of these polytopes captures the important complexity measures \emph{(weight) margin} and \emph{gap}. Indeed, one of our main contributions is to analyze parts of this combinatorics, thereby deducing complexity barriers for certain computational problems.

Since the action of $\T$ via $\pi$ is completely determined by the weight space decomposition $V = \bigoplus_{\omega \in \Omega(\pi)} V_\omega$ of $V$, one can compute $\mu_{\T}(v)$ in terms of this decomposition. For this, write $v = \sum_{\omega} v_\omega$ with $v_\omega \in V_\omega$ and define the support of $v$ with respect to $\pi$ as
	\begin{align*}
	\supp(v) := \lbrace \omega \in \Omega(\pi) \mid v_\omega \neq 0 \rbrace.
	\end{align*}
Using that distinct weight spaces are orthogonal, one computes
	\begin{align*}
	\mu_{\T}(v) = \sum_{\omega} \frac{ \langle v_\omega, v_\omega \rangle}{ \langle v,v \rangle} \; \omega,
	\end{align*}
which is a convex combination of the weights in $\supp(v)$. Noting that $\supp(v) = \supp(t \cdot v)$ for $t \in \T$ also $\mu_{\T}(t \cdot v) \in \Delta_{\T}(v) := \conv \lbrace \omega \mid \omega \in \supp(v) \rbrace$. In fact,
	\begin{align*}
	\Delta_{\T}(v) = \overline{ \lbrace \mu_{\T}(t \cdot v) \mid t \in \T \rbrace } = \overline{\left\lbrace \mu_{\T}(w) \mid w \in \T \cdot v \right\rbrace } \subseteq i \Lie(\T_K)
	\end{align*}
and $\Delta_{\T}(v)$ is called the \emph{weight polytope} of $v$. 

It is an astonishing result that for fixed $v \in V \setminus \{0\}$, the set $\{\mu_G(g \cdot v):g \in G\}$ gives rise to a polytope as follows. 
Let $\spec \colon \Herm(n) \to \RR^n$ be the function sending a Hermitian matrix to its eigenvalues in decreasing order. Recalling that $i\Lie(K) \subseteq \Herm(n)^d$ is block-diagonally embedded in $\CC^{dn \times dn}$, we set
	\begin{align*}
	s \colon i\Lie(K) \to \left( \RR^n \right)^d, \quad \diag(A_1, \ldots, A_d) \mapsto \big( \spec(A_1), \ldots, \spec(A_d) \big).
	\end{align*}
Then for $v \in V \setminus \{0\}$ the set\footnote{In an earlier version we stated $\Delta_G(v) = \left\lbrace s \big( \mu_G(w) \big) \mid w \in \overline{G \cdot v}, w \neq 0  \right\rbrace$, which is in general not correct as it may not contain parts of the relative boundary of $\Delta_G(v)$. Instead, considering the induced action of $G$ on $\PP(V)$ and the line $[v] \in \PP(V)$ spanned by $v \neq 0$, we have $\Delta_G(v) = \big\lbrace s \big( \mu_G(w) \big) \mid [w] \in \overline{G \cdot [v]} \subseteq \PP(V)  \big\rbrace$.}
	\begin{align*}
	\Delta_G(v) :=  \overline{\left\lbrace s \big( \mu_G(w) \big) \mid w \in G \cdot v  \right\rbrace }
	\end{align*}
is a rational convex polytope, see e.g. \cite{GuilleminSternberg} or \cite[Appendix]{NessStratification} by Mumford. We call $\Delta_G(v)$ the \emph{moment polytope} of $v$. Noting that $\| A \|_F = \| \spec(A) \|_2$ for any $A \in \Herm(n)$ we have $\| \mu_G(v) \|_F = \| s(\mu_G(v)) \|_2$ for all $v \in V \setminus\{0\}$. Thus, we can formulate the duality from \eqref{eq:KempfNessDuality} also as follows:
	\begin{align*}
	\capa_G(v) = 0 \qquad \Leftrightarrow \qquad  \dist \big( 0, \Delta_G(v) \big) > 0 
	\qquad \Leftrightarrow \qquad 0 \notin \Delta_G(v),
	\end{align*}
and similarly for $\T$. This motivates the following two definitions.

\begin{dfn}\label{dfn:WeightMarginGapConstant}
Let $\pi \colon G \to \GL(V)$ be a rational representation. We define the \emph{gap} of $\pi$ as\footnote{Gap and weight margin are well-defined, i.e. the minimum is attained. Indeed, the moment maps give rise to continuous maps on $\PP(V)$ and the non-zero $G$-unstable (respectively non-zero $\T$-unstable) vectors form a projective subvariety of $\PP(V)$; in particular they form a compact set.}
	\begin{align*}
	\gamma_G(\pi) := \min \big\lbrace \norm{\mu_G(v)}_F \mid v \neq 0 \text{ is } G\text{-unstable} \big\rbrace
	= \min \big\lbrace \dist \big( 0, \Delta_G(v) \big) \mid v \neq 0 \text{ is } G\text{-unstable} \big\rbrace,
	\end{align*}
and the \emph{weight margin} of $\pi$ as
	\begin{align*}
	\gamma_{\T}(\pi) := \min \big\lbrace \norm{\mu_{\T}(v)}_F \mid v \neq 0 \text{ is } \T\text{-unstable} \big\rbrace
	= \min \big\lbrace \dist \big(0, \Delta_{\T}(v)\big) \mid v \neq 0 \text{ is } \T\text{-unstable} \big\rbrace.
	\end{align*}
Equivalently, $\gamma_{\T}(\pi)$ is the margin of the set of weights $\Omega(\pi)$, i.e. $\gamma_{\T}(\pi) = \gamma(\Omega(\pi))$.
\end{dfn}

Thus, the gap $\gamma_G(\pi)$ is the largest constant $C > 0$ with the following property: If $\| \mu_G(v) \|_F < C$ for some vector $v\in V$, then $v$ is $G$-semistable. The same statement holds for the weight margin $\gamma_{\T}(\pi)$ replacing $G$ by $\T$. Therefore, these notions capture how small $\mu_G(g\cdot v)$ (respectively $\mu_{\T}(t\cdot v)$) must be to certify null-cone non-membership. The next remark connects the gap to the classical notion of \emph{instability} due to Mumford \cite{mumford1965geometric}.

\begin{rem}
The gap is twice the minimum value of all positive instabilities. Indeed, let $M(v)$ denote the instability of a non-zero vector $v$,  see e.g. \cite[eq.~(9)]{NessStratification}. Then $\dist(0,\Delta_G(v)) \geq 2 M(v)$ and \cite[Theorem~6.1]{NessStratification} implies
    \[\gamma_G(\pi) = \inf \lbrace 2M(v) \colon v \neq 0, v \text{ is } G\text{-unstable} \rbrace.\]
\end{rem}

\begin{exa}\label{exa:TensorWeightMarginAndGap}
Recall the tensor scaling action, in which the group $G = \SL(n)^d$ acts on $( \CC^n )^{\otimes d}$ via the representation
\begin{align*}
\pi_{n,d} \colon \SL(n)^d \to \GL \left( (\CC^n)^{\otimes d} \right), \; (g_1,\ldots,g_d) \mapsto g_1 \otimes \cdots \otimes g_d \, .
\end{align*}
Similar computations to those in \cref{exa:LeftMultSL} show that the set of weights of $\pi_{n,d}$ is
\begin{align*}
\Omega(\pi_{n,d}) = \Omega_{n,d} = \big\lbrace \eps_i \mid i \in [n] \big\rbrace^d \subseteq (\RR^n)^d.
\end{align*}
Therefore, the weight margin $\gamma_{\T}(\pi_{n,d})$ is the margin $\gamma(\Omega_{n,d})$ for the array scaling problem from \cref{thm:tensor-margin} and \cref{thm:MarginTensor}.
Moreover, the moment map $\mu_G$ for $\pi_{n,d}$ can be computed in terms of the quantum marginals as described in the introduction, i.e. $\gamma_G(\pi_{n,d})$ is indeed the tensor gap.
\end{exa}

The weight margin and the gap satisfy the following inequality.

\begin{prp}\label{prp:GapConstantWeightMargin}
It holds that $\gamma_{\T}(\pi) \leq \gamma_G(\pi)$.
\end{prp}

\begin{proof}
Let $v \neq 0$ be $G$-unstable. Then there exists $k \in K$ such that $k \cdot v$ is $\T$-unstable; see \cite[Theorem~3.25]{wallach}. By \cref{prp:MomentMaps} we obtain
	\begin{align*}
	\norm{\mu_G(v)}_F = \norm{\mu_G( k \cdot v)}_F \geq \norm{\mu_{\T}( k \cdot v)}_F \geq \gamma_{\T}(\pi)
	\end{align*}
where we used in the first equality that $\mu_G( k \cdot v) = k \mu_G(v) k^\dagger$. Therefore $\gamma_G(\pi) \geq \gamma_{\T}(\pi)$.
\end{proof}

This inequality motivates the next subsection.

\subsection{Free sets of weights}

\cref{prp:GapConstantWeightMargin} from the preceding subsection shows us that an upper bound for the weight margin $\gamma_{\T}(\pi)$ need not necessarily apply to the gap $\gamma_G(\pi)$. Still, many of our bounds in the commutative case (weight margin and diameter) transfer to the noncommutative case (gap and diameter). We use crucially the notion of a \emph{free} subset of weights (or \cite{franz}). Freeness is also known as \emph{strong orthogonality} \cite{dadok1985polar}. 

\begin{dfn}\label{dfn:freeGeneral}
Let $\pi \colon G \to \GL(V)$ be a rational representation with set of weights $\Omega(\pi)$.

A subset $\Gamma \subseteq \Omega(\pi)$ is called \emph{free} if no two distinct elements of $\Gamma$ differ by a root of $G$. In other words, $\Gamma \cap (\Gamma + \alpha) = \emptyset$ holds for all roots $\alpha$ of $G$.

Furthermore, a vector $v \in V\setminus \{0\}$ is called \emph{free} if its support $\supp(v) \subseteq \Omega(\pi)$ is free.
\end{dfn}

We transfer the results from the commutative to the noncommutative case with the upcoming \cref{prp:FreeForGapConstant}. 
It is known that for vectors $v$ with free support one has $\mu_G(v) = \mu_{\T}(v)$. This appears implicitly in \cite[Lemma~7.1]{Sjamaar} and \cite[Proposition~2.2]{franz}, but we prove it below for completeness. We thank Visu Makam for pointing out to us that this equality still holds under a weaker condition on $v$, when the representation decomposes into orthogonal subrepresentations. This can be used to turn our weight margin upper bound for quivers into a gap upper bound (\cref{thm:UpperBoundQuiver}). This weaker condition also appears in \cite[Theorem~6.5]{derksen2020exponential}.

\begin{prp}\label{prp:FreeForGapConstant}
Let $\pi \colon G \to \GL(V)$ be a rational representation and suppose $V = \bigoplus_{i=1}^k V_i$ is an orthogonal decomposition into $G$-subrepresentations with respect to the $K$-invariant inner product, that is used to define $\mu_{\T}$ and $\mu_G$. Let $v = (v_1,\ldots,v_k) \in V \setminus \{0\}$, $v_i \in V_i$ be such that all supports $\Gamma_i := \supp(v_i) \subseteq \Omega(\pi)$ are free. Then for all $t \in \T$ it holds that $\mu_G(t \cdot v) \in i\Lie(\T_K)$ and $\mu_G(t \cdot v) = \mu_{\T}(t \cdot v)$.

If additionally $0 \notin \Delta_{\T}(v) = \conv(\Gamma)$, where $\Gamma = \bigcup_i \Gamma_i$, then the upper bound $\dist(0, \conv(\Gamma))$ for the weight margin $\gamma_{\T}(\pi)$ also applies to the gap, i.e. $\gamma_G(\pi) \leq \dist(0, \conv(\Gamma))$.
\end{prp}

\begin{proof}
The action of $\T$ preserves the supports $\Gamma_i$, and in particular preserves their freeness. Hence, it suffices to show $\mu_G(v) \in i\Lie(\T_K)$, which immediately yields $\mu_G(v) = \mu_{\T}(v)$ by \cref{prp:MomentMaps}. Moreover, the orthogonality with respect to the $K$-invariant inner product shows $\mu_G(v) = H_1 + \cdots + H_k$, where $H_i = \mu_G^{(i)}(v_i)$ is given by the moment map $\mu_G^{(i)}$ of the $G$-module $V_i$ if $v_i \neq 0$ and otherwise $H_i = 0$. The latter holds similarly for $\mu_{\T}$.

Therefore, we may assume $k=1$, i.e. $v \neq 0$ has free support $\Gamma$. We write $v = \sum_{\omega \in \Gamma} v_\omega$ for $v_\omega \in V_\omega$. First, we note that the root space decomposition\footnote{This is the weight space decomposition of the adjoint representation, compare \cref{exa:Roots}.} $\Lie(G) = \Lie(T) \oplus \bigoplus_{\alpha} \Lie(G)_\alpha$ and $\Lie(T) = \Lie(T_K) \oplus i \Lie(T_K)$ yield the \emph{orthogonal} decomposition 
\begin{equation}\label{eq:DecompositioniLieK} 
	i \Lie(K) = i \Lie(T_K) \oplus \big( i \Lie(K) \cap \bigoplus_\alpha \Lie(G)_\alpha \big).
\end{equation}
We fix $\omega \in \Gamma$, a root $\alpha$ of $G$ and some $A \in \Lie(G)_\alpha$. Then $\Pi(A) v_\omega \in V_{\omega + \alpha}$ by \cref{prp:Roots}, where either $V_{\omega + \alpha} = \{0\}$ or $\omega + \alpha \in \Omega(\pi) \backslash \Gamma$ as $\Gamma$ is free. Using that distinct weight spaces are orthogonal we obtain that $\langle v_\sigma, \Pi(A)v_{\omega} \rangle = 0$ holds for all $\sigma \in \Gamma$. Therefore, we conclude that $\langle v, \Pi(X)v \rangle = 0$ for all $X \in \bigoplus_\alpha \Lie(G)_\alpha$. In particular, $\tr \big( \mu_G(v)X \big) = 0$ holds for all $X \in i \Lie(K) \cap \bigoplus_\alpha \Lie(G)_\alpha$. Hence, together with the orthogonal decomposition in  \cref{eq:DecompositioniLieK} we deduce $\mu_G(v) \in i \Lie(\T_K)$. The first statement is proven.

For the second claim we note that indeed $\bigcup_i \Gamma_i = \supp(v)$. If additionally $0 \notin \conv(\Gamma) = \Delta_{\T}(v)$, then $v$ is $\T$-unstable. In particular, $v$ is $G$-unstable and thus
	\[ \gamma_G(\pi) \leq \dist \big( 0, \Delta_G(v) \big) . \]
On the other hand, we have
	\begin{align*}
	\dist \big( 0, \Delta_G(v) \big) = \inf_{g \in G} \, \| \mu_G(g \cdot v) \|_F
	\leq \inf_{t \in \T} \, \| \mu_G(t \cdot v) \|_F \overset{(*)}{=} \dist \big( 0, \conv(\Gamma) \big) ,
	\end{align*}
where we used $\mu_G(t \cdot v) = \mu_{\T}(t \cdot v)$ in $(*)$. We conclude by combining the two inequalities.
\end{proof}

\begin{rem}
It is well-known that any rational representation $\pi \colon G \to \GL(V)$ can be decomposed into $G$-irreducible subrepresentations that are pairwise orthogonal with respect to the fixed $K$-invariant inner product. \cref{prp:FreeForGapConstant} shows that ensuring freeness on the irreducible subrepresentations suffices.
\end{rem}

We end the section with an interesting connection between the weight margin and the gap.

\begin{prp}\label{prp:WeightGapForDirectPower}
Let $\pi \colon G \to \GL(V)$ be a rational representation and denote its $m$-fold direct sum by $\pi^m$.
	\begin{enumerate}
	\item The weight margin satisfies $\gamma_{\T}(\pi) = \gamma_{\T}(\pi^m)$ for all $m \geq 1$.
	
	\item The gap satisfies $\gamma_{G}(\pi^m) \geq \gamma_{G}(\pi^{m+1})$ for all $m \geq 1$.
	
	\item There exists some $m \leq \dim(V)$ such that $\gamma_{G}(\pi^m) = \gamma_{\T}(\pi^m) = \gamma_{\T}(\pi)$.
	\end{enumerate}	 
\end{prp}

\begin{proof}
We note that $\pi^m$ is given by the action $g \cdot (v_1,\ldots,v_m) = (g \cdot v_1 , \ldots, g \cdot v_m)$ on $V^m$. Furthermore, the $K$-invariant inner product $\langle \cdot, \cdot \rangle$ of $V$ induces naturally a $K$-invariant product on $V^m$ by
	\begin{align*}
	\langle (v_1,\ldots,v_m), (w_1,\ldots,w_m) \rangle_{V^m} := \sum_{i=1}^m \; \langle v_i, w_i \rangle.
	\end{align*}

For the first claim just note that the weight space decomposition for $\pi^m$ is $V^m = \bigoplus_{\omega \in \Omega(\pi)} V_{\omega}^m$ and hence $\Omega(\pi^m) = \Omega(\pi)$.

For the second claim, let $(v_1,\ldots,v_m) \in V^m \setminus \{0\}$ be $G$-unstable such that $\| \mu_G(v_1,\ldots,v_m) \|_F = \gamma_G(\pi^m)$. Then $(v_1,\ldots,v_m,0) \in V^{m+1} \setminus \{0\}$ is $G$-unstable as well, so $\| \mu_G(v_1,\ldots,v_m,0) \|_F \geq \gamma_{G}(\pi^{m+1})$. Moreover, under the inner product $\langle \cdot, \cdot \rangle_{V^{m+1}}$ the first $m$ copies of $V$ are orthogonal to the last copy. Thus, we have $\mu_{G}(v_1,\ldots,v_m,0) = \mu_{G}(v_1,\ldots,v_m)$ and hence $\| \mu_G(v_1,\ldots,v_m,0) \|_F = \| \mu_G(v_1,\ldots,v_m) \|_F =  \gamma_G(\pi^m)$.

Finally, let $\Gamma = \{ \omega_1,\ldots,\omega_m \} \subseteq \Omega(\pi)$ be such that $0 \notin \conv(\Gamma)$ and $\dist(0, \conv(\Gamma)) = \gamma_{\T}(\pi)$. We have $m \leq \vert \Omega(\pi)\vert \leq \dim(V)$ by the weight space decomposition $V = \bigoplus_{\omega \in \Omega(\pi)} V_{\omega}$. Now, for each $\omega_i \in \Gamma$ fix some weight vector $v_i \in V_{\omega_i} \setminus \{0\}$. Then $v := (v_1,\ldots,v_m) \in V^m$ satisfies the assumptions of \cref{prp:FreeForGapConstant}, because $\Gamma_i = \lbrace \omega_i \rbrace$ is free and the distinct copies of $V$ are orthogonal under $\langle \cdot, \cdot \rangle_{V^m}$. Thus, we obtain
	\begin{align*}
	\gamma_G(\pi^m) \leq \dist \big( 0, \conv(\Gamma) \big) = \gamma_{\T}(\pi) = \gamma_{\T}(\pi^m),
	\end{align*}
but on the other hand $\gamma_G(\pi^m) \geq \gamma_{\T}(\pi^m)$ by \cref{prp:GapConstantWeightMargin}.
\end{proof}

\subsection{Freeness for tensors}

We recall from \cref{exa:TensorWeightMarginAndGap} that $\pi_{n,d}$ denotes the natural representation of $G = \SL(n)^d$ on $(\CC^n)^{\ot d}$ and that the weight margin $\gamma_{\T}(\pi_{n,d})$ is the margin $\gamma(\Omega_{n,d})$ for the array scaling problem from \cref{thm:tensor-margin} and \cref{thm:MarginTensor}.
The purpose of this subsection is to prove the bounds for $\gamma_{\T}(\pi_{n,d})$ from \cref{thm:MarginTensor} also for the gap $\gamma_G(\pi_{n,d})$.

\begin{thm}\label{thm:GapConstantTensor}
Let $\pi_{n,d}$ be the representation induced by the natural action of $G := \SL(n)^d$ on $(\CC^n)^{\otimes d}$. Then the weight margin $\gamma_{\T}(\pi_{n,d})$ and the gap $\gamma_G(\pi_{n,d})$ can be bounded as follows:
	\begin{itemize}
	\item[(a)] If $n=2$ and $d \geq 3$, then
	$
	\gamma_{\T}(\pi_{2,d}) \leq \gamma_G (\pi_{2,d}) \leq 2^{-\frac{d}{2} + 1}.
	$
	\item[(b)] If $n \geq 3$ and $d = 3$, then  $\gamma_{\T}(\pi_{n,3}) \leq \gamma_G(\pi_{n,3}) \leq 2^{-n+1}$.
	\item[(c)] If $n \geq 3$ and $d = 6 r - 3$ for some integer $r \geq 2$, then
		\begin{align*}
		\gamma_{\T}(\pi_{n,d}) \leq \gamma_G(\pi_{n,d}) \leq \frac{\sqrt{6}}{(n-1)\sqrt{r}} \; 2^{-r(n-1) + 1} 
		\leq 2^{- r(n-1) + 1} = 2^{- \frac{(d+3)(n-1)}{6} + 1}.
		\end{align*}
	\end{itemize}
\end{thm}

Though the above theorem only applies to certain $d$, we can ``pad'' the tensors to obtain similar results for all $d \geq 3$. This is because bounds for $\gamma_G(\pi_{n,d})$ via free subsets of weights also hold for $\gamma_G(\pi_{n,d+2})$ and $\gamma_G(\pi_{n,d+3})$, see \cref{prp:dTensorsPadding}. The missing case $n \geq 3$ and  $d=4$ is treated in \cref{prp:4Tensors}. Therefore, we can conclude \cref{thm:tensor-gap} from the above \cref{thm:GapConstantTensor}. 

Our main method for transfering the bounds from the commutative case (\cref{thm:MarginTensor}) to the noncommutative case is to use the concept of freeness in conjunction with \cref{prp:FreeForGapConstant}. The following definition will be convenient for proving freeness of tensors.

\begin{dfn}[Free sets]\label{dfn:free}
A set $M \subseteq [n]^d$ is called \emph{free}, if $i = (i_1, \ldots, i_d), j = (j_1, \ldots, j_d) \in M$ with $i \neq j$ always implies $\vert \lbrace i_l \neq j_l \mid l=1, \ldots, d \rbrace \vert \geq 2$.
\end{dfn}

\begin{prp}\label{prp:FreeTensorVsFreeGeneral}
Let $M \subseteq [n]^d$ and denote the induced subset of weights by 
	\[\Gamma_M := \lbrace (\eps_{i_1}, \ldots, \eps_{i_d}) \mid (i_1,\ldots,i_d) \in M \rbrace \subseteq (\RR^n)^d .\]
Then $M$ is a free set if and only if the set of weights $\Gamma_M \subseteq \Omega(\pi_{n,d})$ is free as in \cref{dfn:freeGeneral}.
\end{prp}

\begin{proof}
We recall that $\Gamma_M$ is free if and only if no two distinct elements of $\Gamma_M$ differ by a root of $G = \SL(n)^d$, see \cref{dfn:freeGeneral}. Furthermore, remember that the roots of $G$ are
\begin{align*}
	(e_i - e_j, 0_n, \ldots, 0_n), (0_n, e_i - e_j, 0_n, \ldots, 0_n), \ldots \ldots, (0_n, \ldots, 0_n, e_i - e_j) \in \left( \RR^n \right)^d
	\end{align*}
for $i,j \in [n]$ with $i \neq j$; see also \cref{exa:Roots}. Now, if $M \subseteq [n]^d$ is not free, then there exist $i = (i_1, \ldots, i_d), j = (j_1, \ldots, j_d) \in M$ with $i \neq j$ such that they exactly differ one component. Without loss of generality we assume $i_1 \neq j_1$ and $i_l = j_l$ for $l=2,\ldots,n$. But then
	\[ (\eps_{i_1}, \ldots, \eps_{i_d}) = (\eps_{j_1}, \ldots, \eps_{j_d}) + (e_{i_1} - e_{j_1}, 0_n, \ldots, 0_n),\]
and hence $\Gamma_M$ is not free. Clearly, the argument can be inverted to show that if $\Gamma_M$ is not free, then $M$ is not free.
\end{proof}
The above proposition shows how the equality $\mu_G(t \cdot v) = \mu_{\T}(t \cdot v)$ of \cref{prp:FreeForGapConstant} can be verified directly for tensors. For tensors, the moment map components are the quantum marginals, and the equality $\mu_G(t \cdot v) = \mu_{\T}(t \cdot v)$ simply says that the quantum marginals are diagonal. Each off-diagonal entry of a quantum marginal is the inner product between distinct $d-1$-dimensional slices of a tensor, and if the support of the tensor is free then the supports of such slices are entirely disjoint - thus the quantum marginals are diagonal.

In the following two Propositions we show, that the subsets of weights, which witness the upper bounds for the (weight) margin in \cref{thm:MarginTensor}, are all free. Thereby, we will implicitly use \cref{prp:FreeTensorVsFreeGeneral}.

\begin{prp}\label{prp:freeQubits}
For $r \geq 2$ the rows of $A_{2r}$ form a free subset of $[2]^{2r}$, i.e. $\Gamma_{2,2r}$ is free. Moreover, for $r \geq 1$ the set of weights $\Gamma_{2,2r+1}$ is free.
\end{prp}

\begin{proof}
Clearly, $\Gamma_{2,3} = \{ \eps_{1,1,1}, \eps_{2,1,2} \}$ is free. Recall the constructions of $\Gamma_{2,2r}$ and $\Gamma_{2,2r+1}$ from \cref{sec:qubits}. If $\Gamma_{2,2r}$ is free, then $\Gamma_{2,2r+1}$ is clearly also free. Thus, we are left to prove the former.

Consider $A_{2r}$ as defined in \cref{eq:defA2r}. We must show that distinct rows of $A_{2r}$ differ in at least two entries for all $r \geq 2$. The claim is proven by induction on $r \geq 3$. For $r = 3$, we verify the claim by inspection of $A_6$. Let $a_i$ be the $i^{th}$ row of $A_6$; its definition is recalled in the left-hand table below. The right-hand table lists for each pair $a_i$, $a_j$ with $i < j$ two distinct entries in which $a_i$ and $a_j$ differ, which shows the claim for $r=3$.
	\begin{center}
	\begin{tabular}{|c||c|c|c|c|c|c|}
	\hline 
	entry & 1 & 2 & 3 & 4 & 5 & 6 \\ 
	\hline 
	\hline
	$a_1$ & 1 & 1 & 1 & 1 & 1 & 1 \\ 
	\hline 
	$a_2$ & 2 & 1 & 2 & 2 & 2 & 2 \\ 
	\hline 
	$a_3$ & 1 & 2 & 2 & 1 & 1 & 1 \\ 
	\hline 
	$a_4$ & 2 & 2 & 1 & 1 & 2 & 2 \\ 
	\hline 
	$a_5$ & 1 & 2 & 1 & 2 & 2 & 1 \\ 
	\hline 
	$a_6$ & 2 & 2 & 2 & 2 & 1 & 1 \\ 
	\hline 
	\end{tabular} $\qquad \qquad$
	\begin{tabular}{|c||c|c|c||c|c|}
	\hline 
	  & $a_2$ & $a_3$ & $a_4$ & $a_5$ & $a_6$ \\ 
	\hline \hline
	$a_1$  & 1,3 & 2,3 & 1, 2 & 2,4 & 1,2 \\ 
	\hline 
	$a_2$ & & 1,2 & 2,3 & 1,2 & 5,6 \\ 
	\hline 
	$a_3$ & & & 1,3 & 3,4 & 1, 4 \\ 
	\hline \hline
	$a_4$ & & & & 1,4 & 3,4 \\ 
	\hline 
	$a_5$ & & & & & 1,3 \\ 
	\hline 
	\end{tabular} 
	\end{center}
In fact, the table also proves the claim for $r=2$, since $a_1,\ldots,a_4$ already pairwise differ in at least two of the first four entries.

Now assume that the claim holds for some fixed $r \geq 3$. Let $a_i, a_j$ be distinct rows of $A_{2r + 2}$; we will show they differ in at least two entries. If $1 \leq i<j \leq 2r$, then by our inductive hypothesis there is nothing to prove because the first $2r$ rows of $A_{2r+2}$ contain $A_{2r}$ as a submatrix. 

To complete the proof, it is enough to show that the $4\times (2r + 2)$ submatrix formed by restricting to the $m^{th}$ block row, $m \in [r]$, and the last block row of $A_{2r + 2}$ satisfies the hypothesis, i.e. any two distinct rows of this submatrix differ in at least two entries. This is the case as restricting to its $1^{st}$, $m^{th}$ and last block columns yields a $4 \times 6$ submatrix of $A_6$ if $m \neq 1$, namely 
$$\begin{pmatrix}
B_2 & B_3 & B_1\\
B_2 & B_2 & B_3
\end{pmatrix},$$
and a $4 \times 4$ submatrix equal to $A_4$ if $m = 1$. 
\end{proof}

\begin{prp}\label{prp:WnFree}
For $n \geq 3$ the set $\mathfrak{W}_n \subseteq [n]^3$ is free, i.e. $\Gamma_{n,3} \subseteq \Omega(\pi_{n,3})$ is free. Furthermore, for $n \geq 3$ and $r \geq 2$ the set of weights $\Gamma_{n,6r-3} \subseteq \Omega(\pi_{n,6r-3})$ is free.
\end{prp}

\begin{proof}
We remind the reader that
	\[ \mathfrak{W}_n = \big\lbrace (s,1,s), (s,s,1), (s-1,s,s) \mid s=2,3,\ldots,n \big\rbrace. \]
Let $x = (x_1, x_2, x_3), y = (y_1, y_2, y_3) \in \mathfrak{W}_n$ be such that $x \neq y$. We prove by a distinction of cases that $x$ and $y$ differ in at least two entries. First, we assume $x_1 = y_1$. Then $a := x_1 = y_1 \geq 2$, otherwise $x = (1, 2, 2) = y$ contradicts $x \neq y$. Thus $x,y \in \lbrace (a, 1, a), (a, a, 1), (a, a+1, a+1) \rbrace$ and we conclude that $x$ and $y$ differ in at least two entries as $x \neq y$.
Second, we assume $x_1 \neq y_1$. There is nothing to show if $x_2 \neq y_2$, so we additionally assume $b:= x_2 = y_2$. If $b=1$, then we are done by $x = (x_1, 1, x_1)$ and $y = (y_1, 1, y_1)$. On the other hand, $b \geq 2$ yields $x, y \in \lbrace (b, b, 1), (b-1, b, b) \rbrace$ and as $x \neq y$ they differ in the first and third entry. This proves the first statement.

For the second claim, recall that
	\[ \Gamma_{n,6r-3} = \{ \eps_{\sigma(i), \sigma(j), \sigma(k)} \mid (i,j,k) \in \mathfrak{W}_{rn} \setminus \mathfrak{J}_r \}, \]
where $\sigma \colon [r n] \to [n]^{2r-1}$ is injective, compare \cref{rem:Sigma}. By the first part $\mathfrak{W}_{r n }$ is free and so is its subset $\mathfrak{W}_{r n } \setminus \mathfrak{J}_r$. Hence $\Gamma_{n, 6r - 3}$ is free as $\sigma$ is injective.
\end{proof}

We are now ready to deduce \cref{thm:GapConstantTensor}.

\begin{proof}[Proof of \cref{thm:GapConstantTensor}]
Recall that all the bounds in \cref{thm:GapConstantTensor} hold for the weight margin $\gamma_{\T}(\pi)$ by \cref{thm:MarginTensor}. This was proven by exhibiting witness sets $\Gamma_{n,d} \subseteq \Omega(\pi_{n,d})$ such that $0 \notin \conv(\Gamma_{n,d})$, which gives the bound $\gamma_{\T}(\pi_{n,d}) \leq \dist(0, \conv(\Gamma_{n,d}))$. But if $\Gamma_{n,d}$ is free, then we even have
	\[ \gamma_{G}(\pi_{n,d}) \leq \dist \big( 0, \conv(\Gamma_{n,d}) \big) \]
by \cref{prp:FreeForGapConstant}. By \cref{prp:freeQubits} the witness sets $\Gamma_{2,3}$ and $\Gamma_{2,2r}$, $\Gamma_{2,2r+1}$, $r \geq 2$ for \cref{thm:MarginTensor}(a)  are free, which proves \cref{thm:GapConstantTensor}(a). Similarly, we conclude parts (b) and (c) with \cref{prp:WnFree}, which shows that for $n\geq 3$ and $r \geq 2$ the witness sets $\Gamma_{n,3}$ and $\Gamma_{n,6r-3}$ are free.
\end{proof}

\subsection{Freeness for homogeneous polynomials}

In the following we transfer the result from $d$-tensors to the natural $\SL(n)$ action on homogeneous $d$-forms in $n$ variables. This representation is given by
	\begin{align*}
	\varrho_{n,d} \colon \SL(n) \to \GL \big( \CC[x_1, \ldots, x_n]_d \big), \; g \mapsto \big( p(x) \mapsto p(g^{-1} x) \big).
	\end{align*}
Each monomial $x^\alpha = x_1^{\alpha_1} \cdots x_n^{\alpha_n}$, given by a multi-index $\alpha = (\alpha_1, \ldots, \alpha_n) \in (\ZZ_{\geq 0})^n$ with $\vert \alpha \vert := \sum_i \alpha_i = d$, is a weight vector for $\varrho_{n,d}$ with weight $- \alpha + \frac{d}{n} \id_n$. Therefore
	\begin{align*}
	\Omega(\varrho_{n,d}) = \left\lbrace - \alpha + \frac{d}{n} \id_n \; \bigg\vert \; \alpha \in (\ZZ_{\geq 0})^n \text{ with } \vert \alpha \vert = d \right\rbrace ,
	\end{align*}
i.e. $\Omega(\varrho_{n,d}) = \Omega'$ from \cref{eq:poly-scaling} and the bounds from \cref{cor:dFormsWeightMargin} apply to $\gamma_{\ST(n)}(\varrho_{n,d}) = \gamma(\Omega')$. If $n = dm$ for some integer $m \geq 1$, then we have $- \Omega(\pi_{m,d}) \subseteq \Omega(\varrho_{n,d})$.

\begin{prp}\label{prp:FreeTensorVsPolynomial}
Let $n = dm$ for some integer $m \geq 1$. If $\Gamma \subseteq \Omega(\pi_{m,d})$ is free, then $-\Gamma \subseteq \Omega(\varrho_{n,d})$ is free.
\end{prp}

\begin{proof}
We prove the statement by contraposition. Assume that $-\Gamma \subseteq \Omega(\varrho_{n,d})$ is not free. Then there exists a root $\alpha = e_i - e_j \in \RR^n$ of $\SL(n)$, where $i,j \in [n]$ with $i \neq j$, and two distinct weights $\omega, \omega' \in -\Gamma$ such that $\omega = \omega' + e_i - e_j$, equivalently $-\omega = -\omega' - e_i + e_j$. The latter equation enforces $- \alpha$ to be of the form
	\begin{align*}
	( 0_m, \ldots, 0_m, e_k - e_l, 0_m, \ldots, 0_m ) \in \left( \RR^m \right)^d \cong \RR^n
	\qquad \text{for some }  k,l \in [m] \text{ with } k \neq l,
	\end{align*}
because $-\omega, -\omega' \in \Omega(\pi_{m,d})$. Thus, $-\alpha$ is a root of $\SL(m)^d$ and hence $\Gamma \subseteq \Omega(\pi_{m,d})$ is not free.
\end{proof}

As a consequence of the preceding Proposition we obtain bounds for the gap $\gamma_{\SL(n)}(\varrho_{n,d})$.

\begin{thm}[Gap for Polynomial scaling]\label{thm:dFormsGap}
Let $d \geq 3$ and let $n = dm$ for some integer $m \geq 2$. Then there exists a constant $C > 0$, independent of $n$ and $d$ such that
	\begin{align*}
	\gamma_{\SL(n)}(\varrho_{n,d}) \leq 2^{-C d m} = 2^{-Cn}.
	\end{align*}
More concretely, for $d=3$ and $m \geq 3$ it holds that
	\begin{align*}
	\gamma_{\SL(n)}(\varrho_{n,d}) \leq \dist\big( 0, \Gamma_{m,3} \big) \leq 2^{-m + 1} = 2^{-\frac{n}{3} + 1},
	\end{align*}
and if $m \geq 3$ and $d = 6r-3$ for some $r \geq 2$, we have
	\begin{align*}
	\gamma_{\SL(n)}(\varrho_{n,d}) \leq \dist \big( 0, \Gamma_{m,6r-3} \big)
		\leq 2^{- r (m-1) + 1} = 2^{- \frac{(d+3)(m-1)}{6} + 1} \approx 2^{- \frac{n}{6}}.
	\end{align*}
\end{thm}

\begin{proof}
We recall that \cref{thm:tensor-gap} was proven by padding the results from \cref{thm:GapConstantTensor}. Thus, for each $m \geq 2$ and $d \geq 3$ the bound $\gamma_{\SL(m)^d}(\pi_{m,d}) \leq 2^{-Cmd}$ from \cref{thm:tensor-gap} is witnessed by a free set of weights $\Gamma_{m,d} \subseteq \Omega(\pi_{m,d})$, i.e. $0 < \dist(0, \conv(\Gamma_{m,d}) ) \leq 2^{-C d m}$. But then $0 \notin \conv(-\Gamma_{m,d})$ and $-\Gamma_{m,d} \subseteq \Omega(\varrho_{n,d})$ is free by \cref{prp:FreeTensorVsPolynomial}. Therefore, \cref{prp:FreeForGapConstant} yields
	\begin{align*}
	\gamma_{\SL(n)}(\varrho_{n,d}) \leq \dist \big(0, \conv(-\Gamma_{m,d}) \big) = \dist \big(0, \conv(\Gamma_{m,d}) \big) \leq 2^{-C d m}.
	\end{align*}
Similarly, we get the other bounds by using freeness of $\Gamma_{m,3}$ and, respectively, $\Gamma_{m,6r-3}$ (see \cref{prp:WnFree}) 
combined with the distance bounds \cref{lem:distKravtsov} and \cref{lem:distStackingKravtsov}, respectively.
\end{proof}

\subsection{Freeness and diameter bound}\label{subsec:free-diameter}

In this section we show that the diameter lower bound of \cref{thm:diameter} generalizes to diameter bounds for the capacity \cref{eq:capacity} over the noncommutative group $G = \SL(n)^d$. Many algorithms for computing the capacity have resorted to geodesically convex optimization - $G$ can be viewed as a manifold on which $g \mapsto \|g \cdot v\|^2$ is geodesically convex. The distance between an element of $g$ and the identity in this geometry is closely related to the condition number of the matrix $g$. The diameter bound question is the following: \emph{given an input $v$ and $\eps > 0$, how large a ball in $G$ about the identity must we optimize over to find an approximate minimizer $g \in G$ such that $\|g \cdot v\|^2 - \capa(v) \leq \eps$?} In other words, how well-conditioned can we expect approximate minimizers to \cref{eq:capacity} to be? This matters because all the algorithms we know start at the origin and take small steps in the manifold, and if all the high-precision solutions are far from the origin then such algorithms cannot reach any of them quickly. 

Before tackling this question we must make our notions of distance more precise. The manifold we use is actually not $G$ but rather the manifold $P$ of Hermitian, positive-definite matrices in $G$. Indeed, we can write 
$$\inf_{g \in G} \| g \cdot v\|^2 = \inf_{g \in G} \langle v,  g^\dagger g \cdot v\rangle = \inf_{x \in P} \langle v, x \cdot  v \rangle.$$
Thus we may instead optimize the function $f_v:g \mapsto \langle v, g \cdot v \rangle$ over $P$. The manifold $P$ is a prototypical example of a \emph{Hadamard manifold}, a complete, simply connected Riemannian manifold of non-positive sectional curvature \cite{bacak2014convex}. For us, $G = \SL(n)^d$ for some $d$, and so $P$ is just the set of $d$-tuples of positive-definite matrices of determinant $1$. Even for $d=1$, $P$ contains a totally geodesic submanifold isometric to the hyperbolic plane; as such the volumes of balls grow exponentially in their radius.\footnote{The volume of a ball can be computed exactly \cite{gual1999volume}, but the very crude bound of volume $\Omega(e^{\Theta(r) - O(n \log n)})$ for the geodesic ball of radius $r$ can be proved elementarily. The manifold $\PD(n) \cap \SL(n)$ contains the hyperbolic plane as a totally geodesic submanifold, in which the ball of radius $r$ has area $e^{\Theta(r)}$ \cite{cannon1997hyperbolic}. This shows the ball of radius $r$ in $\PD(n) \cap \SL(n)$ contains $\Omega(e^{\Theta(r)})$ balls of radius $1$, which themselves have volume at least $e^{- O(n \log n)}$ by comparison with the Euclidean ball. } The function $f_v:g \mapsto \|g \cdot v\|^2$ is convex along geodesics in this manifold \cite{gradflow}\footnote{This was implicitly shown much earlier in \cite{KempfNess}.}. The geodesics through a point $X \in P$ are given by $\gamma(t) = \sqrt{X} e^{H t} \sqrt{X}$ for Hermitian $H$. The Riemannian gradient $\nabla \log f_v (g)$ of $\log f_v$ at $g\in P$ is given by the moment map $\mu_G(g\cdot v)$. The geodesic ball of radius $R$ in $P$ about the identity is given by 
$$B_R:=\{e^{A}: A \text{ traceless, Hermitian}, \|A\|_F \leq R\} \subseteq P.$$

In a slight abuse of notation, we define the geodesic ball in $G$ (rather than $P$) to be $K B_R$, as in the introduction. The values taken by $f_v$ over $B_{2R}$ are the same as the values taken by $g \mapsto \|g \cdot v\|^2$ on $K B_R$. We now define diameter bounds.
\begin{dfn}\label{dfn:noncommDiameterBound} The diameter bound $D_f(\eps)$ for a function $f$ on $P$ and a real number $\eps > 0$ is defined as the infimum over $R > 0$ such that 
$$ \inf_{g \in B_R } f(g) \leq \eps + \inf_{g \in P}f(g).$$
\end{dfn}
We will show that the diameter bound for the norm-squared function can grow faster than $\poly(n, \log(1/\eps))$ for $d = 3$. Firstly, we need to review how diameter bounds for tensors in $(\RR_{\geq 0}^n)^d$ like that in \cref{thm:diameter} relate to diameter bounds for tensors in $(\CC^{\ot n})^d$ over $\SL(n)^d$ and $\ST(n)^d$. Infimizing $f_v(g)$ over the subset $P \cap \ST(n)^d \subseteq P$, or the tuples of positive-definite diagonal matrices within $\SL(n)^d$, results in a program of the form \cref{eq:abelian}. For $d = 3$, for example,
    \begin{equation} \inf_{g \in P \cap \ST(n)^3} \langle v, g \cdot v\rangle = \capa(p) = \inf_{x \in (\RR^n)^3} \sum_{\omega \in \Omega_{n,3}} p_\omega e^{\omega \cdot x} 
    = \inf_{x \in (\id^\perp_n)^3} \sum_{\omega \in \Omega_{n,3}} p_\omega e^{\omega \cdot x} \label{eq:restrict-abelian}
    \end{equation}
where $\Omega_{n,3} = \{(\eps_i,\eps_j,\eps_k) :i,j,k \in [n]\}$ and $p_{(\eps_i,\eps_j,\eps_k)} = |v_{ijk}|^2$. The correspondence is exactly $g = e^{\diag(x)}$ for $x \in (\id_n^\perp)^3$, which implies the following.

\begin{lem}\label{lem:equiv-abelian} For all $\eps> 0$, the diameter bound $D_f(\eps)$ for the function $f_v: g \mapsto \langle v, g \cdot v\rangle$ on $\ST(n)^3$ is equal to the diameter bound $D_h(\eps)$ of the function $f_p$ where $p_{ijk} = |v_{ijk}|^2$, or 
$$f_p \colon (\RR^n)^3 \to \RR, \quad x \mapsto \sum_{i,j,k \in [n]} |v_{ijk}|^2 e^{(\eps_i, \eps_j, \eps_k) \cdot x}.$$ 
\end{lem} 

Of course, there's nothing special about $d = 3$ here, and the lemma generalizes straightforwardly to other $d$. For instance, applying \cref{lem:equiv-abelian} for $d = 2$ shows that restricting operator scaling to diagonal matrices yields an instance of matrix scaling. We have shown how diameter bounds over $\ST(n)^d$ relate to those over $(\RR^n)^d$. Now we complete the chain by showing how to relate diameter bounds over $\SL(n)^d$ to those over $\ST(n)^d$. We will show that tensors with free support (defined in \cref{dfn:free}) have the same diameter bound over $\SL(n)^d$ as they do over $\ST(n)^d$, which by \cref{thm:diameter} and \cref{lem:equiv-abelian} we have shown can be superpolynomial. We then show that the construction from \cref{subsec:diameter-constr} is free.

\begin{thm} \label{thm:free-diameter}
Let $G$ denote $\SL(n)^d$, and let $\T$ denote $\ST(n)^d$. Suppose $\mu_{\T}(t \cdot v) = \mu_G(t \cdot v)$ for all $t \in \T$ (which holds if $v$ has free support). Then for any $R > 0$ we have
$$\inf_{g \in B_R} f_v(g) = \inf_{g \in \T \cap B_R} f_v(g),$$
where $B_R$ denotes the geodesic ball of radius $R$ about the identity in $G$.
\end{thm}

\begin{proof}
Define $B:=B_R$ and recall that $P$ denotes the positive-definite matrices in $G$. Let $f:  P \to \RR$ be given by $f:g \to \langle v, g \cdot v \rangle$. Clearly $\inf_{g \in B} f(g) \leq \inf_{g \in \T \cap B} f(g)$. We must show the converse inequality. Let $g^{*} := \argmin_{g \in  B} f(g)$. Recall that $P$ is a Hadamard manifold.
Define $\T_+$ to be $\T \cap P$. Let $\pi g^*$ denote the projection of $g^*$ to $\T_+$, that is, the closest point in $\T_+$ to $g^*$. As $\T_+$ is a geodesically convex set, projections to $\T_+$ are unique and distances decrease under the projection \cite[Theorem 2.1.12]{bacak2014convex}. Thus, $\pi g^* \in B$. If we can show that $f(\pi g^*) \leq f(g^*)$ then the proof is complete. 

Let $g^* = \exp_{\pi g^*} (x)$ for some $x$ in the tangent space $T_{\pi g^*}P$ to $P$ at $\pi g^*$. That is, $\gamma \colon [0,1] \to P,\; t \mapsto \exp_{\pi g^*} (tx)$ is the geodesic between $\pi g^*$ and $g^*$. Then, in the local inner product $\langle \cdot , \cdot \rangle_{\pi g^*}$ at $\pi g^*$, $x$ is orthogonal to the tangent space $T_{\pi g^*} \T_+ \subseteq T_{\pi g^*}P$ of $\T_+$ at $\pi g^*$, because $\pi g^*$ is a local minimum of the geodesically convex function $d(g^*, \cdot )^2$ on $\T_+$ and $x$ is proportional to the gradient of $d(g^*, \cdot )^2$ at $\pi g^*.$ 

The function $f$ is geodesically convex, and its gradient $\nabla f( \pi g^*)$ is proportional to the moment map $\mu_G(\pi g^* \cdot v)$. By the assumption that $\mu_{\T}(t \cdot v) = \mu_G(t \cdot v)$ for all $t \in \T$, $\mu_G(\pi g^* \cdot v)$ is in $i \Lie(\T_K)$, which is precisely the tangent space of $\T_+$ at $\pi g^*$. Thus
	$$f(g^*) = f(\exp_{\pi g^*}(x)) \geq f(\pi g^*) + \langle x, \nabla f(\pi g^*) \rangle_{\pi g^*} = f(\pi g^*),$$
which completes the proof. \end{proof}

\begin{lem} The support of the tensor $p$ from \cref{thm:diameter} is free. \end{lem}

\begin{proof} 
Recall that a tensor in $(\CC^n)^{\otimes 3}$ is free if and only if the supports of distinct rows of its weight matrix intersect in at most one element. The construction in \cref{prp:pad-diameter} preserves freeness, so we can consider the case $n = 3(l + 1)$ treated in the proof of \cref{thm:diameter}.
Recall that, in this case, the support of $p$ is $\Omega'_0 \cup \omega'$ where $\Omega_0'$ is the rows of a matrix $M$ defined from the directed graph $D_l$.
Each row in the matrix $M$ corresponds to some edge $D_l$. Let us first verify that $\Omega_0'$ is free. Assuming the rows correspond to the same edge, they can be verified to have intersection in at most one element, because the nonzero entries of the three rows corresponding to an edge are contained in a $3\times 6$ submatrix with the following form:
$$\begin{bmatrix} A & I \end{bmatrix} =  
\begin{bmatrix} 
0 & \cellcolor{green!30}1 & \cellcolor{green!30}1 & \cellcolor{green!30}1 & 0 & 0 \\ 
\cellcolor{green!30}1 & 0 & \cellcolor{green!30}1 & 0 & \cellcolor{green!30}1 & 0 \\ 
\cellcolor{green!30}1 & \cellcolor{green!30}1 & 0 & 0 & 0 & \cellcolor{green!30}1 
\end{bmatrix}
$$
Here the cells containing $1$ are colored for readability. Now consider the case that the rows belong to two different edges. If the two edges share no vertices, then clearly the corresponding edges do not intersect. Because the graph is a directed tree, edges may only share a vertex which is the sink of at least one of the edges. If the vertex is a sink for both edges, then the nonzero entries in the 6 rows belonging to either edge (after permutation) take the form 
$$\begin{bmatrix} 0 & A & I \\ A & 0 & I \end{bmatrix}  = \begin{bmatrix}
 0 & 0 & 0 & 0 &\cellcolor{green!30} 1 &\cellcolor{green!30} 1 & \cellcolor{green!30}1 & 0 & 0 \\
 0 & 0 & 0 & \cellcolor{green!30}1 & 0 & \cellcolor{green!30}1 & 0 & \cellcolor{green!30}1 & 0 \\
 0 & 0 & 0 & \cellcolor{green!30}1 & \cellcolor{green!30}1 & 0 & 0 & 0 & \cellcolor{green!30}1 \\
 0 & \cellcolor{green!30}1 & \cellcolor{green!30}1  & 0 & 0 & 0& \cellcolor{green!30}1 & 0 & 0 \\
 \cellcolor{green!30}1 & 0 & \cellcolor{green!30}1  & 0 & 0 & 0 & 0 & \cellcolor{green!30}1 & 0 \\
 \cellcolor{green!30}1 & \cellcolor{green!30}1 & 0  & 0 & 0 & 0 & 0 & 0 & \cellcolor{green!30}1
\end{bmatrix}.$$
If the shared vertex is a sink for only one edge, then the rows are
$$\begin{bmatrix} 0 & A & I \\ A & I & 0 \end{bmatrix}  = \begin{bmatrix}
 0 & 0 & 0 & 0 & \cellcolor{green!30}1 & \cellcolor{green!30}1 & \cellcolor{green!30}1 & 0 & 0 \\
 0 & 0 & 0 & \cellcolor{green!30}1 & 0 & \cellcolor{green!30}1 & 0 & \cellcolor{green!30}1 & 0 \\
 0 & 0 & 0 & \cellcolor{green!30}1 & \cellcolor{green!30}1 & 0 & 0 & 0 & \cellcolor{green!30}1 \\
 0 & \cellcolor{green!30}1 & \cellcolor{green!30}1  & \cellcolor{green!30}1 & 0 & 0  & 0 & 0 & 0\\
 \cellcolor{green!30}1 & 0 & \cellcolor{green!30}1   & 0 & \cellcolor{green!30}1 & 0 & 0 & 0 & 0\\
 \cellcolor{green!30}1 & \cellcolor{green!30}1 & 0   & 0 & 0 & \cellcolor{green!30}1& 0 & 0 & 0
\end{bmatrix}.$$

In all these cases it can be verified that supports of distinct rows intersect in at most one element. Lastly, we need to make sure that the intersection of the support of $\omega'$ with the support of any element of $\Omega_0'$ is at most one. Recall that $\omega'$ is defined to have entry one in each block corresponding to the leaves $u_l, v_l, w_l$ in $D_l$. However, there are no edges between the leaves, so the support of no row can intersect that of $\omega'$ in more than one element.
 \end{proof}
 
We are now nearly ready to prove \cref{thm:nc-diameter}. We would simply use the array $p$ from the proof of \cref{thm:diameter}, but setting $|v_{ijk}|^2 = p_{ijk}$ would not be solvable over the rationals. Therefore we must round $\sqrt{p_{ijk}}$, which requires some additional technical lemmas proven in \cref{sec:rounding}.

\begin{lem}[Rounding and diameter bounds]\label{lem:rounding} Let $p,q:\Omega \to \RR_{\geq 0}$ be positive functions on a finite set $\Omega \subseteq \RR^m$. Suppose there is a set $B$ such that 
$$ \inf_{x \in B} f_p(x) \geq (1 + \eps) \capa p,$$
and let $M = \max\{1/q_\omega, 1/p_\omega:\omega \in \Omega\}$. Then 
$$\inf_{x \in B} f_q(x) \geq ((1 +\eps)(1 - M \|p - q\|_\infty) -M \| p - q\|_1) \capa q.$$
\end{lem}

\begin{lem}[Rounding and capacity]\label{lem:cap-diff} Let $\Omega \subseteq \RR^m$ be finite and let $p, q:\Omega \to \RR_{\geq 0}$ be positive functions on $\Omega$. Let $M_0 = \max_{\omega \in \Omega} 1/q_\omega$. Then 
$$ \log \capa q \geq \log\capa p - M_0 \|p - q\|_\infty. $$

\end{lem}

\begin{proof}[Proof of \cref{thm:nc-diameter}] First recall that the values taken by $g \mapsto \|g \cdot v\|^2$ on the geodesic ball $K B_R$ in $G$ are the same as the values taken by $f_v: g \mapsto \langle v, g \cdot v \rangle$ on $B_{2R}$ in $P$. Thus it is enough to show that $f:=f_v$ has diameter bound $D_f(\eps) = \Omega(2^{n/3}\log(1/\eps))$ for $\eps \leq e^{- Cn^2 \log n}$.

We will apply \cref{lem:rounding} with $p$ as in the proof of \cref{thm:diameter} and $q_{ijk} = |v_{ijk}|^2$, with $v_{ijk}$ chosen so that $v$ has the same support as $p$ and $p_{ijk} - \delta < |v_{ijk}|^2 \leq  p_{ijk}$ for $\delta$ small.
Because $v$ is free, by \cref{thm:free-diameter} the diameter bound for $f_v$ is the same as the diameter bound for $f_v$ over $\ST(n)^3$. By \cref{lem:equiv-abelian}, this is the same as the diameter bound for $f_q$. It remains to show that $D_{f_q}(\eps) = \Omega(2^{n/3} \log (1/\eps))$ . We will do this by relating $D_{f_q}(\eps)$ to $D_{f_p}(\eps)$; in particular we will show $D_{f_q}(\Omega(\eps)) \geq D_{f_p}(\eps).$

Let $R = D_{f_p}(\eps)$. We have $\inf_{x \in (\RR^n)^3, \|x\| \leq R} f_p(x) \geq \capa(p) + \eps = (1 + 2\eps) \capa(p)$, recalling that $\capa(p) =1/2$. By \cref{lem:rounding}, 
    $$\inf_{x \in (\RR^n)^3, \|x\| \leq R} f_q(x) \geq ((1 + 2\eps)(1 - M \|p - q\|_\infty) -M \| p - q\|_1)\capa(q).$$

As $\capa q \leq 1/2$, if $M \|p - q\|_\infty \leq M \|p - q\|_1 \leq c \eps$ for $c$ a small enough constant, then we have $((1 +2\eps)(1 - M \|p - q\|_\infty) -M \| p - q\|_1) \capa q = \capa q + \Omega(\eps)$, so 
    $$\inf_{x \in (\RR^n)^3, \|x\| \leq R} f_q(x) \geq \capa q + \Omega(\eps).$$
Thus $D_{f_q}(\Omega(\eps)) \geq D_{f_p}(\eps)$ assuming $ M \|p - q\|_1 \leq c \eps$. To ensure that this constraint is satisfied, choose $v$ of bit complexity $O(\log n + \log (1/\eps))$ such that $\|p - q\|_1 = \frac{c}{n}\eps$. Because $p_{ijk}  = \Omega(1/n)$ for $i,j,k$ in the support of $p$ by construction, we have $q_{ijk} = \Omega(1/n)$ for $i,j,k$ in the support of $q$ and hence $M = O(n)$. Thus $M \|p - q\|_1 \leq c \eps$. Applying \cref{lem:cap-diff} together with our assumptions about the size of $p - q$ and the fact that $\capa(q) = \capa(v)$ implies the final claim that $\capa(v) \geq 1/4$ and that $1 \geq \|v\| \geq 1/2$.
\end{proof}

Finally, we remark that the same diameter bound holds for $d \geq 3$ for \emph{tuples} of tensors. We note that if $v  \in (\CC^{n})^{\ot 3}$ has free support, then so does the tensor $v \ot e_{l} \ot \dots \ot e_{l} \subset (\CC^n)^{\ot d}$ for $d \geq 3$. By \cref{prp:FreeForGapConstant}, the tuple $w \in ((\CC^{n})^{\ot d})^{n}$ given by 
$$w_l = \frac{1}{n} \; v \ot e_{l} \ot \dots \ot e_{l} \text{ for } l \in [n]$$
has $\mu_T(t \cdot v) = \mu_G(t \cdot v)$ for all $t \in \ST(n)^d$. The commutative problem obtained by restricting to $\SL(n)^d$ as in \cref{lem:equiv-abelian} is precisely $f_q$ as in \cref{cor:diameter-d}. As in the proof of \cref{thm:nc-diameter}, by \cref{thm:free-diameter}, \cref{lem:equiv-abelian} and \cref{cor:diameter-d}, we have the following.
\begin{cor}\label{cor:diameter-d-noncomm}
There is a constant $C > 0$ such that the following holds for all $d \geq 3$. For all $\eps \leq  \exp(- C n^2 \log n)$, there is a tuple of tensors $w = w(\eps) \in ((\CC^n)^{\ot d})^{n}$ with $O(n^2)$ nonzero entries of bit complexity $O(\log n + \log(1/\eps))$, and a geodesic ball $B=B(\eps)$ of radius $\Omega\left(2^{n/3}\log(1/\eps)\right)$ about the identity in $\SL(n)^d$, such that 
$$ \inf_{g \in B} \; \|g \cdot w\|^2 \geq \capa(v) + \eps.$$
Moreover, it holds that  $1/4 \leq \capa(w) \leq 1$ and $1/2 \leq \|w\| \leq 1$.

\end{cor}


\subsection{A bound on weight margin and gap for quivers}\label{subsec:quivers}

For $d \geq 2$ let $Q_d$ be the quiver
\[ \begin{tikzcd}[row sep = tiny]
1 \ar[r, leftarrow] & 2 \ar[r, rightarrow] & 3 \ar[r, dotted, no head ,thick] & d-2 \ar[r, rightarrow] & d-1 \ar[r, leftarrow] & d & \text{if } d \text{ even} \\
1 \ar[r, rightarrow] & 2 \ar[r, leftarrow] & 3 \ar[r, dotted, no head ,thick] & d-2 \ar[r, rightarrow] & d-1 \ar[r, leftarrow] & d & \text{if } d \text{ odd}.
\end{tikzcd} \]
and let $Q_{d}^{(k)}$ be the quiver one obtains from $Q_d$ by adding $k-1$ additional copies of each arrow in $Q_d$.
As before, let $G = \SL(n)^d$ and $\T = \ST(n)^d$. Then $G$ acts on the quiver $Q_d$ with dimension vector $(n, \ldots, n)$ as described in the introduction. We denote the corresponding representation by $\pi_{d}$. Note that the action of $G$ on $Q_d^{(k)}$ with dimension vector $(n, \ldots, n)$ is given by $\pi_d^{k}$.  In this subsection we prove a bound on the weight margin of $\pi_d$ and on the gap of $\pi_d^n$. The bound on $\gamma_{G}(\pi_d^n)$ is thanks to the refinement of freeness in Proposition~\ref{prp:FreeForGapConstant} pointed out by Visu Makam.

\begin{thm}\label{thm:UpperBoundQuiver}
Let $n, d \geq 2$ and denote the natural action of $G = \SL(n)^d$ on the quiver $Q_d$ with dimension vector $(n, \ldots, n)$ by $\pi_d \colon \SL(n)^d \to \GL(V_d)$, where $V_d = \left( \CC^{n \times n} \right)^{d-1}$. The representation $\pi_d^n$ corresponds to the $G$-action on the quiver $Q_d^{(n)}$ with dimension vector $(n,\ldots,n)$.
It holds that
\begin{align*}
\gamma_{\T}(\pi_d) \leq (n-1)^{-d+1} \qquad \text{and} \qquad \gamma_{G}(\pi_d^n) \leq (n-1)^{-d+1}.
\end{align*}
\end{thm}

\begin{rem} Before proving the theorem, we point out a few consequences.
	\begin{enumerate}
	\item Theorem~\ref{thm:UpperBoundQuiver} shows that $\gamma_{\T}(\pi_d)^{-1}$ and $\gamma_{G}(\pi_d^n)^{-1}$ are not polynomially bounded with respect to $\dim V_d = (d-1)n^2$ and $\dim \SL(n)^d = d(n^2 -1)$. Instead we see for fixed $n$ and $d \to \infty$ an exponential behaviour in the number of vertices $d$. Thus, our bound shows that the exponential behaviour in $d$ cannot be avoided in general lower bounds for quiver actions like \cite[Theorem~6.21 Item~4]{gradflow}. The latter applied to $\pi_d$ shows $\gamma_{\T}(\pi_d) \geq n^{-d^2-(3/2)d}(dn+1)^{-d}$.
	
	\item The proof of \cref{thm:UpperBoundQuiver} below shows that for the bound on the gap it is enough to consider the quiver $Q_d^{(n-1)}$ with an additional $n^{th}$ arrow from $d$ to $d-1$. 
	
	\item The ideas presented below can be adjusted to prove similar bounds for other dimension vectors. For example, one can show that the gap for the $\SL$-action on $Q_d^{(2)}$ with dimension vector $(1,3,3,\ldots,3,2)$ is inverse exponential in $d$. This aligns with an algebraic barrier for this action; the invariants that cut out the null cone for this action have exponential degree \cite[Proposition~1.5]{derksen2018degree}.
	
	\item The quiver $Q_d$ is of finite representation type and has no oriented cycles. Therefore, the null-cone membership problem for $\pi_d$ can be solved in polynomial-time by algebraic algorithms.\footnote{Personal communication with Visu Makam. There does not seem to be an explicit reference in the literature.} This means $Q_d$ is an example where the weight margin is very small but there still exist efficient algorithms. Can the existence of efficient algorithms still be explained by a large gap in this case? This leads to the following interesting open question.
	\end{enumerate}
\end{rem}
	
\begin{prb}
Is the gap $\gamma_G(\pi_d)$ inverse polynomial in $n$ \emph{and} $d$?
\end{prb}

A positive answer would provide an interesting example, since in this case the weight margin of $\pi_d$ would be \emph{significantly} smaller than the gap of $\pi_d$.

We now introduce several lemmas needed to prove Theorem~\ref{thm:UpperBoundQuiver}. Note that the set of weights of $\pi_d$ viewed as a subset of $(\RR^{n})^d$ is
\begin{align*}
\left\lbrace \big( (-1)^d \eps_i, (-1)^{d-1} \eps_j,0,\ldots,0 \big), \big( 0, (-1)^{d-1} \eps_i, (-1)^{d-2} \eps_j,0,\ldots,0 \big), \ldots, \big( 0,\ldots,0,\eps_i, - \eps_j \big) \mid i,j \in [n] \right\rbrace.
\end{align*}
We define recursively the subsets of weights
\begin{align*}
\Gamma_2 &:= \left\lbrace (\eps_i, -\eps_j) \mid i \in [n-1], \, j \in [n] \right\rbrace \subseteq \Omega(\pi_2) \subseteq \RR^{2n} \\
\text{for } d \geq 3, \; \Gamma_d &:= \left\lbrace \big( (-1)^d \eps_i, (-1)^{d-1} \eps_n,0_n,\ldots,0_n \big) \mid i \in [n-1] \right\rbrace \cup \big( \lbrace 0_n \rbrace \times \Gamma_{d-1} \big) \subseteq \Omega(\pi_d) \subseteq \RR^{dn} \, .
\end{align*}

\begin{rem}\label{rem:QuiverNotFree}
We note that for $d \geq 2$, $\Gamma_d$ is \emph{not} free. For instance, we can always write
	\begin{align*}
	(0_n,\ldots,0_n, \eps_1, - \eps_1) = (0_n,\ldots,0_n, \eps_1, - \eps_2) + (0_n,\ldots,0_n,0_n, e_2 - e_1),
	\end{align*}
i.e.  the weights $(0_n,\ldots,0_n, \eps_1, - \eps_1), \, (0_n,\ldots,0_n, \eps_1, - \eps_2) \in \Gamma_d$ differ by the root $(0_n,\ldots,0_n,0_n, e_2 - e_1)$ of $\SL(n)^d$.
Therefore, we \emph{cannot} deduce a bound on the gap $\gamma_G(\pi_d)$ via \cref{prp:FreeForGapConstant}. However, the latter allows us to deduce at least a bound on the gap of $\pi_d^n$.
\end{rem}

In the next two lemmas we show that $\Gamma_d$ witnesses the bound on $\gamma_{\T}(\pi_d)$ and afterwards we use \cref{prp:FreeForGapConstant} to transfer this bound to $\gamma_G(\pi_d^n)$.

\begin{lem}\label{lem:quiverConvHull}
For all $d \geq 2$ it holds that $0 \notin \conv(\Gamma_d)$.
\end{lem}

\begin{proof}
We prove the statement by induction on $d \geq 2$. For $d=2$, just note that any element in $\conv(\Gamma_2) \subseteq \RR^{2n}$ has value $-1/n$ in the $n$-th entry. In particular, $0 \notin \conv(\Gamma_2)$.
For $d \geq 3$ let
\begin{align*}
x = \sum_{\omega \in \Gamma_d} \lambda_\omega \, \omega \; , \quad \lambda_\omega \geq 0
\end{align*}
be a convex combination of the elements in $\Gamma_d$. Assume there is an $i \in [n-1]$ such that for
\begin{align*}
\omega_i := \big( (-1)^d \eps_i, (-1)^{d-1} \eps_n,0_n,\ldots,0_n \big)
\end{align*}
one has $\lambda_{\omega_i} > 0$. Then the $n$-th entry of $x$ is non-zero, since $\omega_i$ has $n$-th entry $(-1)^{d+1}/n$ and all (other) $\omega \in \Gamma_d$ have $(-1)^{d+1}/n$ or zero as $n$-th entry.
On the other hand, if $\lambda_{\omega_i} = 0$ for all $i \in [n-1]$, then $x \in \lbrace 0_n \rbrace \times \conv(\Gamma_{d-1})$. By induction hypothesis on $d-1$ we necessarily have $x \neq 0$.
\end{proof}

\begin{lem}\label{lem:quiverDist}
For $d \geq 2$ it holds that $x_d := \lambda_d \big( (-1)^{d-1} \eps_n, 0_n, \ldots, 0_n \big) \in \conv(\Gamma_d)$, where \begin{align*}
\lambda_d := \left( \sum_{i=1}^{d-1} (n-1)^i \right)^{-1} \, .
\end{align*}
In particular, $\|x_d\|_2 < \vert \lambda_d \vert \leq (n-1)^{-d+1}$.
\end{lem}

\begin{proof}
We proceed by induction on $d \geq 2$. In the case $d=2$, consider the convex combination
\begin{align*}
\sum_{i=1}^{n-1} \sum_{j=1}^n \frac{1}{(n-1)n} (\eps_i,-\eps_j)
= \frac{1}{n-1} (-\eps_n, 0_n) = x_2 \, ,
\end{align*}
where we used \eqref{eq:SumEps-i}.
Now assume the claim is proven for some $d \geq 2$, hence
\begin{equation}\label{eqConvCombGammad}
\lambda_d \big( 0_n, (-1)^{d-1} \eps_n, 0_n, \ldots, 0_n \big) \in \lbrace 0_n \rbrace \times \conv(\Gamma_d) \subseteq \conv(\Gamma_{d+1}).
\end{equation}
Setting $\mu := (n-1)\lambda_{d+1} \lambda_d^{-1}$ we have $\mu \lambda_d = (n-1)\lambda_{d+1}$ and $\mu + (n-1)\lambda_{d+1} = 1$. Together with \eqref{eq:SumEps-i} and \eqref{eqConvCombGammad} we deduce $x_{d+1} \in \conv(\Gamma_{d+1})$ via
\begin{align*}
\mu \, \lambda_d \big( 0_n,(-1)^{d-1} \eps_n, 0_n, \ldots, 0_n \big) + \lambda_{d+1} \sum_{i=1}^{n-1} \big( (-1)^{d+1} \eps_i,(-1)^{d}\eps_n, 0_n, \ldots, 0_n \big)
= x_{d+1}.
\end{align*}
This ends the induction. Finally, $\|x_d\|_2 < \vert \lambda_d \vert$ follows from $\|\eps_n\|_2 < 1$.
\end{proof}

\begin{proof}[Proof of \cref{thm:UpperBoundQuiver}]
By \cref{lem:quiverConvHull} and \cref{lem:quiverDist} we have 
	\begin{align*}
	\gamma_{\T}(\pi_d) \leq (n-1)^{-d+1}.
	\end{align*}
With the fact $\Omega(\pi_d) = \Omega(\pi_d^n)$ and with \cref{prp:FreeForGapConstant} we transfer this bound to the gap of $\pi_d^n$. To do so, we note that the natural inner product on $V_d^n = (\CC^{n \times n})^{n(d-1)}$, given by the trace inner product on each $\CC^{n \times n}$ copy, is invariant under the action of $K = \SU(n)^d$. Clearly, distinct $\CC^{n \times n}$ copies are orthogonal under this inner product. Thus, to be able to apply \cref{prp:FreeForGapConstant} it is enough to assign to each $\CC^{n \times n}$ copy, i.e. to each arrow of $Q_d^{(n)}$, a matrix $M_i$ such that $\supp(M_i)$ is free and $\Gamma_d = \bigcup_i \supp(M_i)$.

For this, we consider the $n \times n$ matrices
	\begin{align*}
	M := \begin{pmatrix} I_{n-1} & 0 \\ 0 & 0 \end{pmatrix} \qquad \text{ and } \qquad
	P := \begin{pmatrix} 0 & I_{n-1} \\ 1 & 0 \end{pmatrix},
	\end{align*}
and $E_{i,j}$ is the matrix with $(i,j)$-entry one and all other entries zero. Then $E_{i,i}P = E_{i,\sigma(i)}$, where $\sigma \colon [n] \to [n]$ is the cycle $(1 \; 2\; \ldots \; n)$. Therefore, for $k \in [n]$ we have
	\begin{align*}
	\supp \left(MP^{k-1} \right) = \left\lbrace \big( 0_{n(d-2)}, \eps_i, -\eps_{\sigma^{k-1}(i)} \big)  \mid i \in [n-1] \right\rbrace
	\text{ and }  \lbrace 0_{n(d-2)} \rbrace \times \Gamma_2 = \bigcup_{k \in [n]} \supp \left(MP^{k-1} \right).
	\end{align*}
For fixed $k$, $i_1 \neq i_2$ implies $\sigma^{k-1}(i_1) \neq \sigma^{k-1}(i_2)$, so any distinct elements of $\supp(MP^{k-1})$ differ in the last two $\RR^n$-components. Hence, each $\supp(MP^{k-1})$ is free and we assign $M, MP, \ldots, MP^{n-1}$ to the $n$ arrows that go from vertex $d$ to vertex $d-1$. For $l \in [d-2]$, we assign to the $n$ arrows between the vertices $l$ and $l+1$ each of the matrices $E_{1,n}, E_{2,n}, \ldots, E_{n-1,n}$ at least once. (Exactly one of the latter matrices is assigned to two of these arrows.) Clearly, the support of $E_{i,n}$, $i \in [n-1]$ is free as it contains just one weight. By construction, this assignment does the job. Moreover, the argument shows that $n-1$ arrows between the vertices $l$ and $l+1$, $l \in [d-2]$, suffice.
\end{proof}

\appendix


\section{Notation}

\begin{longtable}{ r p{13cm} }
 $f_p$ & the function $\RR^m \to \RR_{\geq0}, x \mapsto \sum_{\omega \in \Omega} p_\omega e^{\omega \cdot x}$, see \cref{eq:abelian}  \\
 $\capa(p)$ & the capacity of a non-negative function $p$ on a finite set $\Omega \subseteq \RR^m$, see \cref{eq:abelian} \\
 $\capa(v)$ & the capacity of a vector $v$ under a group action, see \cref{eq:capacity}\\
 $[n]$ & the set $\{1,2,\ldots,n \}$ \\
 $0_n$ & the zero vector in $\RR^n$ \\
 $e_i$ & the $i^{th}$ canonical unit vector in $\RR^n$\\
 $\id_n$ & the all-ones vector in $\RR^n$ \\
 $\id_n^\perp$ & the orthogonal complement of $\id_n$ in $\RR^n$, i.e. $\left\lbrace (v_1,\ldots,v_n) \in \RR^n \colon \sum_i v_i =0 \right\rbrace$\\
 $\eps_i$ & the vector $e_i - \frac{1}{n} \id_n$\\
 $I_n$ & the $n \times n$ identity matrix\\
 $\dist(0,S)$ & the distance from the origin to the set $S$\\
 $\conv(S)$ & the convex hull of $S$ in $\RR^n$\\
 $\aff(S)$ & the affine hull of $S$ in $\RR^n$\\
 $\pi_{n,d}$ & the representation for $d$-dimensional tensor scaling\\
 $\Omega(\pi)$ & the set of weights of a representation $\pi$\\
 $\Omega_{n,d} = \Omega(\pi_{n,d})$ & the set $\left\lbrace \eps_i \colon i \in [n] \right\rbrace^d$ corresponding to $d$-dimensional array scaling; equal to the set of weights of the tensor scaling representation $\pi_{n,d}$, see \cref{exa:TensorWeightMarginAndGap}\\
 $\gamma(\Omega)$ & the \emph{margin} of the finite set $\Omega \subseteq \RR^m$, see \cref{dfn:margin}\\
 $\gamma_{\T}(\pi)$ & the \emph{weight margin} of a representation $\pi$, i.e. $\gamma(\Omega(\pi))$, see \cref{dfn:WeightMarginGapConstant}\\
 $\gamma_{G}(\pi)$ & the \emph{gap} of a representation $\pi$, see \cref{dfn:WeightMarginGapConstant}\\
 $\tr(A)$ & the trace of a square matrix $A$\\
 $D_f(\eps)$ & the diameter bound of a function $f$ for $\eps > 0$, see \cref{dfn:diameterBound} respectively \cref{dfn:noncommDiameterBound}\\
 $\|A\|_F$ & the Frobenius norm of a square matrix $A$\\
 $e^A$ & the exponential of a square matrix $A$\\
 $\Lie(G)$ & the Lie algebra of a matrix Lie group $G$\\
 $\GL(n)$ & the group of invertible \emph{complex} $n \times n$ matrices\\
 $\SL(n)$ & the group of invertible \emph{complex} $n \times n$ matrices with determinant one\\
 $\ST(n)$ & the group of diagonal invertible \emph{complex} $n \times n$ matrices with determinant one\\
 $\SU(n)$ & the group of unitary matrices of size $n \times n$ and determinant one\\
 $\Herm(n)$ & the set of complex Hermitian $n \times n$ matrices \\
 $\GL(V)$ & the group of $\CC$-linear, bijective maps $V \to V$, where $V$ is a $\CC$-vector space
\end{longtable}

\section{Representation theory background}\label{sec:RepTheoryBackground}

In this section we briefly recall some representation theory. All the concepts we present here actually work in the very general setting of reductive groups and their rational representations, see e.g. \cite[section~2]{gradflow}. For the sake of clarity and concreteness we stick to the special case needed in this paper, i.e. the reductive group $\SL(n)^d := \SL(n) \times \cdots \times \SL(n)$ with $d \geq 1$ many copies of $\SL(n)$.

We call a Euclidean-closed subgroup $H \subseteq \GL(n)$ a \emph{matrix Lie group}. Indeed, such an $H$ is naturally a Lie group (c.f. \cite[Theorem~1.19]{HallBook}) with real \emph{Lie algebra}
	\begin{align*}
	\Lie(H) := \left\lbrace A \in \CC^{n \times n} \mid \forall \; t \in \RR \colon   e^{tA} \in H \right\rbrace.
	\end{align*}
The Lie bracket for $\Lie(H)$ is the commutator $[A, B] := AB - BA$. Moreover, for $d \geq 1$ the product $H^d := H \times \cdots \times H$ becomes a matrix Lie group via block-diagonal embedding into $\GL(dn)$, i.e.
	\begin{align*}
	H^d \hookrightarrow \GL(dn), \quad (h_1, \ldots, h_d) \mapsto
	\begin{pmatrix}
	h_1 &  &  \\ 
	 & \ddots &  \\ 
	 &  & h_d
	\end{pmatrix} 
	\end{align*}
Then the Lie algebra of $H^d$ is $\Lie(H)^d = \Lie(H) \times \cdots \times \Lie(H)$ block-diagonally embedded into $\CC^{dn \times dn}$.
If $G \subseteq \GL(n)$ is another matrix Lie group, then $G \cap H$ is again a matrix Lie group with Lie algebra $\Lie(G \cap H) = \Lie(G) \cap \Lie(H)$.

\begin{exa}\label{exa:LieGroups}
The groups $\GL(n)$, $\SL(n)$, $\Un(n)$ and $\GT(n)$ are matrix Lie groups with Lie algebras
	\begin{align*}
	\Lie(\GL(n)) &= \CC^{n \times n} &\quad
	\Lie(\Un(n)) &= \lbrace A \in \CC^{n \times n} \mid A^\dagger = -A \rbrace  = i \Herm(n) \\
	\Lie(\SL(n)) &= \lbrace A \in \CC^{n \times n} \mid \tr(A) = 0 \rbrace &\quad
	\Lie(\GT(n)) &= \lbrace A \in \CC^{n \times n} \mid A \text{ diagonal matrix} \rbrace.
	\end{align*}
Therefore, also $\SU(n)$, $\ST(n)$ and $\Un(n) \cap \ST(n)$ are matrix Lie groups and their Lie algebras are obtained by corresponding intersections of the above Lie algebras. In particular, we have
	\[ \Lie(\Un(n) \cap \ST(n)) = \big\lbrace i \diag(x_1, \ldots, x_n) \mid x_j \in \RR,  x_1 + \ldots + x_n = 0 \big\rbrace. \]
Thus, we can identify $i \Lie(\Un(n) \cap \ST(n))$ with the orthogonal complement $(\id_n)^\perp \subseteq \RR^n$ of the all-ones vector $\id_n$.
\end{exa}

In the following, let $G := \SL(n)^d$ for some $d \geq 1$. Then $K := \SU(n)^d$ is a maximal compact subgroup of $G$, and $\T := \ST(n)^d$ and $\T_K := K \cap \T$ are maximal tori of $G$ and $K$, respectively. As explained above, we think of all these groups as matrix Lie subgroups of $\GL(dn)$, and hence of their Lie algebras as subsets of $\CC^{dn \times dn}$.

A \emph{rational representation} of $G = \SL(n)^d$ is a group morphism $\pi \colon G \to \GL(V)$, such that in some basis of $V$ the matrix entries of $\pi(g) \in \GL(V)$ are polynomials in the matrix entries of $g$.\footnote{In other words, $\pi$ is a morphism of affine algebraic groups.} Such a rational representation of $G$ induces a representation of the Lie algebras by
	\begin{align*}
	\Pi \colon \Lie(G) \to \End(V), \quad A \mapsto \left. \frac{d}{dt} \right|_{t=0} \pi \left( e^{tA} \right)
	\end{align*}
with the property $\pi(e^A) = e^{\Pi(A)}$ for all $A \in \Lie(G)$.
Restricting $\pi$ to the commutative subgroup $\T$ induces a so-called \emph{weight space decomposition} of $V$. That is, there is some finite set $\Omega(\pi) \subseteq i \Lie(\T_K)$ and a decomposition $V = \bigoplus_{\omega \in \Omega(\pi)} V_\omega$ into non-zero subspaces such that each $\omega \in \Omega(\pi)$ and any $v_\omega \in V_\omega$ satisfy
	\begin{align*}
	\forall A \in \Lie(\T)  \colon \quad
	\pi \left( e^A \right) v_\omega = e^{ \tr(A \omega)} v_\omega
	\end{align*}
or, equivalently,
	\begin{align*}
	\forall A \in \Lie(\T)  \colon \quad
	\Pi \left( A \right) v_\omega =  \tr(A \omega) v_\omega.
	\end{align*}
The elements $\omega \in \Omega(\pi)$ are called \emph{weights} of $\pi$ and the $v_\omega \in V_\omega$ are called \emph{weight vectors}. Considering \cref{exa:LieGroups} we frequently use the identification $i\Lie(\T_K) \cong (\id_n^\perp)^d$, where $\id_n^\perp$ is the orthogonal complement of $\id_n$ in $\RR^n$. We note that for $\omega \in i\Lie(\T_K) \subseteq \CC^{dn \times dn}$ the Frobenius norm $\| \omega \|_F$ becomes under this identification the 2-norm $\| \omega \|_2$ in $(\RR^n)^d$.

\begin{exa}\label{exa:LeftMultSL}
Let $d=1$. The group $G = \SL(n)$ acts on $\CC^n$ by left-multiplication, which induces the rational representation $\pi \colon \SL(n) \to \GL(n), g \mapsto g$ with corresponding Lie algebra representation $\Pi \colon \Lie(\SL(n)) \to \CC^{n \times n}, A \mapsto A$. For $i \in [n]$ we set
	\begin{align*}
	\eps_{i} := e_i - \frac{1}{n} \id_n \in \id_n^\perp \subseteq \RR^n.
	\end{align*}
For all $A = \diag(a_1, \ldots, a_n) \in \Lie(\T)$ and all $i \in [n]$
	\begin{align*}
	\pi \left( e^{A} \right) e_i = \diag(e^{a_1}, \ldots, e^{a_n}) e_i = e^{a_i} e_i
	\overset{(*)}{=} e^{\tr(A \diag(\eps_i))} e_i
	\end{align*}
where we used $a_1 + \ldots + a_n = 0$ in $(*)$. Thus, $\eps_i \in \id_n^\perp \cong i\Lie(\T_K)$ is a weight of $\pi$ with weight vector $e_i$. Since $\CC^n = \bigoplus_i \CC e_i$, we deduce
	$
	\Omega(\pi) = \lbrace \eps_i  \mid i \in [n] \rbrace
	$.
\end{exa}

\begin{exa}\label{exa:Roots}
Of particular importance in representation theory is the \emph{adjoint representation}. That is, $G = \SL(n)^d$ acts on its Lie algebra by conjugation $\mathrm{Ad} \colon G \to \GL(\Lie(G)), g \mapsto (A \mapsto gA g^{-1})$, which induces the representation of Lie algebras $\mathrm{ad} \colon \Lie(G) \mapsto \End(\Lie(G)), A \mapsto (B \mapsto [A,B])$. The non-zero weights $\alpha \in \Omega(\mathrm{Ad})$ are called \emph{roots} of $G$ and the weight spaces $\Lie(G)_{\alpha}$ are called \emph{root spaces}.

Let $d=1$ and for $i,j \in [n]$ denote by $E_{i,j}$ the matrix with entry one at position $i,j$ and all other entries being zero. Then for $i,j \in [n]$ with $i \neq j$ and for all $A = \diag(a_1, \ldots, a_n), B \in \Lie(\T)$ we compute
	\begin{align*}
	\mathrm{ad}(A)E_{i,j} &= [A,E_{i,j}] = (a_i - a_j) E_{i,j} = \tr \big( A \diag(e_i - e_j) \big) E_{i,j}, \\
	\mathrm{ad}(A)(B) &= [A, B] = 0.
	\end{align*}
Since $0_n, e_i - e_j \in \id_n^\perp \cong i\Lie(\T_K)$, we deduce $e_i - e_j \in \Omega(\mathrm{Ad})$ with weight vector $E_{i,j}$ and $0_n \in \Omega(\mathrm{Ad})$ with weight vector $B \in \Lie(\T)$. Therefore, the set of roots of $\SL(n)$ is $\lbrace e_i - e_j \mid i,j \in [n], i\neq j \rbrace$, because $\Lie(G) = \Lie(\T) \oplus \bigoplus_{i \neq j} \CC E_{i,j}$.

More generally, one can deduce that the roots of $G = \SL(n)^d$ are the 
	\begin{align*}
	(e_i - e_j, 0_n, \ldots, 0_n), (0_n, e_i - e_j, 0_n, \ldots, 0_n), \ldots \ldots, (0_n, \ldots, 0_n, e_i - e_j) \in \left( \RR^n \right)^d
	\end{align*}
for $i,j \in [n]$ with $i \neq j$ and that $\Lie(G) = \Lie(\T) \oplus \bigoplus_\alpha \Lie(G)_\alpha$.
\end{exa}

We need the following property of roots, see e.g. \cite[Lemma~7.11]{HallBook}.

\begin{prp}\label{prp:Roots}
Let $\alpha$ be a root of $G = \SL(n)^d$ and let $\pi \colon G \to \GL(V)$ be a rational representation of $G$. If $V_\omega$ is the weight space of some weight $\omega \in \Omega(\pi)$, then
	\begin{align*}
	\Pi \big( \Lie(G)_\alpha \big) (V_\omega) \subseteq V_{\omega + \alpha},
	\end{align*}
where $V_{\omega + \alpha} := \{0\}$, if $\omega + \alpha \notin \Omega(\pi)$.
\end{prp}

\section{Padding for tensor margin and tensor gap}\label{subsec:extension}

The Theorems~\ref{thm:MarginTensor} and \ref{thm:GapConstantTensor} only give for all $n \geq 2$ bounds for certain sub-families of $\{ (n,d) \mid d \geq 3 \}$. Still, we can deduce Theorems~\ref{thm:tensor-margin} and \ref{thm:tensor-gap} via some padding on the number of tensor factors $d$;  that padding is provided in \cref{prp:dTensorsPadding} below. Recall the representation for tensor scaling
	\begin{align*}
	\pi_{n,d} \colon \SL(n)^d \to \GL\left( (\CC^n)^{\otimes d} \right), \; (g_1, \ldots, g_d) \mapsto g_1 \otimes \cdots \otimes g_d,
	\end{align*}
which set of weights is $\Omega(\pi_{n,d}) = \Omega_{n,d} = \lbrace \eps_{i} \mid i \in [n] \rbrace^d \subseteq (\RR^n)^d$.

\begin{prp}\label{prp:dTensorsPadding}
Let $G:= \SL(n)^d$ and $n,d \geq 1$. Consider a set of weights $\Gamma_{n,d} \subseteq \Omega_{n,d}$ such that $0 \notin \conv(\Gamma_{n,d})$, i.e. $\Gamma_{n,d}$ witnesses the inequality $\gamma(\Omega_{n,d}) = \gamma_{\T}(\pi_{n,d}) \leq \dist (0, \conv(\Gamma_{n,d}))$.
	\begin{enumerate}
	\item  Then $\gamma(\Omega_{n,d+1}) \leq \dist \big(0, \conv(\Gamma_{n,d}) \big)$. Consequently, $\gamma(\Omega_{n,d+1}) \leq \gamma(\Omega_{n,d})$.
	\item If additionally $\Gamma_{n,d}$ is free, then $\gamma_{G}(\pi_{n,d+r}) \leq \dist \big(0, \conv(\Gamma_{n,d}) \big)$ for all $r \geq 2$.
	\end{enumerate}

\end{prp}

\begin{proof}
To prove the statement we set for $r \geq 1 $
	\[ \Delta_r := \{ (\eps_i,\ldots,\eps_i) \mid i \in [n] \} \subseteq (\RR^n)^r \qquad \text{and} \qquad
	\Gamma_{n,d+r} := \Gamma_{n,d} \times \Delta_r \subseteq \Omega(\pi_{n,d+r}) . \]
By \cref{eq:SumEps-i} we have $0 \in \conv(\Delta_r)$ and therefore
	\begin{align*}
	\conv(\Gamma_{n,d+r}) = \conv(\Gamma_{n,d}) \times \conv(\Delta_r)
	\supseteq \conv(\Gamma_{n,d}) \times \{0\}.
	\end{align*}
The latter implies
	\begin{equation}\label{eq:Padding}
	\dist \big( 0, \conv(\Gamma_{n,d+r}) \big) \leq \dist \big( 0, \conv(\Gamma_{n,d}) \big).
	\end{equation}
Clearly, $0 \in \conv(\Gamma_{n,d+r})$ implies $0 \in \conv(\Gamma_{n,d})$ or, by contraposition, the assumption $0 \notin \conv(\Gamma_{n,d})$ yields $0 \notin \conv(\Gamma_{n,d+r})$. The latter for $r=1$ shows $\gamma_{\T}(\pi_{n,d+1}) \leq \dist \big(0, \conv(\Gamma_{n,d + 1}) \big)$ and we conclude the first assertion with \cref{eq:Padding}.

Assume in addition that $\Gamma_{n,d}$ is free and let $r \geq 2$. Considering \cref{dfn:free} and \cref{prp:FreeTensorVsFreeGeneral} we prove that also $\Gamma_{n,d+r}$ is free. For this, let $M \subseteq [n]^d$ be such that $\Gamma_M = \Gamma_{n,d}$ and consider $(x,i,\ldots,i), (y,j,\ldots,j) \in M \times [n]^r$ with $(x,i,\ldots,i) \neq (y,j,\ldots,j)$. If $x \neq y$, then $x$ and $y$ differ in at least two components by freeness of $M$. If $x = y$, then we have $i \neq j$ and so $(x,i,\ldots,i)$ and $(y,j,\ldots,j)$ differ in at least two components, using $r \geq 2$. This shows that $\Gamma_{n,d+r}$ is free for $r \geq 2$. Since also $0 \notin \conv(\Gamma_{n,d+r})$ we obtain with \cref{prp:FreeForGapConstant} that $\gamma_{G}(\pi_{n,d+r}) \leq \dist \big(0, \conv(\Gamma_{n,d + r}) \big)$ holds for all $r \geq 2$. Finally, we deduce the second statement using \cref{eq:Padding}.
\end{proof}

\begin{prp}\label{prp:4Tensors}
For $n \geq 3$ it holds that $\gamma_{\T}(\pi_{n,4}) \leq \gamma_G(\pi_{n,4}) \leq 2^{-n+1}$.
\end{prp}

\begin{proof}
This result can be obtained by imitating the proof of \cref{thm:MarginTensor}(b) in subsection~\ref{sec:3tensors} by using
	\begin{align*}
	\Gamma_{n,4} := \lbrace (\eps_i,\eps_j,\eps_k, \eps_i) \mid (i,j,k) \in \mathfrak{W}_n \rbrace \subseteq \Omega(\pi_{n,4}).
	\end{align*}
Clearly, $0 \notin \conv(\Gamma_{n,4})$ as $0 \notin \conv(\Gamma_{n,3})$ by \cref{lem:affineHullKravtsov}. Moreover, one can show with \cref{lem:Kravtsov} (similar to the proof of \cref{lem:distKravtsov}) that
	\begin{align*}
	x:= -\frac{1}{c \, 2^{n-1}} (\eps_1,\eps_1,\eps_1, \eps_1) \in \conv(\Gamma_{n,4}), \quad \text{where }\;\;
	c = n-2^{-n+1} \geq 2.
	\end{align*}
Thus, $\norm{(\eps_1,\eps_1,\eps_1, \eps_1)} \leq \sqrt{4}$ implies $\norm{x} \leq c^{-1} 2^{-n+1} \sqrt{4} \leq 2^{-n+1}$. This proves $\gamma_{\T}(\pi_{n,4}) \leq 2^{-n+1}$.

Since $\mathfrak{W}_n$ is free by \cref{prp:WnFree}, the set $\lbrace (i,j,k,i) \mid (i,j,k) \in \mathfrak{W}_n \rbrace$ is free. Hence, we conclude $\gamma_{G}(\pi_{n,4}) \leq 2^{-n+1}$ with \cref{prp:FreeTensorVsFreeGeneral} and \cref{prp:FreeForGapConstant}.
\end{proof}


\section{Proof of \cref{lem:convStackingKravtsov}}\label{sec:StackingKravtsovGeneral}

\begin{proof}
For the sake of contradiction assume that $0 \in \aff(\Gamma_{n, 6 r - 3})$. Then there are coefficients $a_s, b_s, c_s \in \RR$, where $2 \leq s \leq r n$, such that $a_2 = \ldots = a_r = b_2 = \ldots = b_r = 0$, $\sum_s (a_s + b_s + c_s) = 1$ and
	\begin{align}
	\sum_{s= 2}^{r n} \left( a_s \, \eps_{\sigma(s),\sigma(1),\sigma(s)}
	+ b_s \, \eps_{\sigma(s),\sigma(s),\sigma(1)} + c_s \, \eps_{\sigma(s-1),\sigma(s),\sigma(s)}  \right) = 0 \in (\RR^n)^{6 r - 3}.\label{eq:main-combo}
	\end{align}
	The bulk of our work will consist of proving the equations
	\begin{align}\label{eq:StackBandC}
	b_2 + c_2 &= b_3 + c_3 = \ldots = b_{r n} + c_{r n}\\
	\label{eq:StackAandC} a_2 + c_2 &= a_3 + c_3 = \ldots = a_{r n} + c_{r n}.
	\end{align}
	From here we will derive a contradiction. We now set about proving \cref{eq:StackAandC,eq:StackBandC}. Rewrite the left-hand-side of \cref{eq:main-combo} as the collection for $k \in [2 r - 1]$ of the following affine linear combinations of $\eps_1,\ldots,\eps_n$ in $\RR^n$:
	\begin{align}
	\sum_{s= 2}^{r n} \left( a_s \, \eps_{\sigma_k(s)} 
	+ b_s \, \eps_{\sigma_k(s)} + c_s \, \eps_{\sigma_k(s-1)}  \right) &= 0 \label{eq:StackComp1}\\
	\sum_{s= 2}^{r n} \left( a_s \, \eps_{\sigma_k(1)}
	+ b_s \, \eps_{\sigma_k(s)} + c_s \, \eps_{\sigma_k(s)}  \right) &= 0 \label{eq:StackComp2} \\
	\sum_{s= 2}^{r n} \left( a_s \, \eps_{\sigma_k(s)}
	+ b_s \, \eps_{\sigma_k(1)} + c_s \, \eps_{\sigma_k(s)}  \right) &= 0. \label{eq:StackComp3}
	\end{align}
If we expand this expressions as affine linear combinations of the $\eps_l$, then by Lemma~\ref{lem:convCombEps-i} the coefficient of $\eps_l$ must be $n^{-1}$ for all $l \in [n]$. Translating this for equations \eqref{eq:StackComp1}, \eqref{eq:StackComp2} and \eqref{eq:StackComp3} respectively with $2 \leq l \leq n$ and $k \in [r]$, and using for $j \in [r]$ that
	\begin{align}
	\sigma_{k} \big(r(l-1)+j - k + 1 \big) = \left\lceil \frac{(r(l-1)+j - k + 1) + (k-1)}{r} \right\rceil = l \label{eq:sigma-inverse}
	\end{align}
we get
	\begin{align}
	&\forall \, k \in [r], l \in \{2,3,\ldots,n\} \colon  &\sum_{j=1}^r \big( a_{r(l-1)+j - k + 1} + b_{r(l-1)+j - k + 1} + c_{r(l-1)+j- k + 2} \big) &= \frac{1}{n} \label{eq:Coeff1} \\
	&\forall \, k \in [r], l \in \{2,3,\ldots,n\} \colon &\sum_{j=1}^{r} \big( b_{r(l-1)+j - k + 1} + c_{r(l-1)+j - k + 1} \big) &= \frac{1}{n} \label{eq:Coeff2} \\
	&\forall \, k \in [r], l \in \{2,3,\ldots,n\} \colon &\sum_{j=1}^{r} \big( a_{r(l-1)+j - k + 1} + c_{r(l-1)+j - k + 1} \big) &= \frac{1}{n} \label{eq:Coeff3}
	\end{align}
respectively, where we set  $c_{r n + 1} := 0$. Fixing some $l \geq 2$ and subtracting \cref{eq:Coeff2} with $k = 1$ from \cref{eq:Coeff2} for $k= 2$, we find a telescoping sum that reduces to $b_{r(l-1)} + c_{r(l-1)} = b_{r l} + c_{r l}$. Indeed, subtracting the two yields
\begin{align*}
	0 &= \sum_{j=1}^{r} \big( b_{r(l-1)+j - 1} + c_{r(l-1)+j - 1} \big) - \sum_{j=1}^{r} \big( b_{r(l-1)+j} + c_{r(l-1)+j} \big) \\
	&= \sum_{j=0}^{r-1} \big( b_{r(l-1)+j } + c_{r(l-1)+j } \big)  - \sum_{j=1}^{r} \big( b_{r(l-1)+j} + c_{r(l-1)+j} \big) \\
	&= ( b_{r(l-1)} + c_{r(l-1)}) - (b_{r l} + c_{r l}).
	\end{align*}
More generally, for $k \in  [r-1]$ combining \eqref{eq:Coeff2} for $k$ and $k \leftarrow k + 1$, implies 
	$b_{r l - k + 1} + c_{r l - k + 1} = b_{r(l-1)- k + 1} + c_{r(l-1)- k + 1}$
for all $l = 2,\ldots,n$, i.e. for every $k \in [r - 1]$ we have
	\begin{equation}\label{eq:StackBandC1}
	c_{r - k + 1} = b_{r - k + 1} + c_{r - k+ 1} = b_{2r - k + 1} + c_{2r - k + 1} = \ldots = b_{r n - k + 1} + c_{r n - k + 1}.
	\end{equation}
We are still missing the value $k = 0$, or the equations 
\begin{equation}\label{eq:StackBandC2}
	b_{r+1} + c_{r+1} = b_{2r + 1} + c_{2r + 1} = \ldots = b_{r(n-1) + 1} + c_{r(n-1) + 1}.
	\end{equation}
We obtain this by subtracting, for $l = 2, \dots, n$,  \eqref{eq:Coeff2} for $k = 1$ and $l$ from \eqref{eq:Coeff2} with $k = r$ and $l \leftarrow l + 1$ . Indeed, 
\begin{align*}0 &= \sum_{j=1}^{r} \big( b_{r l+j - r + 1} + c_{r l +j - r + 1} \big) - \sum_{j=1}^{r} \big( b_{r(l-1)+j} + c_{r(l-1)+j} \big) \\
&= \sum_{j=2}^{r+1} \big( b_{r(l-1)+j } + c_{r(l-1) +j}\big) - \sum_{j=1}^{r} \big( b_{r(l-1)+j} + c_{r(l-1)+j} \big) \\
&= \big( b_{r l + 1} + c_{r l + 1}\big) - \big( b_{r(l-1)+1} + c_{r(l-1)+1} \big).
\end{align*}

Lastly, we are missing the equations $b_2 + c_2 = b_3 + c_3 = \ldots = b_{r+1} + c_{r+1}$  for \cref{eq:StackBandC}. We have not yet used in \cref{eq:StackComp2} the values $k = r + m$ with $m \in [r-1]$. For this we note that
\begin{align*}
	\sigma_{r + m} \big(j \big) =  2 \quad &\text{ for } \; j \in  \{r - m + 1\} \cup \{r + 2, r + 3, \dots, 2 r\}.
	\end{align*}
We use this equation to apply Lemma~\ref{lem:convCombEps-i} to \eqref{eq:StackComp2} for $\eps_2$ and $k=r + m$ with $m \in [r-1]$ to obtain 
$$ b_{r - m + 1} + c_{r - m +1} + \sum_{j=2}^{r} \big( b_{r + j} + c_{r + j} \big) = \frac{1}{n}.
$$
We need one more equation to eliminate the right-hand term, so we use the following. Lemma~\ref{lem:convCombEps-i} applied to 
equation~\eqref{eq:Coeff2} for $k=1$ and $l=2$ yields
	\begin{align*}
	\sum_{j=1}^r \big( b_{r + j} + c_{r + j} \big) = \frac{1}{n}.
	\end{align*}
Subtracting this equation from the previous one yields, $b_{r - m + 1} + c_{r - m +1} = b_{r+1} + c_{r + 1}$ for all $m = 1,\ldots,r-1$. Together with the equations \eqref{eq:StackBandC1} and \eqref{eq:StackBandC2} we conclude \cref{eq:StackBandC}. Analogously, \eqref{eq:StackComp3} and \eqref{eq:Coeff3} can be used to obtain
	\cref{eq:StackAandC}. 

To get a contradiction we show that $a_s = b_s = c_s = 0$ for all $s = 2,3,\ldots, r n$. For this, we set $a := \sum_{s} a_s$ and $b := \sum_s b_s$. \cref{eq:sigma-inverse} still applies for $l =1, k = 1$, so Lemma~\ref{lem:convCombEps-i} applied to the coefficient of $\eps_1$ in \eqref{eq:StackComp1}, in \eqref{eq:StackComp2} and in \eqref{eq:StackComp3} respectively for $k=1$ gives
	\begin{align*}
	\sum_{j=1}^{r} c_{j+1} = \frac{1}{n}, \qquad a + \sum_{j=1}^{r-1} c_{j+1} = \frac{1}{n} \qquad \text{ and } \qquad
	b + \sum_{j=1}^{r-1} c_{j+1} = \frac{1}{n}
	\end{align*}
respectively. Subtracting the second equation from the first gives $a=c_{r+1}$, and reasoning analogously for the third yields $a=b=c_{r+1}$. Moreover, \eqref{eq:Coeff2} with $k=r$ and $l=2$ is $\sum_{j=1}^r (b_{j+1} + c_{j+1}) = n^{-1}$. Using the latter together with $b_2 = \ldots = b_r = 0$ and $\sum_{j=1}^{r} c_{j+1} = n^{-1}$ yields $b_{r+1} = 0$ and similarly $a_{r+1} = 0$ via \eqref{eq:Coeff3} with $k=r$ and $l=2$.

Since now also $a_{r + 1} = b_{r +1} = 0$, the equation~\eqref{eq:Coeff1} with $k=r$ and $l=2$ simplifies to $\sum_{j=1}^r c_{j+2} = n^{-1}$. In conjunction with $\sum_{j=1}^{r} c_{j+1} = n^{-1}$ we deduce $c_2 = c_{r+2}$ and hence $b_{r+2} = 0 = a_{r +2} $ by \eqref{eq:StackBandC} and \eqref{eq:StackAandC}. But now \eqref{eq:Coeff1} with $k=r-1$ and $l=2$ is $\sum_{j=1}^r c_{j+3} = n^{-1}$ and together with $\sum_{j=1}^r c_{j+2} = n^{-1}$ we get $c_3 = c_{r+3}$. Continuing inductively we obtain 
	\begin{align*}
	\forall \, j \in [r] \colon \quad c_{j+1} = c_{r + j + 1} \quad \text{ and } \quad a_{r + j +1} = b_{r + j + 1} = 0
	\end{align*}
via \eqref{eq:Coeff1} with $l=2$, $k \in [r]$ and via \eqref{eq:StackBandC}, \eqref{eq:StackAandC}. Then \eqref{eq:Coeff1} with $k=r$ and $l=3$ simplifies to $\sum_{j=1}^r c_{r + j+2} = n^{-1}$ and together with $n^{-1} = \sum_{j=1}^{r} c_{j+1} =  \sum_{j=1}^{r} c_{r + j +1}$ we have $c_{r + 2} = c_{2r + 2}$. Hence, $b_{2 r+2} = 0 = a_{2 r + 2}$ via \eqref{eq:StackBandC} respectively \eqref{eq:StackAandC}. Continuing inductively in the outlined manner with equation~\eqref{eq:Coeff1} for $k \in [r]$, $l=3,\ldots,n$ and with the equations \eqref{eq:StackBandC} and \eqref{eq:StackAandC} we conclude $a_s = b_s = 0$ for all $s=2,3 \ldots, r n$, so $a=b=0$. Finally, \eqref{eq:StackBandC} implies $c_{r + 1} = c_s$ for all $s = 2,\ldots, r n$, but $c_{r + 1} = b = 0$ giving the desired contradiction.
\end{proof}

\section{Padding and rounding for diameter bounds}\label{sec:rounding}

We begin with the proof of \cref{prp:pad-diameter}. We prove it only for $d = 3$, but the proof goes through mutatis mutandis for all $d \geq 1$.

\begin{proof}[Proof of \cref{prp:pad-diameter}]
Recall that $q$ is the $n\times n \times n$ array such that $q_{ijk} = \frac{t}{n} p_{ijk}$ for $i,j,k \in [t]$, $q_{iii} = 1/n$ for $t+1 \leq i \leq n$, and $q_{ijk} = 0$ otherwise. We may split the inputs $x,y,z \in \id_n^\perp$ into 
\begin{align*}x & =\left(x' + \alpha_1 \id_t ,x'' - \frac{t}{n-t} \alpha_1 \id_{n-t}\right),\\
 y&=\left(y' + \alpha_2 \id_t, y'' - \frac{t}{n-t} \alpha_2 \id_{n-t}\right),\\
 z& = \left(z'+ \alpha_3 \id_t ,z''- \frac{t}{n-t} \alpha_3 \id_{n-t}\right)
\end{align*} where $x',y',z' \in \RR^{t}, x'', y'', z'' \in \RR^{n-t}$ each sum to zero; write $w = (x',y',z')$. As $\| (x,y,z)\|_2 \geq \|w\|_2$, it is enough to prove that $\|w\|_2$ is large for any approximate minimizer. By optimizing over $\alpha_i$ and $x'', y'', z''$ for fixed $w$, one computes that the optimum value for $f_q$ for any fixed $w$ is $f_p(w)^{t/n}$. To see this, write
$$ f_q(x,y,z) = \frac{t e^{\alpha_1 + \alpha_2 + \alpha_3}}{n} f_p(w) + \frac{e^{ -\frac{t}{n - t} (\alpha_1 + \alpha_2 + \alpha_3)} }{n} \sum_{i = t+1}^n e^{x''_i + y''_i + z''_i}.$$
First note that for fixed $\alpha_i$'s, the second term is minimized at $x'' = y'' = z'' = 0$ by Jensen's inequality. Furthermore, the value only depends on $\alpha:= \alpha_1 + \alpha_2 + \alpha_3$. With $x'',y'',z'' = 0$, we have
$$ f_q(x,y,z) =g(w, \alpha):= \frac{t e^{\alpha}}{n} f_p(w) + \frac{(n - t) }{n} e^{- \frac{t}{n - t} \alpha}.$$
Taking the derivative in $\alpha$, we see that this is minimized when $ f_p(w) e^{\alpha} = e^{- \frac{t}{n - t} \alpha},$ or $e^{\alpha} = f_p(w)^{-1/(1 + \frac{t}{n-t})} = f_p(w)^{-\frac{n-t}{n}}.$ Plugging this value in proves that the optimum is $f_p(w)^{t/n}$.
By concavity of $x^{t/n},$ provided $f_p(w) \leq 1$ we have 
$$f_p(w)^{t/n} - \capa(p)^{t/n} \geq \frac{1 - \capa(p)^{t/n}}{1 - \capa(p)} (f_p(w) - \capa(p)).$$
The first factor in the second term is the slope of the line from $(\capa(p), \capa(p)^{t/n})$ to $(1,1)$. Thus for any $\eps \leq 1 - \capa(p)$, any $\eps$-approximate minimizer for $f_q$ has norm at least that of some $\big( \frac{1 - \capa(p)}{1 - \capa(p)^{t/n}} \big)\eps$-approximate minimizer for $f_p$.
\end{proof}
\begin{proof}[Proof of \cref{lem:cap-diff}] We use the dual expression: $\log\capa q = - \inf_{\EE_r \omega = 0} D_{KL}(r|| q)$ where $r$ ranges over probability distributions on $\Omega$. In particular, 
$$ \log\capa q \geq - D_{KL}(r|| q) $$ for any distribution $r$ on $\Omega$ with $\EE_r \omega = 0$. Let $r$ be a probability distribution; calculate 
\begin{align*} \log \capa q \geq - D_{KL}(r||q) &= - D_{KL}(r||p) + D_{KL}(r||p) - D_{KL}(r||q)\\
&= - D_{KL}(r||p) + \sum_{\omega \in \Omega} r_\omega \log (r_\omega/p_\omega) - \sum_{\omega \in \Omega} r_\omega \log (r_\omega/q_\omega)\\
&= - D_{KL}(r||p) + \sum_{\omega \in \Omega} r_\omega (\log q_\omega  - \log p_\omega ).
\end{align*}
We lower bound $\log q_\omega  - \log p_\omega  \geq \frac{1}{q_\omega}(q_\omega  - p_\omega )$ by applying the inequality $\log x \leq x - 1$ to $x = p_\omega/q_\omega$. Hence
\begin{align*}
\log \capa q &\geq - D_{KL}(r||p) + \sum_{\omega \in \Omega} r_\omega \frac{1}{q_\omega}(q_\omega  - p_\omega )\\
&\geq  - D_{KL}(r||p) - M_0 \|p - q\|_\infty.
\end{align*}
Allowing $- D_{KL}(r||p)$ to tend to $\log \capa p$ completes the proof.\end{proof}

\begin{proof}[Proof of \cref{lem:rounding}]

Applying \cref{lem:cap-diff} with the roles of $p$ and $q$ switched yields $$\log \capa p \geq \log\capa q - M \|p - q\|_\infty.$$ Exponentiating both sides and applying the inequality $e^{x} \geq 1 + x$ yields
$\capa p \geq (1 - M \|p - q\|_\infty) \capa q.$
Thus
$$ \inf_{x \in B} f_q (x) = \inf_{x \in S} f_q (x)\geq - \sup_{x \in S}  |f_q(x) - f_p(x)| +  \inf_{x \in S} f_p(x). $$
Note that the minimizer for $f_q$ over $B$ lies in the set $S:=B \cap \{x: \; \forall\; \omega,\;q_\omega e^{x \cdot \omega} \leq f_q(0) = \|q\|_1\}.$ For all $x \in S$, we have $e^{x \cdot \omega} \leq \capa q/p_\omega$ for all $\omega \in \Omega$, so 
\begin{align*}
f_q(x) - f_p(x) &\leq \sum_{\omega \in \Omega} |p_\omega - q_\omega| e^{x \cdot \omega}\\
& \leq \sum_{\omega \in \Omega} |p_\omega - q_\omega| \|q\|_1/ q_\omega)\\
& \leq \| p - q\|_1 M \|q\|_1.
\end{align*}
Combining the above inequality with the lower bound for $\capa(p)$, 
\begin{align*} \inf_{x \in B} f_q (x) & \geq - M \|q\|_1\|p - q\|_1 + (1 + \eps)\capa p\\
& \geq (1 +\eps)(1 - M \|p - q\|_\infty) \capa q -M \| p - q\|_1\|q\|_1. \qedhere\end{align*}
\end{proof}

\section*{Acknowledgements}
The authors thank Jason Altschuler, Peter B\"urgisser, Visu Makam, Adam Sawicki and Michael Walter for helpful discussions. Furthermore, the authors thank the anonymous referees for helpful comments and suggestions. We thank Jan Draisma for pointing out repairable mistakes in the proofs of \cref{prp:FreeForGapConstant,prp:WeightGapForDirectPower}.
PR acknowledges funding by the European Research Council (ERC) under the Europeans Horizon 2020 research and innovation programme (grant
agreement no. 787840).

\bibliographystyle{alphaurl}
\addcontentsline{toc}{section}{References}
\bibliography{main}

\end{document}